\documentclass{fundam}

\publyear{25}
\papernumber{4}
\volume{194}
\issue{1}
\theDOI{10.46298/fi.11318}

\versionForARXIV

\usepackage{xspace}   
\usepackage{mathtools}
\usepackage{amsfonts}
\usepackage{bussproofs}
\usepackage{scalerel}
\usepackage{amssymb}
\usepackage{bm}
\usepackage{paralist}
\usepackage{tikz}

\usepackage{hyperref}
\renewcommand{\vec}[1]{\boldsymbol{#1}}

\begin{document}

\title{Tractable and Intractable Entailment Problems 
in Separation Logic with Inductively Defined Predicates}

\author{Mnacho Echenim\thanks{This work has been partially funded by the 
the French National Research Agency ({\tt ANR-21-CE48-0011}).} \and Nicolas Peltier\thanksas{1}\\ 
Univ. Grenoble Alpes, CNRS, Grenoble INP, LIG, 38000 Grenoble, France}

\runninghead{M. Echenim, N. Peltier}{Tractable and Intractable Entailment Problems 
in Separation Logic}

\newcommand{\set}[1]{\left\{#1\right\}}
\newcommand{\sorts}{\mathfrak{S}}
\newcommand{\asort}{{\mathtt s}}
\newcommand{\vars}{\mathcal{V}}
\newcommand{\csts}{\mathcal{C}}
\newcommand{\terms}{\mathcal{T}}
\newcommand{\addr}{\mathtt{loc}}
\newcommand{\true}{\top}
\newcommand{\card}[1]{\mathit{card}(#1)}

\newcommand{\len}[1]{\mathit{len}(#1)}
\newcommand{\bigO}[1]{O(#1)}

\newcommand{\NP}{{\sc Np}\xspace}
\newcommand{\PSPACE}{{\sc PSpace}\xspace}
\newcommand{\EXPTIME}{{\sc ExpTime}\xspace}
\newcommand{\PTIME}{{\sc PTime}\xspace}
\newcommand{\SAT}{{\sc Sat}}

\newcommand{\bigast}{\scaleobj{1.5}{*}}

%
\newcommand{\mn}[1]{\mycomment[me]{#1}}

\newcommand{\subseteqm}{\subseteq_m}

\newcommand{\included}{\sqsubseteq}

\newcommand{\isdef}{\overset{\text{\tiny \it def}}{=}}

\maketitle

\newcommand{\progtype}{\mathtt{P}}
\newcommand{\ls}{{\mathtt{ls}}}
\newcommand{\lsB}{p}
\newcommand{\lsC}{q}
\newcommand{\als}{{\mathtt{als}}}
\newcommand{\tree}{{\mathtt{tree}}}
\newcommand{\treeB}{{\mathtt{tll}}}
\newcommand{\treeC}{{\mathtt{tptr}}}
\newcommand{\dll}{{\mathtt{dll}}}
\newcommand{\dllseg}{{\mathtt{dllseg}}}
\newcommand{\hdeterministic}{$\addr$-deterministic\xspace}
\newcommand{\hDeterministic}{$\addr$-Deterministic\xspace}
\newcommand{\prule}{$\progtype$-rule\xspace}
\newcommand{\prules}{$\progtype$-rules\xspace}
\newcommand{\pRulesname}{$\progtype$-Rules\xspace}
\newcommand{\prulesname}{$\progtype$-rules\xspace}
\newcommand{\prulename}{$\progtype$-rule\xspace}

\vspace*{-2ex}

\begin{abstract}
We establish various complexity results for the entailment problem between
formulas in Separation Logic (SL) with user-defined predicates denoting recursive data structures.
The considered fragments are characterized by syntactic conditions on the inductive rules 
that define the semantics of the predicates. We focus on so-called {\em {\prule}s}, which are 
 similar to  (but simpler than)  
 the so-called {\em bounded treewidth fragment} of SL studied by Iosif et al.~in 2013. 
In particular, for a specific fragment where predicates are defined by so-called \emph{\hdeterministic} inductive rules, we devise a sound and complete 
cyclic proof procedure running  in polynomial time. Several complexity lower bounds are provided, showing that any relaxing of the 
provided conditions makes the problem intractable.
\end{abstract}


\paragraph{ACM Computing Classification System:} --Theory of computation, Logic, Automated reasoning; --Theory of computation, Logic, Separation logic

\section{Introduction}

Separation Logic \cite{IshtiaqOHearn01,Reynolds02} (SL) is widely used in verification to reason on 
programs manipulating pointer-based data structures. 
It forms the basis of several automated static program analyzers such as Smallfoot \cite{DBLP:conf/fmco/BerdineCO05},
Infer \cite{DBLP:conf/nfm/CalcagnoD11} (Facebook) or SLAyer \cite{DBLP:conf/cav/BerdineCI11} (Microsoft Research) and several correctness proofs were carried out by embedding SL in interactive 
theorem provers such as Coq \cite{COQ_REF}, see for instance \cite{DBLP:journals/jfp/JungKJBBD18}.
SL uses a special connective $*$, called {\em separating conjunction}, modeling heap compositions and allowing for concise and natural specifications. 
More precisely, atoms in SL are expressions of the forms 
$x \mapsto (y_1,\dots,y_n)$, where $x,y_1,\dots,y_n$ are variables denoting locations (i.e., memory addresses), asserting that location $x$ is allocated and refers to the tuple (record) $(y_1,\dots,y_n)$.
The special connective $\phi * \psi$ asserts that formulas $\phi$ and $\psi$ hold on disjoint parts of the memory.
Recursive data structures may then be described  by considering predicates associated with inductive rules, such as:
\[
\begin{array}{llllll}
\ls(x,y) & \Leftarrow & x \mapsto y & 
\tree(x) & \Leftarrow & x \mapsto () \\
\ls(x,y) & \Leftarrow & x \mapsto z * \ls(z,y) &
\tree(x) & \Leftarrow & x \mapsto (y,z) * \tree(y) * \tree(z)
\end{array}
\]
where $\ls(x,y)$ denotes a nonempty list segment and $\tree(x)$ denotes a tree.
For the sake of genericity, such rules are not built-in but may be provided by the user. 
Due to the expressive power of such inductive definitions, the input language is usually restricted in this context to so-called {\em symbolic heaps}. These are existentially quantified conjunctions and separating conjunctions of atoms, including inductively defined predicates and equational atoms but dismissing for instance universal quantifications, negations and separating implications.
Many problems in verification require to solve entailment problems between such SL formulas, for instance 
when these formulas denote pre- or post-conditions of programs.
Unfortunately, the entailment problem between symbolic heaps is undecidable in general \cite{DBLP:journals/tocl/MathejaPZ23}, but it is decidable if the considered inductive rules satisfy the so-called {\em PCE conditions} (standing for {\bf P}rogress, {\bf C}onnectivity and {\bf E}stablishment)  \cite{IosifRogalewiczSimacek13}.
However even for the PCE fragment the complexity of the entailment problem is still very high; more precisely, this problem is $2$-\EXPTIME-complete \cite{DBLP:conf/lpar/EchenimIP20,EIP21a,DBLP:conf/lpar/KatelaanZ20}.
Less expressive fragments have thus been considered, for which more efficient algorithms were developed.
In \cite{DBLP:conf/atva/IosifRV14} a strict subclass of PCE entailments is identified with an \EXPTIME complexity based on a reduction to the language inclusion problem for tree automata \cite{tata2007}. 
In \cite{DBLP:journals/fmsd/EneaLSV17}, an algorithm is developed 
 to handle various kinds of (possibly nested) singly-linked lists based on a reduction to the membership problem for tree automata. The complexity of the procedure is dominated by the boolean satisfiability and unsatisfiability tests, that are NP and co-NP complete, respectively. A polynomial proof procedure has been devised for the specific case of singly-linked lists \cite{cook-haase-ouaknine-parkinson-worell11}. 
In \cite{DBLP:conf/concur/ChenSW17}, the tractability result is extended to more expressive fragments, with formulas defined on some unique nonlinear compositional  inductive predicate with distinguished  source, destination, and static parameters. 
The compositional properties satisfied by the considered predicate (as originally introduced in \cite{spen}) ensure that the entailment
problem can be solved efficiently.
Recently \cite{DBLP:conf/fossacs/LeL23} introduced  a polynomial-time cyclic proof system to 
solve entailment problem efficiently, under some conditions on the inductive rules.



In the present paper, we study the complexity of the entailment problem for a specific fragment that is similar to the 
PCE fragment, but simpler.
The fragment inherits most of the conditions given in \cite{IosifRogalewiczSimacek13} and 
admits an additional restriction that is meant to ensure that entailment problems can be solved in a more efficient way\footnote{At the cost, of course, of a loss of expressivity.}:  every predicate is bound to allocate {\em exactly} one of its parameters (forbidding for instance predicates denoting doubly-linked list segments from $x$ to $y$, as both $x$ and $y$ would be allocated). 
This means that the rules do not allow for multiple pointers {\em into} a data structure (whereas multiple pointers {\em out} of the structure are allowed).
We first show that this additional restriction is actually not sufficient to ensure tractability. More precisely, we establish several lower-bound complexity results for the entailment problem
under various additional hypotheses.
Second, we define a new class of inductive definitions for which the entailment problem can be solved in polynomial time, based mainly on the two following additional restrictions:
(i) the arity of the predicates is bounded; 
and (ii) the rules defining the same predicate do not overlap, in a sense that will be formally defined below. Both conditions are rather natural restrictions in the context of programming. 
Indeed, the number of parameters is usually small in this context. Also, data structures are typically defined using a finite set of free constructors, which yields inductive definitions that are trivially non-overlapping.

If Condition (i) is not satisfied, then the complexity is simply exponential.
In contrast with other polynomial-time algorithms, the formulas we consider  may contain several inductive predicates, and these predicates are possibly non-compositional 
(in the sense of \cite{DBLP:journals/fmsd/EneaLSV17}).  
The algorithm for testing entailment is defined as a sequent-like cyclic proof procedure, with standard unfolding and decomposition rules, together with a specific strategy ensuring efficiency and additional syntactic criteria to detect and dismiss non-provable sequents.
Our approach is close to that of \cite{DBLP:conf/fossacs/LeL23}, in the sense that the two procedures
use cyclic proof procedures with non-disjunctive consequents. 
However, the conditions on the rules
are completely different: our definition allows for multiple inductive rules with mutually recursive definitions, yielding richer recursive data structures. 
On the other hand, the SHLIDe rules in \cite{DBLP:conf/fossacs/LeL23} support ordering and equality relations on non-addressable values, whereas the predicate we consider are purely spatial. 
Moreover, the base cases of the rules in \cite{DBLP:conf/fossacs/LeL23} correspond to empty heaps, which are forbidden in our approach.

\newcommand{\simple}{simple\xspace}
\newcommand{\mypred}{\mathtt{P}}
\newcommand{\acstL}{\mathtt{list}}
\newcommand{\acstT}{\mathtt{tree}}
\newcommand{\acstP}{\mathtt{loop}}
\newcommand{\swedge}{\curlywedge}
\newcommand{\iseq}{\approx}


To provide some intuition on what can and cannot be expressed in the  fragment we consider, we provide some examples (formal definitions will be given later); consider the predicate $\mypred$ defined by the following rules, which encode a combination of lists and trees, possibly looping on an initial element $y$, and ending with an empty tuple:
\[
\begin{tabular}{lll}
$\mypred(x,y)$ & $\Leftarrow$ & $x \mapsto (\acstL,u) * \mypred(u,y)$ \\
$\mypred(x,y)$ & $\Leftarrow$ & $x \mapsto (\acstT,u_1,u_2) * \mypred(u_1,y) * \mypred(u_2,y)$ \\
$\mypred(x,y)$ & $\Leftarrow$ & $x \mapsto (\acstP,y)$ \\
$\mypred(x,y)$ & $\Leftarrow$ & $x \mapsto ()$ \\
\end{tabular}
\]
All variables are implicitly universally quantified at the root level in every rule\footnote{Alternatively, the variables $u,u_1,u_2$, which do not occur in the left-hand side of the rules, may be viewed as being quantified existentially on the right-hand side of $\Leftarrow$.}. The constants $\acstL,\acstT$ and $\acstP$ may be viewed as constructors for the data structure. This predicate does not fall in the scope of the fragment considered in \cite{DBLP:conf/fossacs/LeL23} since it involves a definition with several inductive rules, but it falls in the scope of the fragment considered in the present paper. Our restrictions require that the definition must be deterministic, in the sense that there can be no overlap between distinct rules. This is the case here, as the tuples $(\acstL,u)$, $(\acstT,u_1,u_2)$ and $(\acstP,y)$ are pairwise distinct (not unifiable), but replacing for instance the constant $\acstP$ by $\acstL$ in the third rule would not be possible, as the resulting rule would overlap with the first one (both rules could allocate the same heap cell).
As explained above, a key limitation of our fragment (compared to that of \cite{IosifRogalewiczSimacek13}) is that it does not allow predicates allocating several parameters, such as the following predicate $\dllseg(x,y,z,u)$ defining a 
doubly-linked list segment from $x$ to $z$ (each cell points to a pair containing the previous and next element and $y$ and $u$ denote the previous and next element in the list, respectively):
\[
\begin{tabular}{lll}
$\dllseg(x,y,z,u)$ & $\Leftarrow$ & $(x \mapsto (y,x') * \dllseg(x',x,z,u)) \wedge x \not \iseq z$ \\
$\dllseg(x,y,z,u)$ & $\Leftarrow$ & $x \mapsto (y,u) \wedge x \iseq z$  \end{tabular}
\]
Other definitions of $\dllseg$ are possible, but none would fit in with our restrictions: in every case, both $x$ (the beginning of the list) and $z$ (its end) must be eventually allocated, which is not permitted in the fragment we consider.
On the other hand, the following predicate, defining a doubly-linked list, ending with $()$, can be defined ($y$ denotes the previous element in the list):
\[
\begin{tabular}{lll}
$\dll(x,y)$ & $\Leftarrow$ & $x \mapsto (y,z) * \dll(z,x)$ \\
$\dll(x,y)$ & $\Leftarrow$ & $x \mapsto ()$  \\
\end{tabular}
\]
The rest of the paper is organised as follows.
In Section~\ref{sect:prel}, the syntax and semantics of the logic are defined.
The definitions are mostly standard, although we consider a multisorted framework, with a special sort $\addr$ denoting memory locations and additional 
sorts for data or constructors.
We then introduce a class of  inductive definitions called {\em \prules}.
In Section \ref{sect:limits}, various lower bounds on the complexity of the entailment problem for SL formulas with \prulesname are established which allow one to motivate additional restrictions on the inductive rules.
These lower bounds show that all the restrictions are necessary to ensure that the entailment problem is tractable.
This leads to the definition of the notion of a {\em \hdeterministic} set of rules, that is a subset of \prulesname for which entailment can be decided in polynomial time. 
The proof procedure is defined in Section~\ref{sect:proof}. 
For the sake of readability and generality we first define generic 
inference rules and establish their correctness, before introducing a specific strategy to further restrict the application of the rules that is both complete
and efficient.
Section~\ref{sect:prop} contains all soundness, completeness and complexity results and Section~\ref{sect:conc} concludes the paper.

\section{Definitions}

\label{sect:prel}
\subsection{Syntax}

\label{sect:syntax}

\newcommand{\size}[1]{|#1|}

We use a multisorted framework, which is essentially useful to distinguish locations from data.
Let 
$\sorts$ be a set of {\em sorts}, containing a special sort $\addr$, denoting memory locations. Let $\vars_{\asort \in \sorts}$ be a family of countably infinite
disjoint sets of {\em variables of sort $\asort$}, with $\vars \isdef \bigcup_{\asort\in \sorts} \vars_\asort$. 
Let $\csts_{\asort \in \sorts}$ be a family of disjoint sets of {\em constant symbols of sort $\asort$}, also disjoint from $\vars$, with $\csts \isdef \bigcup_{\asort\in \sorts} \csts_\asort$.
The set of {\em terms of sort $\asort$} is $\terms_\asort \isdef \vars_\asort \cup \csts_\asort$, and we let $\terms \isdef \bigcup_{\asort\in \sorts}\terms_\asort$.
Constants are especially useful in our framework to denote constructors in data structures. To simplify technicalities, we assume that there is no constant of sort $\addr$, i.e., 
$\csts_{\addr} = \emptyset$.

\newcommand{\pure}{{\cal F}_P}
\newcommand{\spatial}{{\cal F}_S}
\newcommand{\sheaps}{{\cal F}_H}

\newcommand{\apform}{\xi}

An equation (resp.\ a disequation) is an expression of the form 
$t \iseq s$ (resp.\ $t \not \iseq s$) where $t,s \in \terms_\asort$ for some $\asort\in \sorts$.
The set of {\em pure formulas} $\pure$ is the set of formulas of the form
$e_1 \wedge \dots \wedge e_n$, where every expression $e_i$ is either an equation or a disequation. Such formulas are considered modulo contraction, e.g., a pure formula  $\apform \wedge \apform$ is considered identical to $\apform$, 
and also modulo associativity and commutativity of conjunction.
We denote by $\bot$ (false) any  formula of the form $t\not \iseq t$. 
If $n = 0$, then $\bigwedge_{i=1}^n e_i$ may be denoted by $\true$ 
(true).
If $(t_1,\dots,t_n)$ and $(s_1,\dots,s_m)$ are vectors of terms, then 
$(t_1,\dots,t_n) \iseq (s_1,\dots,s_m)$ denotes the formula $\bot$ if either $n \not = m$ or $n = m$ and there exists $i \in \{ 1,\dots,n \}$ such that $s_i$ and $t_i$ are of different sorts;
and denotes $\bigwedge_{i=1}^n t_i \iseq s_i$ otherwise.

\newcommand{\preds}{\mathcal{P}}
\newcommand{\emp}{\mathit{emp}}

\newcommand{\asform}{\phi}
\newcommand{\asformB}{\psi}
\newcommand{\asformC}{\eta}

\newcommand{\apformB}{\zeta}
\newcommand{\apformC}{\chi}

\newcommand{\asheap}{\lambda}
\newcommand{\asheapB}{\gamma}

\newcommand{\aexp}{\beta}
\newcommand{\anatom}{\alpha}
\newcommand{\anatomB}{\beta}
\newcommand{\rootof}[1]{\mathit{root}(#1)}
 \newcommand{\dom}[1]{\mathit{dom}(#1)}
\newcommand{\codom}[1]{\mathit{codom}(#1)}

Let $\preds$ be a set of {\em predicate symbols}. Each symbol in $\preds$ is associated with a unique {\em profile} of the form $(\asort_1,\dots,\asort_n)$ with $n \geq 1$, $\asort_1 = \addr$ and $\asort_i \in \sorts$, for all $i \in \{ 2,\dots,n\}$. 
We write $p: \asort_1,\dots,\asort_n$ to denote a symbol with  profile 
$(\asort_1,\dots,\asort_n)$ and 
we write $p: \asort_1,\dots,\asort_n \in \preds$ to state that $p$ is a predicate symbol of profile $\asort_1,\dots,\asort_n$ in $\preds$.
A {\em spatial atom} $\anatom$ is either  a {\em points-to atom} 
$x \mapsto (t_1,\dots,t_n)$ with $x \in \vars_{\addr}$ and $t_1,\dots,t_n \in \terms$, 
or a {\em predicate atom} of the form $p(x,t_1,\dots,t_n)$, where 
$p$ is a predicate of profile $\addr,\asort_1,\dots,\asort_n$ in $\preds$, the term $x$ is a variable in $\vars_{\addr}$ 
and $t_i \in \terms_{\asort_i}$  for all $i \in \{ 1,\dots, n\}$.
In both cases, the variable $x$ is called the {\em root} of $\anatom$ and is denoted by $\rootof{\anatom}$.

The set of {\em spatial formulas} $\spatial$ is the set of formulas of the form
$\aexp_1 * \dots * \aexp_n$, where every expression $\aexp_i$ is a spatial atom.
If $n = 0$ then $\aexp_1 * \dots * \aexp_n$ is denoted by $\emp$.
The number $n$ of occurrences of spatial atoms in a spatial formula $\asform = \aexp_1 * \dots * \aexp_n$ is denoted by 
$\len{\asform}$.
We write $\asform \included \asformB$ if 
$\asformB$ is of the form $\asform * \asform'$, modulo associativity and commutativity of $*$. 
The set of (non quantified) {\em symbolic heaps} $\sheaps$ is the set of expressions of the form 
$\asform \swedge \apform$, where $\asform \in \spatial$ and $\apform \in \pure$.  
Note that for clarity we use $\swedge$ to denote conjunctions between spatial and pure formulas and $\wedge$ to denote conjunctions occurring within pure formulas.
If $\apform = \true$, then $\asform \swedge \apform$ may be written $\asform$ (i.e., any 
spatial formula may be viewed as a symbolic heap). 
For any formula $\asheap$, $\size{\asheap}$ denotes the size of $\asheap$ (which is defined inductively as usual).
Note that $\top$ is not a symbolic heap (but $\emp \swedge \top$ is a symbolic heap). 

\newcommand{\repl}[3]{#1{\{ #2 \leftarrow #3\}}}
\newcommand{\replall}[4]{#1\{ #2 \leftarrow #3 \mid #4 \}}
\newcommand{\id}{\mathit{id}}
We denote by $\vars(\aexp)$ (resp.\ $\vars_{\asort}(\aexp)$) the set of variables (resp.\ of variables of sort $\asort$) occurring in a variable or formula
$\aexp$. 
A {\em substitution} is a sort-preserving total mapping from $\vars$ to $\terms$.
 We denote by $\dom{\sigma}$ the set of variables such that $\sigma(x) \not = x$, and by 
$\codom{x}$ the set $\{ \sigma(x) \mid x \in \dom{\sigma} \}$.
The substitution $\sigma$ such that $\sigma(x_i) = y_i$ for all $i = 1,\dots,n$ and $\dom{\sigma} \subseteq \{ x_1,\dots,x_n \}$ is denoted by $\replall{}{x_i}{y_i}{i = 1,\dots,n}$.
For any expression $\aexp$ and substitution $\sigma$, we denote by $\aexp\sigma$ the expression 
obtained from $\aexp$ by replacing every variable $x$ by $\sigma(x)$.
A {\em unifier} of two expressions or tuples of expressions $\aexp$ and $\aexp'$ is a substitution $\sigma$ such that 
$\aexp\sigma = \aexp'\sigma$.

\newcommand{\unfold}{\Leftarrow_{\asid}}
\newcommand{\asid}{\mathfrak{R}}
\newcommand{\unfoldall}{\rightsquigarrow_{\asid}}

An {\em inductive rule} is an expression of the form 
$p(x_1,\dots,x_n) \Leftarrow \asheap$, where 
$p: \asort_1,\dots,\asort_n \in \preds$, $x_1,\dots,x_n$ are
pairwise distinct variables of sorts $\asort_1,\dots,\asort_n$ respectively and $\asheap \in \sheaps$.
The set of variables in $\vars(\asheap) \setminus \{ x_1,\dots,x_n \}$
are the {\em existential variables} of the rule.
Let $\asid$ be a set of inductive rules. 
We write $p(t_1,\dots,t_n) \unfold \asheap$ iff
$\asid$ contains (up to a renaming,
and modulo AC) a rule of the form $p(y_1,\dots,y_n) \Leftarrow \asheapB$
and $\asheap = \replall{\asheapB}{y_i}{t_i}{i = 1,\dots,n}$.
We assume by renaming that $\asheapB$ contains no variable in $\set{t_1,\dots,t_n}$.
We write $p(t_1,\dots,t_n) \unfoldall E$ if $E$ is the set of symbolic heaps $\asheap$ such that 
$p(t_1,\dots,t_n) \unfold \asheap$. Note that if $\asid$ is finite then $E$ is finite  up to a renaming of variables not occurring in $\{ t_1,\dots,t_n\}$.
Note also that the considered logic does not allow for negations (hence entailment is not reducible to satisfiability) or separating implications, as this would make satisfiability undecidable (see for instance \cite{DBLP:journals/tocl/MathejaPZ23}).



The symbol $\subseteqm$ denotes the inclusion relation between multisets. 
With a slight abuse of notations, we will sometimes identify sequences with sets when the order and number of repetitions is not important, for instance we may write $\vec{x} \subseteq \vec{y}$ to state that every element of $\vec{x}$ occurs in $\vec{y}$.

In the present paper, we shall consider  entailment problems between symbolic heaps.

\subsection{Semantics}

\newcommand{\astore}{\mathfrak{s}}
\newcommand{\aheap}{\mathfrak{h}}
\newcommand{\aheapB}{\widehat{\aheap}}
\newcommand{\univ}{\mathfrak{U}}
 \newcommand{\loc}{{\ell}}
 
 \newcommand{\locs}[1]{\mathit{ref}(#1)}

\newcommand{\union}{\uplus}
\newcommand{\connect}[1]{\rightarrow_{#1}}

\newcommand{\valueof}[1]{\dot{#1}}

We assume for technical convenience that formulas are interpreted over a fixed universe and that constants are interpreted as pairwise distinct elements.
Let $\univ_{\asort\in \sorts}$ be pairwise disjoint 
countably infinite sets and let $\univ \isdef \bigcup_{\asort \in \sorts} \univ_\asort$.
We assume that an injective function is given, mapping every constant $c \in \csts_\asort$ to 
an element of $\univ_{\asort}$, denoted by $\valueof{c}$.

A {\em heap} is a partial finite function 
from $\univ_\addr$ to 
$\univ^*$, 
where $\univ^*$ denotes as usual the set of finite sequences of elements of $\univ$.
An element $\loc \in \univ_{\addr}$ is {\em allocated} 
in a heap $\aheap$ 
if  $\loc \in \dom{\aheap}$.
Two heaps $\aheap$ and $\aheap'$ are {\em disjoint} if $\dom{\aheap} \cap \dom{\aheap'} = \emptyset$, in which case
$\aheap \union \aheap'$ denotes their disjoint union.  We write $\aheap \subseteq \aheap'$ if there is a heap $\aheap''$ such that $\aheap' = \aheap \union \aheap''$.
 For every heap $\aheap$, we denote by $\locs{\aheap}$ the set of elements $\ell \in \univ_{\addr}$ such that there exists $\ell_0,\dots,\ell_n$ with $\ell_0 \in \dom{\aheap}$, $\aheap(\ell_0) = (\ell_1,\dots,\ell_n)$ and $\ell = \ell_i$ for some $i = 0,\dots,n$.
We write $\ell \connect{\aheap} \ell'$ if $(\ell,\ell') \in \univ_{\addr}^2$,  $\ell \in \dom{\aheap}$, $\aheap(\ell) = (\ell_1,\dots,\ell_n)$ 
and $\ell' = \ell_i$, for some $i = 1,\dots,n$.

\newcommand{\extension}[2]{associate of $#1$ w.r.t.\ $#2$}
\newcommand{\namedextension}[3]{associate $#1$ of $#2$ w.r.t.\ $#3$}

\begin{proposition}
\label{prop:subheap_connect}
Let $\aheap, \aheap'$ be two heaps such that $\aheap \subseteq \aheap'$. For all $\ell,\ell' \in \univ_{\addr}$, if $\ell \connect{\aheap}^* \ell'$ then $\ell \connect{\aheap'}^* \ell'$.
\end{proposition}
\begin{proof}
By definition of the relation $\connect{\aheap}$ we have $\connect{\aheap} \subseteq \connect{\aheap'}$, thus
$\connect{\aheap}^* \subseteq \connect{\aheap'}^*$
\end{proof}

A {\em store} $\astore$ is a total function mapping every term in $\terms_\asort$ 
to an element in $\univ_\asort$ such that $\astore(c) = \valueof{c}$, for all $c \in \csts$ (note that this entails that $\astore$ is injective on $\csts$).
A store $\astore$ is {\em injective on a multiset of variables} $V$ if $\{ x,y \} \subseteqm V \implies \astore(x) \not = \astore(y)$. 
When a store is injective on a multiset of variables, this entails that the latter is a set, i.e., that contains at most one occurrence of each variable.
For any $V \subseteq \vars$, and for any store $\astore$, a store $\astore'$ is an {\em \extension{\astore}{V}} if $\astore(x) = \astore'(x)$ holds for all 
$x \not \in V$.

\begin{definition}
An {\em SL-structure} is a pair $(\astore,\aheap)$, where 
$\astore$ is a store and
$\aheap$ is a heap. 
\end{definition}

The satisfiability relation on SL-formulas is defined inductively as follows:

\newcommand{\modelsr}{\models_{\asid}}
\newcommand{\rmodel}{$\asid$-model\xspace}

\begin{definition}
\label{def:rmodel}
An SL-structure $(\astore,\aheap)$ {\em validates a formula 
	(pure formula, spatial formula, or symbolic heap) $\asheap$ modulo a set of inductive rules $\asid$}, written $(\astore,\aheap) \models_{\asid} \asheap$, if one of the following conditions holds:
\begin{itemize}
	\item{$\asheap = \emp$ and $\aheap = \emptyset$;}
\item{$\asheap = (t \iseq s)$ and $\astore(t) = \astore(s)$;}
\item{$\asheap = (t \not\iseq s)$ and $\astore(t) \not = \astore(s)$;}
\item{$\asheap = x \mapsto (t_1,\dots,t_k)$ and $\aheap = \{ (\astore(x),\astore(t_1),\dots,\astore(t_k))\}$;}
\item{either $\asheap = \asheap_1 \wedge \asheap_2$ or $\asheap = \asheap_1 \swedge \asheap_2$, and 
	$(\astore,\aheap) \models_{\asid} \asheap_i$ for all $i = 1,2$;}
\item{$\asheap = \asheap_1 * \asheap_2$ and there exist disjoint heaps $\aheap_1$ and $\aheap_2$ such that
$\aheap = \aheap_1 \union \aheap_2$ and 
$(\astore,\aheap_i) \models_{\asid} \asheap_i$ for all $i = 1,2$;}
\item{$\asheap = p(t_1,\dots,t_n)$ and
$p(t_1,\dots,t_n) \unfold \asheapB$, where
there exists an  \namedextension{\astore'}{\astore}{\text{the set $\vars(\asheapB) \setminus \vars(\asheap)$}}  such that $(\astore',\aheap) \models_{\asid} \asheapB$.}
\end{itemize}
An {\em \rmodel} of  a formula $\asheap$ is a structure 
$(\astore,\aheap)$ such that $(\astore,\aheap) \modelsr \asheap$.
A formula is {\em satisfiable} (w.r.t.\ $\asid$) if it admits at least one 
\rmodel.
\end{definition}


\begin{remark}
 Note that a formula $x \mapsto (t_1,\dots,t_k)$ asserts  not only that $x$ refers to $(t_1,\dots,t_k)$ but also that $x$ is the only allocated location. This fits  with usual definitions (see, e.g., \cite{IshtiaqOHearn01}). The assertions are meant to describe elementary heaps, which can be combined afterwards using the connective $*$. Simply asserting that 
$x$ refers to $(t_1,\dots,t_k)$  could be done in full SL using the following formula:
$x \mapsto (t_1,\dots,t_k) * \top$, but such a formula is not a symbolic heap and is thus outside of the fragment we consider in the present paper.
\end{remark}

We write $(\astore,\aheap) \models \asheap$ instead of $(\astore,\aheap) \modelsr \asheap$ if $\asheap$ contains no predicate symbol, since the relation is independent of $\asid$ in this case. Similarly, if $\asheap$ is a pure formula then the relation does not depend on the heap, thus we may simply write $\astore \models \asheap$. 


\begin{remark}
Note that there is no symbolic heap that is true in every structure. For instance $\emp \swedge \top$ is true only in structures with an empty heap.
\end{remark}



\begin{proposition}
	\label{prop:subst_store}
	Let $(\astore,\aheap)$ be a structure and let $\sigma$ be a substitution.
	If $(\astore,\aheap) \modelsr \asheap\sigma$ then we have
	${(\astore \circ \sigma,\aheap) \modelsr \asheap}$.
\end{proposition}
\begin{proof}
	By an immediate induction on the satisfiability relation.
\end{proof}

In the present paper, we shall consider inductive rules of a particular form, defined below.


\begin{definition}
\label{def:simple}
An inductive rule is  a {\em \prulename} 
if it is of the form \[p(x_1,\dots,x_n) \Leftarrow x_1 \mapsto (y_1,\dots,y_k) * q_1(z_1,\vec{u}_1) * \dots 
q_m(z_m,\vec{u}_m) \swedge \apform\] possibly with $m = 0$, where:
\begin{enumerate}
\item{
$\apform$ is a conjunction of disequations of the form $u \not \iseq v$, where  $u \in \{ x_1,\dots,x_n, y_1,\dots,y_k \}$  and
 $v  \in \{ y_1,\dots,y_k\} \setminus \{ x_1,\dots,x_n \}$;
 }
 \item{
 $\{ z_1,\dots,z_m \} = (\{ y_1,\dots,y_k\}  \setminus \{ x_1,\dots,x_n \}) \cap \vars_\addr$, and $z_1,\dots,z_m$ are pairwise distinct;\label{it:prog2}}
 \item{All the elements of $\vec{u}_i$ occur in $\{ x_1,\dots,x_n \} \cup \{ y_1,\dots,y_k \} \cup \csts$.} 
 \end{enumerate}
The predicate symbol $p$ is called the {\em head} of the rule.
\end{definition}

\newcommand{\nil}{\mathtt{nil}}


\begin{example}
\label{ex:als}
The rules associated with $\ls$ and $\tree$ in the introduction are \prules, as well as  
the following rules.
Intuitively $\als(x,y) \swedge x \not \iseq y$ denotes an acyclic list 
($\als(x,y)$ is thus ``quasi-acyclic'', in the sense that it may loop only on the first element). Note that the constraint $x\not \iseq y$ cannot be added to the right-hand side of the rules because the obtained rule would not be a \prule, hence it must be added in the formula.
The atom $\dll(x,y)$ denotes a doubly-linked list starting at $x$, with the convention that each element of the list points to a pair containing the previous and next elements. The parameter $y$ denotes the element before $x$ (if any) and the last element points to the empty tuple $()$.
The structures validating $\dll(x,y)$ are of the form (see Figure \ref{fig:dll})
$(\astore_n,\aheap_n)$ with $n \geq 0$, $\dom{\aheap} = \{ \ell_1,\dots,\ell_{n+1} \}$, 
$\aheap(\ell_i) = (\ell_{i-1},\ell_{i+1})$ for all 
$i \in \{ 1,\dots,n\}$, $\aheap(\ell_{n+1}) = ()$, $\astore(x) = \ell_1$ and $\astore(y) = \ell_0$. Locations $\ell_1,\dots,\ell_{n+1}$ must be pairwise distinct and $\ell_0$ is arbitrary. The atom $\treeC(x,y,z)$ denotes a binary tree in which every node refers to its two successors and to its parent and sibling nodes.  The parameters $y$ and $z$ denote the sibling and parent nodes, respectively.
Leaves point to $()$. See Figure \ref{fig:tree} for an example with $5$ allocated locations.
\[
\begin{array}{llll}
\als(x,y) & \Leftarrow & (x \mapsto (z) * \als(z,y)) \swedge y \not \iseq z & \text{\tt \% (quasi-)acyclic list} \\
\als(x,y) & \Leftarrow & x \mapsto (y)  \\
\treeB(x,y) & \Leftarrow & x \mapsto (y,z) * \tree(z) & \text{\tt \% binary trees with} \\
\treeB(x,y) & \Leftarrow & (x \mapsto (z,u) * \treeB(z,y) * \tree(u))  & \text{\tt \% leftmost leaf $y$} \\
& & \qquad \swedge  (y \not \iseq z) & \\
\dll(x,y) & \Leftarrow & x \mapsto (y,z) * \dll(z,x) & \text{\tt \% doubly-linked lists} \\
\dll(x,y) & \Leftarrow & x \mapsto ()  \\
\treeC(x,y,z) & \Leftarrow & x \mapsto (u,v,y,z) * \treeC(u,v,x) & \text{\tt \% binary trees with} \\
& & \qquad * \treeC(v,u,x) & \text{\tt \% pointers to brother} \\
\treeC(x,y,z) & \Leftarrow & x \mapsto () & \text{\tt \% and parent nodes} \\
\end{array}
\]
The following rules are not \prules (if all variables are of sort $\addr$): 
\[
\begin{array}{llll}
\lsB(x) & \Leftarrow & x \mapsto (z) & \text{Condition 2 violated} \\
\lsB(x) & \Leftarrow & \ls(x,z) * \lsB(z) & \text{No points-to atom} \\
\lsC(x,y) & \Leftarrow & x \mapsto (z) \swedge y \iseq z & \text{\text{Condition 2 violated}} \\
\lsC(x,y) & \Leftarrow & \ls(x,y) & \text{No points-to atom} \\
\als(x,y) & \Leftarrow & (x \mapsto (z) * \als(z,y)) \swedge x \not \iseq y & \text{Condition $1$ violated} \\
\als(x,y) & \Leftarrow & x \mapsto (y) \swedge x \not \iseq y &  \text{Condition $1$ violated}  \\
\end{array}
\]


\begin{figure}
\begin{center}
\begin{tikzpicture}[
    node distance = 2cm and 1.5cm,
    dllnode/.style = {draw, rectangle, minimum width=1cm, minimum height=0.6cm, text centered},
    arrow/.style = {->,>=stealth, thick},
]

\node (prev) {$y$};
\node[dllnode, right of= prev] (first) {$\ell_1$};
\node[dllnode, right of= first] (second) {$\ell_2$};
\node[dllnode, right of= second] (dots) {$\dots$};
\node[dllnode, right of= dots] (last) {$\ell_n$};
\node[dllnode, right of= last] (next) {$\ell_{n+1}$};

\draw[arrow] (first.east) -- (second.west);
\draw[arrow] (second.east) -- (dots.west);
\draw[arrow] (dots.east) -- (last.west);
\draw[arrow] (last.east) -- (next.west);

\draw[arrow] (second.west) -- (first.east);
\draw[arrow] (dots.west) -- (second.east);
\draw[arrow] (last.west) -- (dots.east);
\draw[arrow] (first.west) -- (prev.east);

\end{tikzpicture}
\end{center}
\caption{Doubly-linked list ending with $()$ ($\ell_{n+1}$ is allocated but has no successor) \label{fig:dll} }
\end{figure}

\begin{figure}

\begin{center}
\begin{tikzpicture}[
    node distance = 2cm and 2cm,
    treenode/.style = {draw, rectangle, minimum width=1cm, minimum height=0.6cm, text centered},
    arrow/.style = {->,>=stealth, thick},
]

\node[treenode] (root) {$x$};
\node[above of= root] (extra) {$z$};
\node[left of= root] (extra2) {$y$};
\node[treenode, below left of = root] (leftchild) {$\ell_1$};
\node[treenode, below right of= root] (rightchild) {$\ell_2$};
\node[treenode, below left of= leftchild] (leftgrandchild) {$\ell_3$};
\node[treenode, below right of= leftchild] (rightgrandchild) {$\ell_4$};

\draw[arrow] (root) -- (leftchild);
\draw[arrow] (root) -- (rightchild);
\draw[arrow] (leftchild) -- (leftgrandchild);
\draw[arrow] (leftchild) -- (rightgrandchild);
\draw[arrow]  (leftchild) -- (root);
\draw[arrow] (root) -- (extra);
\draw[arrow] (root) -- (extra2);
\draw[arrow] (leftchild) -- (rightchild);
\draw[arrow] (root) -- (extra);
\draw[arrow] (root) -- (extra2);

\end{tikzpicture}
\end{center}

\caption{Binary tree with pointers to parent and sibling ($\ell_2,\ell_3,\ell_4$ are allocated but have no successor) \label{fig:tree}}
\end{figure}

\end{example}

\begin{remark}
As evidenced by the rules in Example \ref{ex:als}, the tuple $()$ is frequently used as a base case, to end a data structure. This departs from standard conventions in SL in which a non-allocated constant $\nil$ is frequently used instead. We avoid considering constants of sort $\addr$ in our framework because this would complicate definitions: one would have to keep track of allocated and non-allocated constants and/or to add syntactic conditions on the formulas and rules to ensure that such constants are never allocated. Note that this convention introduces an asymmetry in the structure. For instance, in doubly-linked lists, the ``next'' field of the last element points to $()$ while the ``previous'' field of the first element contains an arbitrary location $y$. To achieve symmetry in the structure, one could introduce an atom $y \mapsto ()$. For example, the formula $y \mapsto () * \dll(x,y)$ represents a doubly-linked list ending with $()$ in both directions.
\end{remark}

Note that \prulesname are progressing and connected (in the sense of \cite{IosifRogalewiczSimacek13}): every rule allocates exactly 
one location --the first parameter of the predicate-- and 
the first parameter of every predicate in the body of the rule occurs in the right-hand side of the (necessarily unique) points-to atom of the rule.
They are not necessarily established (again in the sense of\cite{IosifRogalewiczSimacek13}) as non-allocated existential variables are allowed provided they are not of sort $\addr$.
\begin{example}
The following (non established, in the sense of \cite{IosifRogalewiczSimacek13}) rules, denoting list segments with unallocated elements are \prules iff $u \not \in \vars_{\addr}$:
\[ \begin{array}{lllll}
\ls(x,y) & \Leftarrow x \mapsto (u,y) \quad &
\quad \ls(x,y) & \Leftarrow x \mapsto (u,z) * \ls(z,y) 
\end{array}
\]
The heap of any model of $\ls(x,y)$ is of the form $\{ \ell_i \mapsto (u_i,\ell_{i+1}) \mid i \in \{ 1,\dots,n\} \}$, where $u_1,\dots,u_n$ denote arbitrary elements (of a sort distinct from $\addr$).
\end{example}
\pRulesname containing no variable of a sort distinct from $\addr$ are established.
\pRulesname also differ from PCE rules  in that every predicate
allocates {\em exactly} one of its parameters, namely the first one (the other allocated locations are associated with existential variables). 
In other words, there may be only one ``entry point'' to the structure allocated by a predicate, namely its root. For instance the rule $p(x,y) \Leftarrow x \mapsto (y) * q(y)$ (along with another rule for symbol $q$, e.g., $q(y) \Leftarrow y\mapsto ()$) is PCE but it is not a \prulename, whereas $p(x) \Leftarrow x \mapsto (y) * q(y)$ is a \prule.
 Such a restriction makes the entailment problem easier to solve because it rules out data structures that can be constructed in different orders (for instance doubly-linked lists with a reference to the end of the list). 


\newcommand{\maxar}[1]{\mathit{ar}_{\mathit{max}}(#1)}
\newcommand{\maxk}[1]{\mathit{record}_{\mathit{max}}(#1)}
\newcommand{\maxB}[1]{\mathit{E}_{\mathit{max}}(#1)}
\newcommand{\maxr}[1]{\mathit{width}(#1)}

We introduce some useful notations and measures on sets of \prulesname.
For every set of \prulesname $\asid$, 
we denote by $\preds(\asid)$ the set of predicate symbols
occurring in a rule in $\asid$.
We define: \[\maxar{\asid} = \max \{ n \mid p: \asort_1,\dots,\asort_n \in \preds(\asid) \}\] and \[\maxk{\asid} = \max \{ k \mid p(x_1,\dots,x_n) \Leftarrow (x \mapsto (t_1,\dots,t_k) * \asform) \swedge \apform \in \asid \}\]
The numbers $\maxar{\asid}$ and $\maxk{\asid}$ respectively denote the maximum arity of the predicate symbols in 
$\asid$ and 
 the maximum number of record fields in a points-to atom occurring in $\asid$.
The {\em width} of $\asid$ is defined as follows: $\maxr{\asid} \isdef \max(\maxar{\asid},\maxk{\asid})$.

We make two additional assumptions about the considered set of rules: we assume that every predicate is productive (Assumption \ref{assume:productive}) and that no parameter is useless (Assumption \ref{assume:no_useless_variable}).
More precisely, the set of {\em productive} predicate symbols is inductively defined as follows:
 $p\in \preds$ is productive w.r.t.\ a set of inductive rules $\asid$ if 
 $\asid$ contains a rule $p(\vec{x}) \Leftarrow \asheap$ such that 
 all the predicate symbols occurring in $\asheap$ are productive. 
 In particular, a rule with no predicate in its right-hand side is always productive (base case). Productive rules can easily be computed using a straightforward least fixpoint algorithm.
 

\begin{example}
Let $\asid = \{ p(x) \Leftarrow q(x), q(x) \Leftarrow p(x), r(x) \Leftarrow x \mapsto (y) * p(y) \}$.
The predicates $p,q,r$  are not productive.
\end{example}

It is easy to check that every formula containing at least one non-productive predicate symbol 
is unsatisfiable. 
Indeed, if a predicate symbol $p$ is non-productive, then an atom $p(\vec{x})$ cannot be unfolded into a formula containing no predicate symbol.
 This justifies the following:
 
\begin{assumption}
\label{assume:productive}
For all sets of \prulesname $\asid$, we assume that  all the predicate symbols are productive w.r.t.\ $\asid$.
\end{assumption}

\newcommand{\outparam}[1]{\mathit{out}_{\asid}(#1)}

For all predicates $p: \asort_1,\dots,\asort_n$, the set $\outparam{p}$ denotes the least set of indices $i$ in $\{ 1,\dots,n\}$ such that $\asort_i = \addr$ and there exists a rule $p(x_1,\dots,x_n) \Leftarrow \asheap$ in $\asid$ such that $\asheap$ contains either a points-to atom $x_1 \mapsto (t_1,\dots,t_k)$ where  $x_i \in \{ t_1,\dots,t_k\}$ or 
 a predicate atom $q(t_1,\dots,t_m)$ with $t_j = x_i$, for some $j \in \outparam{q}$.
 Intuitively, $\outparam{p}$ denote the set of ``out-going'' nodes of the structures corresponding to $p$, i.e., the set of parameters corresponding to locations that can be referred to but not necessarily allocated.
 

 \begin{example}
 Consider the following rules:
 \[
 \begin{tabular}{ccc}
 $p(x,y,z) \Leftarrow x \mapsto (x,y)$ & 
 $p(x,y,z) \Leftarrow x \mapsto (x,u) * q(u,z,z)$ & $q(x,y,z) \Leftarrow x \mapsto (y)$
 \end{tabular}
 \]
 Then $\outparam{p} = \{  1, 2 \}$ and $\outparam{q} = \{  2 \}$.
 \end{example}

 \begin{proposition}
 \label{prop:outparam}
 Let $\asid$ be a set of \prulesname.
 If $(\astore,\aheap) \modelsr p(t_1,\dots,t_n)$ and the index $i\neq 1$ is such that
 $i \not \in \outparam{p}$ 
 (i.e., $i$ is not an outgoing parameter of $p$) and $\asort_i = \addr$, then the entailment
 $(\astore,\aheap) \modelsr p(t_1,\dots,t_{i-1},s,t_{i+1}, t_n)$ holds for all terms  $s$.
 \end{proposition}
\begin{proof}
By an induction on the satisfiability relation. 
By hypothesis there exists a formula $\asheapB$ such that $p(t_1,\dots,t_n) \unfold \asheapB$, where and $\asheapB$ is of the form $t_1\mapsto (t_1', \ldots, t_k') * \asheapB'$ and
there exists an  \namedextension{\astore'}{\astore}{\vars(\asheapB) \setminus \vars(\asheap)}  such that $(\astore',\aheap) \models_{\asid} \asheapB$. 
By hypothesis $i\neq 1$, and since $i \not \in \outparam{p}$, $t_i$ cannot occur in $\set{t_1', \ldots, t_k'}$.  This entails 
that $p(t_1,\dots,t_{i-1},s,t_{i+1}, t_n) \unfold t_1\mapsto (t_1', \ldots, t_k') * \asheapB''$. If $\asheapB'$ contains a predicate $q(s_1, \ldots, s_m)$ and there exists an index $j$ such that $s_j = t_i$, then we cannot have $j = 1$ because $t_i \notin \set{t_1',\ldots, t_k'}$ and the rule under consideration is a \prule. Since $i \not \in \outparam{p}$ by hypothesis, $j$ cannot belong to $\outparam{q}$ and by induction, we deduce that $(\astore',\aheap) \models_{\asid} t_1\mapsto (t_1', \ldots, t_k') * \asheapB''$, hence the result.
\end{proof}
 
 Proposition~\ref{prop:outparam} states that if  $i \not \in \outparam{p} \cup \{ 1 \}$ and $\asort_i = \addr$, then
 the semantics of $p(t_1,\dots,t_n)$ does not depend on $t_i$, thus the $i$-th argument of $p$ is redundant and can be removed.
  This justifies the following: 
  \begin{assumption}
 \label{assume:no_useless_variable}
For all sets of \prulesname $\asid$ and for all predicate symbols $p: \addr,\asort_1,\dots,\asort_n\in \preds$, we assume that 
$\outparam{p} \supseteq \{ 2\leq i \leq n \mid \asort_i = \addr \}$. 
 \end{assumption}
 
 \newcommand{\pathto}[1]{\rightarrow_{#1}}

 
 \newcommand{\pathcompatible}{$\pathto{}$-compatible\xspace}

 \begin{definition}
\label{def:pathcompatible}
For any formula $\asheap$, we write $x \pathto{\asheap} y$ if 
$x,y \in \vars_{\addr}$ 
and $\asheap$ contains an atom $p(t_1,\dots,t_n)$ (resp.\ $t_1 \mapsto (t_2,\dots,t_n)$) such that 
$t_1 = x$ and $t_i = y$, for some $i \in \outparam{p}$ (resp.\ for some $i \in \{ 2,\dots,n\}$).

A structure $(\astore,\aheap)$ is called a {\em \pathcompatible model} of a formula $\asheap$
iff $(\astore,\aheap) \modelsr \asheap$ and for every $x,y \in \vars_{\addr}$,
$\astore(x) \connect{\aheap}^* \astore(y) \implies x \pathto{\asheap}^* y$.
 \end{definition}

Intuitively,  $x \pathto{\asheap} y$ states that the formula $\asheap$ allocates an edge from $x$ to $y$. 

\newcommand{\vdashr}{\vdash_{\asid}}
\newcommand{\vdashrV}[1]{\vdash_{\asid}^{#1}}
\newcommand{\aseq}{{\cal S}}
\newcommand{\eqfree}{equality-free\xspace}

\begin{definition}
\label{def:cmodel}
A {\em sequent} is an expression of the form $\asheap \vdashrV{V} \asheapB$, where $\asheap,\asheapB$ are symbolic heaps, $V$ is a multiset of variables of sort $\addr$ 
and $\asid$ is a finite set of inductive rules.
If $V = \emptyset$ then the sequent is written $\asheap \vdashr \asheapB$.
A sequent is {\em \eqfree} if $\asheap$ and $\asheapB$ contain no atoms of the form $u \iseq v$.
A {\em counter-model} of a sequent $\asheap \vdashrV{V} \asheapB$ is a structure $(\astore,\aheap)$ such that:
\begin{itemize}
	\item $(\astore,\aheap) \modelsr \asheap$ and
	$(\astore,\aheap) \not \modelsr \asheapB$,
	\item $\forall x \in V,\, \astore(x) \not \in \dom{\aheap}$,
	\item $\astore$ is injective on the multiset $V$.
\end{itemize} 
A sequent is {\em valid} iff it has no counter-model.
\end{definition}

\section{Lower Bounds}

\label{sect:limits}

We establish various lower bounds for the validity problems 
for sequents $\asheap \vdashr \asheapB$, where $\asid$ satisfies some additional conditions. These lower bounds will motivate the additional restrictions  that are imposed to devise a polynomial-time proof procedure.

Checking the validity of sequents $\asheap \vdashr \asheapB$ where $\asid$ is a set of \prulesname is actually undecidable in general.  
This result can be established by an argument similar to the one used in \cite{EP22} to prove the undecidability of PCE entailments modulo theories; it is not given here for the sake of conciseness, and because the goal of this paper is to investigate tractable cases.
The undecidability proof relies on the existence of variables of a sort distinct from $\addr$.  If such variables are forbidden, then the rules are PCE hence entailment is decidable \cite{IosifRogalewiczSimacek13}, but we still get an  \EXPTIME lower bound: 

\begin{proposition}
Checking the validity of sequents $\asheap \vdashr \asheapB$ is \EXPTIME-hard, even if $\asid$ is a set of \prulesname and all the  variables occurring in $\asheap \vdashr \asheapB$ are of sort $\addr$.
\end{proposition}
\begin{proof}
The proof is by a straightforward reduction from the inclusion problem for languages 
accepted by tree automata (see \cite{tata2007}).
Indeed, a tree automaton $(Q,V,\{ q_0 \},R)$ can be straightforwardly encoded as a set of \prulesname, where 
each rule $q \rightarrow f(q_1,\dots,q_n)$ in $R$ is encoded by an inductive rule of the form 
${q(x) \Leftarrow x \mapsto (f,x_1,\dots,x_n) }* \bigast_{i=1}^n q_i(x_i)$.
Each function symbol $f$ is considered as a constant of a sort $\asort \not = \addr$, and a 
term $f(t_1,\dots,t_n)$ is represented as a heap $\aheap_1 \union\dots \union \aheap_n \union \{ (\ell_0,\ell_1,\dots,\ell_n) \}$, where $\ell_0,\dots,\ell_n$ are pairwise distinct locations and $\aheap_1,\dots,\aheap_n$ are disjoint representations of $t_i$ with ${\ell_0 \not \in \dom{\aheap_i}}$, for $i = 1,\dots,n$. It is straightforward to verify that the language accepted by 
$(Q,V,\{ q_0 \},R)$ is included in that of
$(Q',V,\{ q_0' \},R')$
iff the sequent $q_0(x) \vdashrV{} q_0'(x)$ is valid.
 \end{proof}

\newcommand{\deterministic}{deterministic\xspace}

Since the inclusion problem is polynomial for top-down deterministic tree automata \cite{tata2007}, it is natural to further restrict the considered rules to make them deterministic, in the following sense:

\newcommand{\pln}{R}
\begin{definition}
\label{def:deterministic}
A set of \prulesname $\asid$ 
is {\em \deterministic}
if  for all pairs of distinct rules   of the form
$p(\vec{x}_i) \Leftarrow (y_i \mapsto \vec{t}_i * \asform_i) \swedge \apform_i$
  (where $i = 1,2$) occurring in $\asid$,
the formula $\vec{x}_1 \iseq \vec{x}_2 \wedge \vec{t}_1 \iseq \vec{t}_2 \wedge \apform_1 \wedge \apform_2$ is unsatisfiable (we assume by renaming that the rules share no variable).
\end{definition}
For instance the rules associated with the predicate $\ls$ in the introduction are not \deterministic, whereas the rules associated with $\tree$ are \deterministic, as well as all those  in Example~\ref{ex:als}.
For the predicate $\ls$, the formula 
$x \iseq x' \wedge y \iseq z$ is satisfiable, whereas for the predicate $\tree$, the formula 
$x \iseq x' \wedge () \iseq (y,z)$ is unsatisfiable (in both cases the variable $x$ is renamed by $x'$ in the second rule).

The following proposition shows that the restriction to \deterministic sets of \prulesname is still not sufficient to obtain a tractable validity problem:
\begin{proposition}
Checking the validity of sequents $\asheap \vdashr \asheapB$ is \PSPACE-hard, even if $\asid$ is a \deterministic set of \prulesname and all variables in $\asheap \vdashr \asheapB$ are of sort $\addr$.
\end{proposition}
\begin{proof}
Let ``$w \in E$'' be any problem in \PSPACE.
By definition, there exists a Turing machine 
$M= (Q,\Sigma,B, \Gamma,\delta, q_0, F)$ 
accepting exactly the words in $E$ and a polynomial 
 $\pln$ such that $M$ runs in space $\pln(n)$ on all words $w\in \Sigma^n$. 
The set $Q$ denotes the set of states of $M$, $\Sigma$ is the input alphabet, $B$ is the blank sumbol, $\Gamma$ is the tape alphabet, $\delta$ is the transition function, $q_0$ is the initial state and $F$ is the set of final states.
We shall reduce the problem ``$w \in E$'' to the entailment problem, for a sequent 
fulfilling the conditions above.
Consider a word $w$ of length $n$, and let $N = \pln(n)$.
Assume that $\csts$ contains all the elements in $\Gamma$.
We consider $\card{Q}\cdot N$ predicates $q^i$ of arity $N+3$, for all $q\in Q$ and $i \in \{ 1,\dots,N\}$, associated with the following rules: 
\begin{align*}
	q^i(x,y_1,\dots,y_N,u,v) &\Leftarrow \begin{multlined}[t]
		x \mapsto (x',u,v,a) * p^{i+\mu}(x',y_1,\dots,y_{i-1},b,y_{i+1},\dots,y_N,y_i,a)\\
		{\text{if $q \not \in F$ and $\delta$ contains a rule $(q,a) \rightarrow (p,b,\mu)$}}\\
		{\text{ with $i + \mu \in \{ 1,\dots,N \}$}}\\
	\end{multlined}  \notag \\
	q^i(x,y_1,\dots,y_N,u,v) &\Leftarrow 
	x \mapsto (x,u,v,B) 
	\text{,\ if $q \in F$.}
\end{align*}
%
Intuitively, $q$ is the state of the machine, the arguments $y_1,\dots,y_N$ denote the tape (that is of length $N$ by hypothesis) and
$i$ denotes the position of the head on the tape.
The constants $a,b$ denote the symbols read and written on the tape, respectively, $p$ is the final state of the transition rule and the integer $\mu$ denotes the move, i.e., an element of $\{ -1, 0, +1 \}$, 
 so that $i+\mu$ is the final position of the head on the tape. 
 Note that at this point the inductive rule does not test whether the symbol $a$ is indeed identical to the symbol at position $i$, namely $y_i$. 
 Instead, it merely  stores both $y_i$ and $a$ within the next tuple of the heap, by passing them as parameters to $p^{i+\mu}$. The arguments $u$ and $v$ are used to encode respectively
the symbol read on the tape at the previous state and the symbol that was expected.
By definition of the above rules, it is clear that $q^i(x,y_1,\dots,y_N,B,B)$ holds if the heap is a list of tuples  $(x_j,u_j,v_j,a_j)$, for $j = 1,\ldots,k$, linked on the first argument (the last element loops on itself). 
 The heap encodes a ``candidate run'' of length $k$ of $M$, i.e., a run for which one does not check, when applying a transition 
 $(p,a) \rightarrow (q,b,\mu)$, that the symbol read on the tape is identical to the expected symbol $a$. The symbols $u_i,v_i,a_i$ stored at each node are precisely the symbols that are read ($u_i$) and expected ($v_i$) at the previous step, respectively, along with the symbol $(a_i$) that is expected at the current step (this last symbol is added to ensure that the rules are deterministic). 
Note that for $i = 1$ there is no previous step and for $i = k$ no symbol is read since the state is final; thus by convention, 
$a_k$ is set to $B$ (see the last rule of $q^i$). Furthermore, $u_1,v_1$ will also be set to $B$ by invoking the initial state predicate with $B$ as the last two arguments  (see the definition of the sequent below).
 To check that the list corresponds to an actual run of $M$, it thus suffices to check that $u_i=v_i$ holds for all $i = 1,\dots,k$. 
 
The right-hand side of the sequent will allocate all structures not satisfying this condition. To this purpose, we associate with each state $r\in Q$ two predicate symbols $r_0$ and $r_1$ defined by the following rules 
(where $i,j \in \{ 0,1 \}$): 
\begin{align*}
	r_i(x) & \Leftarrow \begin{multlined}[t]
		 x \mapsto (x',a,b,c) * r_j(x')\\
		{\text{for all $r \not \in F$, $a,b,c \in \Gamma$, where $j = 1$ iff either $i = 1$ or $a \not = b$;}}
	\end{multlined}  \notag \\
	r_i(x)& \Leftarrow x \mapsto (x,a,b,B)
	\text{\quad if $a,b \in \Gamma$, $r \in F$ and either $i = 1$ or $a \not = b$.}
\end{align*}
Intuitively, the index $i$ in predicate $r_i$ is equal to $1$ iff a faulty location has been encountered (i.e., a tuple $(x',a,b)$ with $a \not =  b$).

Note that the number of rules is polynomial w.r.t.\ $N$, since the machine $M$ is fixed.
Also, the obtained set of rules is \deterministic, because $M$ is deterministic and the expected symbol is referred to by the location allocated by each predicate symbol $q^i$, thus the tuples $(x',u,v,a)$ corresponding to distinct rules associated with the same symbol $q^i$ cannot be unifiable.

Let $w_1.\dots,w_N = w.B^{N-n}$ be the initial tape (where $w$ is completed by blank symbols $B$ to obtain a word of length $N$).
It is clear that the sequent $q_0^1(x,w_1,\dots,w_N,B,B) \vdashr r_0(x)$ is valid  iff 
all the ``candidate runs'' of $M$ fulfill the conditions of the right-hand side, i.e., falsify at least one equality between read and expected symbols. Thus $q_0^1(x,w_1,\dots,w_N,B,B) \vdashr r_0(x)$ is valid  iff $M$ does not accept $w$. 
\end{proof}


In view of this result, it is natural to investigate the complexity of the entailment problem when the maximal arity of the predicates is bounded.
However, this is still insufficient to get a tractable problem, as the following lemma shows.

\begin{lemma}
\label{lem:deter}
Checking the validity of sequents $\asheap \vdashr \asheapB$ is co-\NP-hard, even if 
$\maxr{\asid} \leq 4$ (i.e., if the symbols and tuples are of arity at most $4$) 
and $\asid$ is a \deterministic set of \prulesname.
\end{lemma}

\newcommand{\colors}{\mathtt{Colors}}

\begin{proof}
The proof is by a reduction from the complement of the $3$-coloring problem, that is well-known to be NP-complete.
Let $G = (V,E)$ be a graph, where $V = \{ v_1,\dots,v_n \}$ is a finite set of vertices and $E$ is a set of undirected edges, i.e., a set of unordered pairs 
of vertices.
Let 
$\colors = \{ a,b,c \}$ be a set of colors, with $\card{\colors} = 3$.
We recall that a  solution of the $3$-coloring problem is a function $f: V \rightarrow \colors$ such that $(x,y) \in E \implies f(x) \not = f(y)$.
We assume, w.l.o.g., that all vertices occur in at least one edge.
We consider two distinct sorts $\addr$ and $\asort$.
We assume, w.l.o.g., that $V \cup \colors \subseteq \vars_{\asort}$ 
(i.e., $a,b,c$, as well as the set of vertices in $G$, are variables) and $V \cap \colors = \emptyset$.

Let $E = \{ (x_i,y_i) \mid i = 1,\dots,m \}$ (where the edges are ordered arbitrarily) and
let $u_1,\dots,u_{m+1}$ 
be pairwise distinct variables of sort $\addr$.
Let $\asform$ be the formula: $\bigast_{i=1}^m\ u_i \mapsto (x_i,y_i,u_{i+1}) * u_{m+1} \mapsto ()$. 
Let $p$ and $q$ be predicate symbols associated with the following rules:
\[
\begin{array}{lll}
p(u,a,b,c) & \Leftarrow & u \mapsto (v,v,u') * q(u') \\
p(u,a,b,c) & \Leftarrow & u \mapsto (v_1,v_2,u') * v_1 \not \iseq v_2 * v_1 \not \iseq a * v_1 \not \iseq b * v_1 \not \iseq c * q(u') \\
p(u,a,b,c) & \Leftarrow & u \mapsto (d,v_2,u') * v_2 \not \iseq a * v_2 \not \iseq b * v_2 \not \iseq c * q(u') \\
& & \text{\quad for all $d \in \{ a,b,c \}$} \\
 p(u,a,b,c) & \Leftarrow & u \mapsto (d_1,d_2,u') * p(u',a,b,c) \\
& & \text{\quad for all $d_1,d_2 \in \{ a,b,c \}$ where $d_1 \not = d_2$} \\
q(u) & \Leftarrow & u \mapsto (v_1,v_2,u') * q(u') \\
q(u) & \Leftarrow & u \mapsto () 
\end{array}
\]

Intuitively, any model $(\astore,\aheap)$ of  $\asform$ encodes a candidate solution of the 
$3$-coloring problem, where each variable  $z \in \{ x_1,\dots,x_m \} \cup \{ y_1,\dots, y_m \}$  is mapped to an element $\astore(z)$ in $\univ_{\asort}$. The heap $\aheap$ is a list of tuples linked on the last element and containing a tuple  $(\astore(u_i),\astore(x_i),\astore(y_i),\astore(u_{i+1}))$ for all $(x_i,y_i) \in E$. To check that this candidate solution indeed fulfills the required properties, one has to verify that all the pairs $(\astore(x_i),\astore(y_i))$ are composed of distinct elements in $\{ a,b,c \}$.

By definition of the rules for predicate $p$, $(\astore,\aheap)$ is a model of $ p(u_1,a,b,c)$ iff 
the list contains a pair $(\astore(x),\astore(y))$ such that one of the following holds:
\begin{itemize}
	\item $\astore(x) = \astore(y)$ (first rule of $p$),
	\item $\astore(x) \not = \astore(y)$ and $\astore(x) \not \in \{ \astore(a),\astore(b),\astore(c) \}$ (second rule of $p$),
	\item $\astore(x) \in \{ \astore(a),\astore(b),\astore(c) \}$ and $\astore(y) \not \in \{ \astore(a),\astore(b),\astore(c) \}$ (third rule of $p$).
\end{itemize}   
After the cell corresponding to this faulty pair is allocated, $q$ is invoked to allocate the remaining part of the list.
Thus $p(u_1,a,b,c)$ holds iff the model does {\em not} encode a solution of the $3$-coloring problem, either because  $\astore(x_i) = \astore(y_i)$ for some $(x_i,y_i)\in E$ or 
because one of the variables is mapped to an element distinct from $a,b,c$ -- note that by the above assumption, each of these variables 
occurs in the list.
Consequently $\asform \vdashrV{} p(u_1,a,b,c)$ admits a counter-model iff 
there exists a model of $\asform$ that does not satisfy $p(u_1,a,b,c)$, i.e., iff 
the $3$-coloring problem admits a solution  (thus the entailment is valid iff the $3$-coloring problem admits no solution).
\end{proof}


The results above motivate the following definition, that strengthens the notion of a deterministic set of rules. 

\begin{definition}
A set of \prulesname $\asid$ 
is {\em \hdeterministic}
if it is \deterministic and all the disequations occurring in the rules in $\asid$ are of the form $x \not \iseq y$ with $x,y \in \vars_\addr$.
\end{definition}

The intuition behind \hdeterministic rules is that, to get an efficient proof procedure, we have to restrict the amount of equational reasoning needed to establish the validity of the sequents.
Disequations between locations are relatively easy to handle because (by definition of \prules) all existential variables of sort $\addr$ must be pairwise distinct (as they are allocated in distinct atoms). However, dealing with disequations between data is much more difficult, as evidenced by the proof of Lemma \ref{lem:deter}. Thus we restrict such disequations to those occurring in the initial sequent.


\newcommand{\listc}{\mathtt{cons}}
The rules associated with $\als$, $\tree$, $\treeB$, $\treeC$ or $\dll$ 
in the introduction and in Example~\ref{ex:als} are \hdeterministic. 
In contrast, the following rules are \deterministic, but not \hdeterministic (where $u,v$ denote variables of some sort distinct from $\addr$):
\[
\begin{array}{llllll}
p(x,u) & \Leftarrow & x \mapsto (v) \swedge v \not \iseq u \quad  &
\quad p(x,u) & \Leftarrow & x \mapsto (u)
\end{array}
\]
Rules that are \hdeterministic are well-suited to model 
constructor-based data structures used in standard programming languages; for instance,  lists could be represented as follows (where $\listc$ is a constant symbol denoting a constructor and $y$ is a variable of some sort distinct from $\addr$, denoting data stored in the list):
\[ 
\begin{array}{llllll}
\ls(x) \Leftarrow & x \mapsto (\listc,y,z) * \ls(z) \quad & \quad
\ls(x) \Leftarrow & x \mapsto () 
\end{array}
\]
We end this section by establishing a key property of \deterministic set of rules, namely the fact that every spatial formula $\asform$ is {\em precise}, in the sense of \cite{DBLP:conf/lics/CalcagnoOY07}:
 it is fulfilled on at most one subheap within a given structure.

\begin{lemma}
\label{lem:unique}
Let $\asid$ be a \deterministic set of rules.
For every spatial formula $\asform$, for every store $\astore$ and for every heap $\aheap$ there exists at most one heap $\aheap'$ such that $\aheap' \subseteq \aheap$ and
$(\astore,\aheap') \modelsr \asform$.
\end{lemma}
\begin{proof}
The proof is by induction on the satisfiability relation $\modelsr$. 
Note that by hypothesis $\asform$ is a spatial formula, hence contains no occurrences of $\iseq$, $\not \iseq$, $\swedge$ or $\wedge$.
Assume that there exist two heaps $\aheap_1', \aheap_2'$ such that $\aheap_i' \subseteq \aheap$  and $(\astore,\aheap_i') \modelsr \asform$ (for $i = 1,2$). We show that $\aheap_1' = \aheap_2'$.
\begin{itemize}
\item{If $\asform = \emp$ then necessarily $\aheap_i = \emptyset$ for $i = 1,2$ thus $\aheap_1' = \aheap_2'$.}
\item{If $\asform = y_0 \mapsto (y_1,\dots,y_n)$ then by Definition \ref{def:rmodel} we have $\aheap_i' = \{ (\astore(y_0),\dots,\astore(y_n)) \}$ for $i = 1,2$ thus $\aheap_1' = \aheap_2'$.}
\item{If $\asform = \asform_1 * \asform_2$ then for all $i = 1,2$ there exist two disjoint heaps $\aheap_i^j$ (for $j = 1,2$) such that $\aheap_i' = \aheap_i^1 \union \aheap_i^2$ for $i = 1,2$ and $(\astore,\aheap_i^j) \modelsr \asform_j$, for  $i,j \in \{ 1,2 \}$. Since $\aheap_i^j \subseteq \aheap_i' \subseteq \aheap$ we get by the induction hypothesis 
 $\aheap_1^j = \aheap_2^j$ for $j = 1,2$. Therefore $\aheap_1^1 \union \aheap_1^2 = \aheap_2^1 \union \aheap_2^2$, i.e., $\aheap_1' = \aheap_2'$.}
 \item{If $\asform$ is a predicate atom of root $x$, then for  $i = 1,2$
 we have $\asform \unfold \asheap_i$, 
 and there exists an \namedextension{\astore_i}{\astore}{\vars(\asformB_i) \setminus \vars(\asform)}
 such that $(\astore_i,\aheap_i') \modelsr \asheap_i$.
 Since  $\asid$ is a set of \prulesname, 
 $\asheap_i$ 
 is of the form 
 $(x_i \mapsto \vec{y}_i * \asform_i) \swedge \apform_i$ and there exist disjoint heaps 
 $\aheap_i^j$ (for $j = 1,2$) such that the following conditions are satisfied:
 \begin{inparaenum}[(i)]
 \item{$\aheap_i' = \aheap_i^1 \union \aheap_i^2$;}
 \item{
  $(\astore_i,\aheap_i^1) \modelsr x_i \mapsto (\vec{y}_i)$;}
  \item{
 $(\astore_i,\aheap_i^2) \modelsr \asform_i$;}
 \item{
 and $\astore_i \models \apform_i$.}
 \end{inparaenum}
Furthermore, $x_i$ must be the root of $\asform$, thus $x_1 = x_2 = x$.
 For  $i = 1, 2$ we have $\aheap_i^1 = \{ (\astore(x_i),\astore_i(\vec{y}_i)) \}$, and since $\aheap_i^1 \subseteq \aheap$, necessarily $\astore_1(\vec{y}_1) = \astore_2(\vec{y}_2)$ and $\aheap_1^1 = \aheap_2^1$.
 The heap $\aheap_i^2$ is the restriction of $\aheap_i'$ to the locations distinct from $\astore(x)$.
We distinguish two cases.
\begin{itemize}
\item{Assume that the inductive rules applied on $\asform$ to respectively derive $\asheap_1$ and $\asheap_2$ are different.
We may assume by $\alpha$-renaming that 
 $(\vars(\asheap_1) \setminus \vars(\asform)) \cap (\vars(\asheap_2)  \setminus \vars(\asform)) = \emptyset$, which entails that there exists a store $\astore'$ that coincides with $\astore_i$ on $\vars(\asheap_i)$ (since $\astore_1$ and $\astore_2$ coincide on the variables in $\vars(\asform)$).
 Then  we have $\astore' \models \vec{y}_1 \iseq \vec{y}_2$ (since  $\astore_1(\vec{y}_1) = \astore_2(\vec{y}_2)$) and $\astore' \models \apform_i$ (since $\astore_i \models \apform_i$), which entails that the formula $\vec{y}_1 \iseq \vec{y}_2 \wedge \apform_1 \wedge \apform_2$ is satisfiable, contradicting the fact that $\asid$ is \deterministic.
  }
  \item{Assume that the same rule is used to derive both $\asheap_1$ and $\asheap_2$.
  We may assume in this case (again by $\alpha$-renaming) that the vector of variables occurring in $\asheap_1$ and $\asheap_2$ are the same, so that $\vec{y}_1 = \vec{y}_2$ and $\asform_1 = \asform_2$.
  Since $\asid$ is a set of \prulesname, all variables $z$ in $\vars(\asform_i) \setminus \vars(\asform)$ occur in $\vec{y}_i$. As $\astore_1(\vec{y}_1) = \astore_2(\vec{y}_2)$, this entails that $\astore_1(z) = \astore_2(z)$ holds for all such variables, thus $\astore_1 = \astore_2$. 
  Consequently, $(\astore_1,\aheap_i^2) \modelsr \asform_1$, for all $i = 1,2$ with 
  $\aheap_i^2 \subseteq \aheap_i' \subseteq \aheap$. By the induction hypothesis this entails
  that $\aheap_1^2 = \aheap_2^2$, thus $\aheap_1' = \aheap_2'$.}
\end{itemize}}
\end{itemize}
\end{proof}

\begin{example}
Lemma \ref{lem:unique} does not hold if the rules are not \deterministic. For instance, the formula $\ls(x,y)$ (with the rules given in the introduction) has two models $(\astore, \aheap)$ and $(\astore, \aheap')$ where $\aheap'$ is a strict subheap of $\aheap$: $\astore(x)= \ell_1$, $\astore(y)= \ell_2$, 
$\aheap = \{ \ell_1 \mapsto (\ell_2), \ell_2 \mapsto (\ell_2) \}$ and
$\aheap' = \{ \ell_1 \mapsto (\ell_2) \}$. Intuitively, the formula $\ls(y,y)$ (which is useful to derive $\ls(x,y)$) can be derived by any of the two rules of $\ls$, yielding two different models. In contrast $\als(x,y)$ (with the rules of Example \ref{ex:als}) has only one model with the store $\astore$ and a heap included in $\aheap$, namely $(\astore,\aheap')$.
\end{example}

 \section{Proof Procedure}
 
 \label{sect:proof}
 
From now on, we consider a fixed \hdeterministic set of \prulesname $\asid$, satisfying Assumptions \ref{assume:productive} and \ref{assume:no_useless_variable}. 
 For technical convenience, we also assume that $\asid$ is nonempty and that every constant in $\csts$ occurs in a rule in $\asid$.
   
 \subsection{Some Basic Properties of {\tt P}-Rules}

 \newcommand{\heapunsat}{heap-unsatisfiable\xspace}
\newcommand{\heapsat}{heap-satisfiable\xspace}

\newcommand{\roots}[1]{\mathit{roots}(#1)}

\newcommand{\alloc}[1]{\mathit{alloc}(#1)}
 

We begin by introducing some definitions and deriving straightforward consequences of the definition of \prulesname.
We shall denote by $\alloc{\asheap}$ the multiset of variables allocated by a formula $\asheap$:
\begin{definition}
\label{def:alloc}
For every formula 
$\asheap$, we denote by 
$\alloc{\asheap}$ the multiset of variables $x$
such that $\asheap$ contains a spatial atom with root $x$.
\end{definition}
Lemma \ref{lem:alloc} states that the variables in $\alloc{\asheap}$ are necessarily 
allocated in every model of $\asheap$, which entails (Corollary \ref{cor:heapunsat})  that they must be associated with pairwise distinct locations. Moreover, a formula distinct from $\emp$ has at least one root, hence allocates at least one variable (Corollary \ref{cor:noemp}).
\begin{lemma}
\label{lem:alloc}
Let $\asheap$ be a formula.
If $(\astore,\aheap) \modelsr \asheap$ and $x\in \alloc{\asheap}$ then 
$\astore(x) \in \dom{\aheap}$.
\end{lemma}
\begin{proof}
By hypothesis, $\asheap$ is of the form 
$(\anatom * \asform) \swedge \apform$ 
where $\anatom$ is a spatial atom with root $x$.
Thus $(\astore,\aheap) \models \anatom * \asform$ and 
there exist disjoint heaps $\aheap_1,\aheap_2$ such that
$(\astore,\aheap_1) \modelsr \anatom$,
$(\astore,\aheap_2) \modelsr \asform$,
and $\aheap = \aheap_1 \union \aheap_2$.
Since the root of $\anatom$ is $x$, $\anatom$ is either of the form 
$x \mapsto \vec{y}$ or of the form 
$p(x,\vec{y})$ where $p \in \preds$.
In the former case, it is clear that $\astore(x) \in \dom{\aheap_1} \subseteq \dom{\aheap}$ since 
$(\astore,\aheap_1) \modelsr \anatom$. In the latter case, we have $\anatom \unfold \asheap'$ and 
$(\astore',\aheap) \modelsr \asheap'$, where $\astore'$ is an \extension{\astore}{\vars(\asheap') \setminus \vars(\anatom)}.
Since $\asid$ is a set of \prulesname,  necessarily 
$\asheap'$ contains a points-to atom of the form $x \mapsto \vec{z}$, which entails that 
$\astore'(x) \in \dom{\aheap_1}$, hence
$\astore(x) \in \dom{\aheap}$.
\end{proof}

\begin{corollary}
\label{cor:heapunsat}
Let $\asheap$ be a formula and let $(\astore,\aheap)$ be an \rmodel of $\asheap$.
If $\{ x,y \} \subseteqm \alloc{\asheap}$ then $\astore(x) \not = \astore(y)$.
In particular, if $\alloc{\asheap}$ contains two occurrences of the same variable $x$ 
then $\asheap$ is unsatisfiable.
\end{corollary}
\begin{proof}
By definition, $\asheap$ is of the form $(\anatom_1 * \anatom_2 * \asform) \swedge \apform$, where $\anatom_1$ and $\anatom_2$ are 
spatial atoms of roots $x$ and $y$, respectively, with  $\alloc{\anatom_1} = \{ x \}$ and $\alloc{\anatom_2} = \{ y \}$. 
If $\asheap$ admits a model $(\astore,\aheap)$, then there exists disjoint heaps 
$\aheap_1,\aheap_2,\aheap'$ such that $\aheap = \aheap_1 \union \aheap_2 \union \aheap'$,
$(\astore,\aheap_i) \modelsr \anatom_i$ (for $i = 1,2$)
and $(\astore,\aheap') \modelsr \asform$.
By Lemma~\ref{lem:alloc} we have $\astore(x) \in \dom{\aheap_1}$
and $\astore(y) \in \dom{\aheap_2}$, thus $\astore(x) \not = \astore(y)$ since 
$\aheap_1$ and $\aheap_2$ are disjoint.
\end{proof}
 \begin{corollary}
 \label{cor:noemp}
 Let $\asform$ be a spatial formula and let $(\astore,\aheap)$ be an  \rmodel of $\asform$.
If $\asform \not =\emp$ then $\aheap \not = \emptyset$.
 \end{corollary}
 \begin{proof}
Since $\asform \not = \emp$, necessarily $\asform$ contains at least one atom $\anatom$, thus 
$\rootof{\anatom} \in \alloc{\asform}$. Then the result follows immediately from Lemma~\ref{lem:alloc}.
 \end{proof}
Corollary \ref{cor:heapunsat} motivates the following definition, which provides a simple syntactic criterion to identify some formulas that cannot be satisfiable, due to the fact that the same variable is allocated twice.
\begin{definition}
A formula $\asheap$ is {\em \heapunsat} if 
$\alloc{\asheap}$ contains  two occurrences of the same variable.
Otherwise, it is {\em \heapsat}.	
\end{definition}



The next proposition states that every location that is referred to in the heap of some model of $\asheap$ must be reachable from one of the roots of $\asheap$. This follows from the fact that, by definition of \prules, the set of allocated locations has a tree-shaped structure: the root of each atom invoked in an inductive rule must be connected to the location allocated by the rule (see Condition~\ref{it:prog2} in Definition~\ref{def:simple}).

 \begin{proposition}
 \label{prop:connect}
 Let $\asheap$ be a symbolic heap and let 
 $(\astore,\aheap)$ be an  \rmodel of $\asheap$.
 For every $\ell \in \locs{\aheap}$, there exists $x\in \alloc{\asheap}$ such that $\astore(x) \connect{\aheap}^* \ell$.
 \end{proposition}
 \begin{proof}
 The proof is by induction on the satisfiability relation. 
 We establish the result also for  spatial formulas and pure formulas.
 \begin{itemize}
 \item{If $\asheap = \emp$ or is $\asheap$ is a pure formula then $\locs{\asheap} = \emptyset$ hence the proof is immediate.}
 \item{If $\asheap = y_0 \mapsto (y_1,\dots,y_n)$, then $\aheap = \{ (\astore(y_0),\dots,\astore(y_n)) \}$, by Definition \ref{def:rmodel}, thus $\locs{\aheap} = \{ \astore(y_i) \mid i = 0,\dots,n, \text{and $y_i$ is of sort $\addr$} \}$ and $\astore(y_0) \connect{\aheap}^* \astore(y_i)$, for all $i = 1,\dots,n$ such that $y_i$ is of sort $\addr$.
 By Definition \ref{def:alloc} $\alloc{\asheap} = \{ y_0 \}$, thus the proof is completed.}
 \item{
  If $\asheap = \asform \swedge \apform$, where $\asform$ is a spatial formula and $\apform$ is a pure formula  distinct from $\true$, 
  then we have $(\astore,\aheap) \modelsr \asform$ and $\alloc{\asheap} = \alloc{\asform}$, hence the result follows
  immediately from the induction hypothesis.}
  \item{If $\asheap = \asform_1 * \asform_2$, then there exist two disjoint heaps $\aheap_1,\aheap_2$ such that 
  $\aheap = \aheap_1 \union \aheap_2$ and $(\astore,\aheap_i) \modelsr \asform_i$, for all $i = 1,2$.
If $\ell \in \locs{\aheap}$ then necessarily $\ell \in \locs{\aheap_i}$ for some $i = 1,2$.
By the induction hypothesis, we deduce that there exists $x \in \alloc{\asform_i}$ such that $\astore(x) \connect{\aheap_i}^* \ell$. Since $\alloc{\asform_i} \subseteqm \alloc{\asform}$ and
$\connect{\aheap_i}^* \subseteq \connect{\aheap}^*$ by Proposition~\ref{prop:subheap_connect}, we obtain the result.
 }
 \item{If $\asheap$ is a predicate atom, then $\asheap \unfold \asheapB$ and 
 $(\astore',\aheap) \modelsr \asheapB$ for some formula $\asheapB$ and some \namedextension{\astore'}{\astore}{\vars(\asheapB) \setminus \vars(\asheap)}.
 Let $\ell \in \locs{\aheap}$. By the induction hypothesis, there exists $x \in \alloc{\asheapB}$ such that 
 $\astore(x) \connect{\aheap}^* \ell$. By definition, $x$ is the root of some atom $\anatom$ in $\asheapB$.
 If $\anatom$ is a points-to atom, then since $\asheap \unfold \asheapB$  is an instance of a rule in $\asid$ and  all rules are \prules, $x$ must be the root of $\asheap$; in this case $x\in \alloc{\asheap}$ and the proof is completed.
 Otherwise, $x$ is the root of a spatial atom in $\asheapB$, and, because all rules are \prules, $\asheapB$ must contain an atom of the form $y_0 \mapsto (y_1,\dots,y_n)$, such that $y_0 = \rootof{\asheap}$ and 
 $y_i = x$, for some $i = 1,\dots,n$.
 Since $(\astore',\aheap)\modelsr \asheapB$, we have $(\astore(y_0),\dots,\astore(y_n)) \in \aheap$, hence 
 $\astore(y_0) \connect{\aheap} \astore(x)$.
 Using the fact that $\astore(x) \connect{\aheap}^* \ell$, we deduce that 
 $\astore(y_0) \connect{\aheap}^* \ell$, hence the proof is completed since $\alloc{\asheap} = \{ y_0 \}$.
}
  \end{itemize}
 \end{proof}

\newcommand{\alloccomplete}{alloc-complete\xspace}
\newcommand{\linear}{linear\xspace}

%

The next lemma asserts a key property of the considered formulas: all the locations occurring in the heap of a model of some formula $\asform$ are either allocated or associated with a variable that is free in $\asform$. This follows from the definition of \prules: all  variables of sort $\addr$ that are existentially quantified in an inductive  rule must be allocated (at the next recursive call). Recall that spatial formulas contain no quantifiers.
\begin{lemma}
\label{lem:establish}
Let $\asform$ be a  spatial formula and let 
$(\astore,\aheap)$ be an \rmodel of $\asform$.
Then the following inclusion holds: ${\locs{\aheap} \subseteq \dom{\aheap} \cup \astore(\vars(\asform))}$.
\end{lemma}
\begin{proof}
The proof is by induction on the relation $\modelsr$. 
Note that as $\asform$ is spatial, it contains no occurrence of $\iseq$, $\not \iseq$, $\swedge$ and $\wedge$.
\begin{itemize}
\item{If $\asform$ is of the form $y_0 \mapsto (y_1,\dots,y_n)$ then by definition 
$\aheap = \{ (\astore(y_0),\dots,\astore(y_n)) \}$ and 
$\locs{\aheap} = \{ \astore(y_0),\dots,\astore(y_n) \} = \astore(\vars(\asform))$.
}
\item{
	If $\asform = \emp$ then $\aheap = \emptyset$ thus $\locs{\aheap} = \emptyset$ and the proof is immediate.}
\item{If $\asform$ is a predicate atom then we have $\asform \unfold \asformB \swedge \apform$, and 
$(\astore',\aheap) \modelsr \asformB$, for some \namedextension{\astore'}{\astore}{\vars(\asformB \swedge \apform)\setminus \vars(\asform)}.
Let $\ell \in \locs{\aheap} \setminus \dom{\aheap}$. By the induction hypothesis,
$\ell = \astore'(x)$ for some variable $x \in \vars(\asformB)$. If $x\in \vars(\asform)$ then necessarily $\astore'(x) = \astore(x)$, thus $\ell \in \astore(\vars(\asform))$ and the proof is completed. Otherwise,
by Definition \ref{def:simple}, since all the rules in $\asid$ are \prules; $x$ occurs as the root of some predicate atom in 
$\asformB$, i.e., $x\in \alloc{\asformB}$. 
By Lemma~\ref{lem:alloc} we deduce that $\astore'(x) \in \dom{\aheap}$, i.e., $\ell \in \dom{\aheap}$ which contradicts our assumption.}
\item{If $\asform$ is of the form $\asform_1 * \asform_2$ then 
there exist disjoint heaps $\aheap_1,\aheap_2$ such that  $(\astore,\aheap_i) \modelsr \asform_i$ (for $i = 1,2$)
and $\aheap = \aheap_1 \union \aheap_2$. Let $\ell \in \locs{\aheap} \setminus \dom{\aheap}$. Necessarily we have 
$\ell \in \locs{\aheap_i}$, for some $i = 1,2$ and $\ell \not \in \dom{\aheap_i}$, hence by the induction hypothesis we deduce that 
$\ell = \astore(x)$, for some $x \in \vars(\asform_i)$. Since $\vars(\asform_i) \subseteq \vars(\asform)$, the proof is completed.
}
\end{itemize}
\end{proof}

\subsection{A Restricted Entailment Relation}

\newcommand{\Emodels}{\triangleright_V}
\newcommand{\gEmodels}[1]{\triangleright_{#1}}
\newcommand{\nEmodels}{{\not\triangleright_V}}

We introduce a simple syntactic criterion, used in the inference rules of Section~\ref{sect:rules}, that is sufficient to ensure that 
a given pure formula $\apform$ holds in every counter-model of a sequent with left-hand side $\asheap$ and multiset of variables $V$. The idea is to test that $\apform$ either  occurs in $\asheap$, is trivial, or is a disequation entailed by the fact that the considered store must be injective on $\alloc{\asheap} \cup V$ (using Definition~\ref{def:cmodel} and Corollary~\ref{cor:heapunsat}). Lemma~\ref{lem:emodel} states that the relation satisfies the expected property.
\begin{definition}
\label{def:emodel}
Let $\asheap$ be a symbolic heap,  $\apform$ be a pure formula and let $V$ be a multiset of variables.
We write $\asheap \Emodels \apform$ if for every atom $\apformB$ occurring in $\apform$, one of the following conditions holds:
\begin{enumerate}
\item{$\apformB$ occurs in $\asheap$;}
 \item{either $\apformB = (t \iseq t)$ for some variable $t$, or $\apformB = (t_1 \not \iseq t_2)$ and
 $t_1,t_2$ are distinct constants;}
 \item{$\apformB = (x_1 \not \iseq x_2)$ modulo commutativity, and one of the following holds: $\{ x_1,x_2\} \subseteqm \alloc{\asheap}$,  ($x_1 \in \alloc{\asheap}$ and $x_2 \in V$) or 
 $\{ x_1,x_2 \} \subseteqm V$. 
 This is equivalent to stating that $\{ x_1,x_2\} \subseteqm \alloc{\asheap} + V$ where $\alloc{\asheap} + V$ denotes as usual the union of the multisets $\alloc{\asheap}$ and $V$.}
 \label{it:emodel:neq:sub}
\end{enumerate}
\end{definition}
\begin{example}
Consider the symbolic heap $\asheap = (p(x,y) * q(z)) \swedge x \not \iseq y$, and let $V  = \{ u \}$.
We have $\asheap \Emodels x \not \iseq y \wedge x \not \iseq z \wedge x \not \iseq u$. Indeed, $x$ and $z$ are necessarily distinct since they are allocated by distinct atoms $p(x,y)$ and $q(z)$ (as, by definition of the \prules, every predicate allocates it first parameter) $x$ cannot be equal to $u$ as  $u \in V$ and $V$ is intended to denote a set of non-allocated variables (see Definition \ref{def:cmodel}) and $x$ is allocated, and $x$ cannot be equal to $y$ as the disequation $x \not \iseq y$ occurs in $\asheap$.
\end{example}

\begin{lemma}
\label{lem:emodel}
Let $\asheap$ be a symbolic heap and let $\apform$ be a pure formula such that 
$\asheap \Emodels \apform$.
For every structure
$(\astore,\aheap)$, if $(\astore,\aheap) \modelsr \asheap$, $\astore(V) \cap \dom{\aheap} = \emptyset$ and $\astore$ is injective on $V$ then $\astore \modelsr \apform$.
\end{lemma}
\begin{proof}
We show that $(\astore,\aheap)\modelsr \apformB$, for all atoms $\apformB$ in $\apform$.
We consider each case in Definition \ref{def:emodel} separately:
\begin{enumerate}
\item{If $\apformB$ occurs in $\asheap$ then since $(\astore,\aheap) \models \asheapB$, necessarily $\astore \modelsr \apformB$. }
\item{If $\apformB = (t \iseq t)$ then it is clear that $\astore \models \apformB$. 
If $\apformB = (t_1 \not \iseq t_2)$ and 
 $t_1,t_2$ are distinct constants then $\astore(t_1) \not =\astore(t_2)$ since all stores are injective on constants by definition. 
}
\item{If $\apformB = x_1 \not \iseq x_2$ and $\{ x_1,x_2 \} \subseteqm \alloc{\asheap}$ then by Corollary \ref{cor:heapunsat}, we get
$\astore(x_1) \not = \astore(x_2)$ 
since $\astore(x_1)$ and $\astore(x_2)$ must be allocated in disjoint heaps.
Thus $\astore \modelsr \apformB$.
If $x_1 \in \alloc{\asform}$ and $x_2 \in V$ then  
we have $\astore(x_1) \in \dom{\aheap}$ by Lemma~\ref{lem:alloc}, and since 
$\astore(V) \cap \dom{\aheap} = \emptyset$, we deduce that 
 $\astore \modelsr \apformB$.
Finally, if $\{ x_1,x_2 \} \subseteqm V$ then $\astore(x_1) \not = \astore(x_2)$, because $\astore$ is injective on $V$ by hypothesis.
}
\end{enumerate}
\end{proof}

\subsection{Inference Rules}

\renewcommand{\vdashr}{\vdash_{\asid}^V}

\label{sect:rules}

\newcommand{\ruleid}[1]{\text{{\color{blue} {\tt #1}}}}
\newcommand{\eq}{\ruleid{R}}
\newcommand{\noteq}{\ruleid{E}}
\newcommand{\sep}{\ruleid{S}}
\newcommand{\unf}{\ruleid{U}}
\newcommand{\imi}{\ruleid{I}}
\newcommand{\wea}{\ruleid{W}}
\newcommand{\dec}{\ruleid{C}}
\newcommand{\eli}{\ruleid{V}}


We consider the rules represented in Figure \ref{fig:rules}. The rules apply modulo AC, they are intended to be applied  
bottom-up: 
 a rule is {\em applicable} on a sequent $\asheap \vdashr \asheapB$ if there exists an instance of the rule
the conclusion of which is $\asheap \vdashr \asheapB$.
We recall that an inference rule is {\em sound} if the validity of the premises entails the validity of the conclusion, and 
{\em invertible} if the converse holds.

\begin{remark}
The application conditions given in Figure \ref{fig:rules} are the most general ones ensuring that the rules are sound. Additional application conditions will be provided afterwards (see Definition \ref{def:admissible}) to obtain a proof procedure that 
runs in polynomial time. The latter conditions are rather complex, and in our opinion introducing them at this point could hinder the understanding of the rules.
\end{remark}


\begin{figure}
{\small
\begin{prooftree}
\AxiomC{$\repl{\asform}{x}{t} \swedge \repl{\apform}{x}{t} \vdashrV{\repl{V}{x}{t}} \repl{\asheapB}{x}{t}$}
\RightLabel{\quad if $x\in \vars$}
\LeftLabel{\eq:}
\UnaryInfC{ $\asform \swedge (x \iseq t \wedge \apform) \vdashr \asheapB$}
\end{prooftree}


\begin{prooftree}
\AxiomC{$\asheap \vdashr \asformB  \swedge \apformB$}
\RightLabel{\quad if $\asheap \Emodels \apformB'$}
\LeftLabel{\noteq:}
\UnaryInfC{$\asheap \vdashr \asformB  \swedge (\apformB \wedge \apformB')$}
\end{prooftree}

\begin{prooftree}
\AxiomC{$\asform_1 \swedge \apform_1 \vdashrV{V \cup \alloc{\asform_2}} \asformB_1 \swedge \apformB_1$}
\AxiomC{$\asform_2 \swedge \apform_2 \vdashrV{V  \cup \alloc{\asform_1}} \asformB_2 \swedge \apformB_2$}
\LeftLabel{\sep:}
\RightLabel{\quad if $\asform_1 \not = \emp$ and $\asform_2 \not = \emp$}
\BinaryInfC{$(\asform_1 * \asform_2) \swedge (\apform_1 \wedge \apform_2) \vdashr (\asformB_1 * \asformB_2) \swedge (\apformB_1 \wedge \apformB_2)$}
\end{prooftree} 

\begin{prooftree}
\AxiomC{$(\asform'_1 * \asform) \swedge (\apform'_1 \wedge \apform) \vdashr \asheapB$}
\AxiomC{$\dots$}
\AxiomC{$(\asform'_n * \asform) \swedge (\apform'_n \wedge \apform) \vdashr \asheapB$}
\LeftLabel{\unf:}
\RightLabel{if $p(\vec{t}) \unfoldall \{ \asform'_i \swedge \apform'_i \mid i = 1,\dots,n \}$}
\TrinaryInfC{$(p(\vec{t}) * \asform) \swedge \apform \vdashr \asheapB$}
\end{prooftree}

\begin{prooftree}
\AxiomC{$(x \mapsto (y_1,\dots,y_k) * \asform) \swedge \apform \vdashr (x \mapsto (y_1,\dots,y_k) * \asformB'\sigma * \asformB) \swedge \apformB$}
\LeftLabel{\imi:}
\RightLabel{\quad if all the conditions below hold:}
\UnaryInfC{$(x \mapsto (y_1,\dots,y_k) * \asform) \swedge \apform \vdashr (p(x,\vec{z}) * \asformB) \swedge \apformB$}
\end{prooftree}
(i) $p(x,\vec{z}) \unfold (x \mapsto (u_1,\dots,u_k) * \asformB') \swedge \apformB'$; (ii)
$\sigma$ is a substitution whose domain is included in ${\{ u_1,\dots,u_k \} \setminus (\{ x \} \cup \vec{z})}$; (iii)
 $\sigma(u_i) = y_i$, for all $i \in \{ 1,\dots,k\}$; and (iv) $(x \mapsto (y_1,\dots,y_k) * \asform) \swedge \apform \Emodels \apformB'\sigma$ 

\vspace*{0.25cm}

\begin{prooftree}
\AxiomC{$\asform \swedge \apform \vdashr \asheapB$}
\LeftLabel{\wea:}
\RightLabel{\qquad if one of the conditions below holds:}
\UnaryInfC{$\asform \swedge (\apform \wedge \apform') \vdashr \asheapB$}
\end{prooftree}
(i) $\apform' = x \not \iseq y$ and $\asform \Emodels x \not \iseq y$; or (ii) $\apform' = \bigwedge_{i=1}^n x \not \iseq y_i$ and $x \not \in \vars(\asform) \cup \vars(\apform) \cup \vars(\asheapB) \cup \{ y_1,\dots,y_n \}$

\vspace*{0.25cm}

\begin{prooftree}
\AxiomC{$\repl{\asform}{x}{y} \swedge \repl{\apform}{x}{y} \vdashrV{\repl{V}{x}{y}} \repl{\asheapB}{x}{y}$}
\AxiomC{$\asform \swedge (\apform \wedge x \not \iseq y) \vdashr \asheapB$}
\LeftLabel{\dec:}
\RightLabel{\qquad if $x,y \in \vars_{\addr}$}
\BinaryInfC{$\asform \swedge \apform  \vdashr \asheapB$}
\end{prooftree}

\begin{prooftree}
\AxiomC{$\asheap \vdashrV{V} \asheapB$}
\LeftLabel{\eli:}
\RightLabel{\qquad if $x \not \in \vars(\asheap) \cup \vars(\asheapB)$}
\UnaryInfC{$\asheap \vdashrV{V \cup \{ x \}} \asheapB$}
\end{prooftree}}

\caption{The Inference rules \label{fig:rules}}
\end{figure}


We provide some examples and explanations concerning the inference rules.

\begin{example}
Rules $\eq$ (replacement) and $\noteq$ (elimination) handle equational reasoning.
For instance, given the sequent $(p(x,u) * p(y,u)) \swedge (u \iseq v) \vdashrV{\emptyset} (p(x,u) * p(y,v)) \swedge (x \not \iseq y)$, one may first apply $\eq$, yielding: 
$p(x,u) * p(y,u) \vdashrV{\emptyset} (p(x,u) * p(y,u)) \swedge (x \not \iseq y)$. As $\{ x,y \} \subseteqm \alloc{p(x,u) * p(y,x)}$, $\noteq$ applies, which yields the trivially valid sequent
$p(x,u) * p(y,u) \vdashrV{\emptyset} p(x,u) * p(y,u)$.
\end{example}

\begin{example}
Rule $\sep$ (separation) applies on the sequent $p(x) * p(y) \vdashrV{\emptyset} q(x) *r(y)$, yielding $p(x) \vdashrV{\{ y \}} q(x)$
and
$p(y) \vdashrV{\{ x \}} r(y)$ (as $\{ x \} = \alloc{p(x)}$ and
$\{ y \} = \alloc{p(y)}$). The addition of $y$ (resp.\ $x$) to the variables associated with the sequent allows us to keep track of the fact that these variables cannot be allocated in $p(x)$ (resp.\ $p(y)$) as they are already allocated in the other part of the heap.
Note that the rule also yields
$p(x) \vdashrV{\{ y \}} r(y)$
and
$p(y) \vdashrV{\{ x \}} q(x)$, however as we shall see  later (see Definition~\ref{def:antiax}) the latter premises cannot be valid and this application can be dismissed.
\end{example}

\newcommand{\cstA}{\mathtt{a}}
\newcommand{\cstB}{\mathtt{b}}

\begin{example}
Rules $\unf$ (unfolding) and $\imi$ (imitation) both unfold inductively defined predicate symbols. $\unf$ unfolds predicates occurring on the left-hand side of a sequent, yielding one premise for each inductive rule. In contrast, $\imi$ applies on the right-hand side and selects one rule  in a non-deterministic way (provided it fulfills the rule application condition), yielding exactly one premise.
Let $\asid$ be the following set of rules, where $\cstA,\cstB$ are constant symbols and $z,z_1,z_2$ are variables of the same sort as $\cstA$ and $\cstB$:
\[
\begin{tabular}{ccc}
$p(x)$ & $\Leftarrow$  & $x \mapsto (\cstA,x)$ \\
$p(x)$ & $\Leftarrow$  & $x \mapsto (\cstB,x)$ \\
$q(x)$ & $\Leftarrow$  & $x \mapsto (z,y)$ \\
$q(x)$ & $\Leftarrow$  & $x \mapsto (z_1,z_2,x)$
\end{tabular}
\]
Rule $\unf$ applies on $p(x,y) \vdashrV{\emptyset} q(x)$, yielding
$x \mapsto (\cstA,x) \vdashrV{\emptyset} q(x)$
and
$x \mapsto (\cstB,x) \vdashrV{\emptyset} q(x)$.
Then $\imi$ applies on both sequents, with the respective substitutions 
$\{ y\leftarrow x, z \leftarrow \cstA \}$ and $\{ y\leftarrow x, z \leftarrow \cstB \}$,  yielding
$x \mapsto (\cstA,x) \vdashrV{\emptyset} x \mapsto (\cstA,x)$
and
$x \mapsto (\cstB,x) \vdashrV{\emptyset} x \mapsto (\cstB,x)$.
Note that $\imi$ cannot be applied with the inductive rule $q(x) \Leftarrow  x \mapsto (z_1,z_2,x)$, as $(\cstA,x)$ and $(\cstB,x)$ are not instances of $(z_1,z_2,x)$.
\end{example}

\begin{example}
Rules $\wea$ (weakening) and $\eli$ (variable weakening) allow to get rid of irrelevant information, which is essential for termination.
For instance one may deduce $p(x) \vdashrV{\emptyset} q(x)$ from
$p(x) \swedge (u \not \iseq v \wedge u \not \iseq w) \vdashrV{\emptyset} q(x)$.
Indeed, 
if the former sequent admits a counter-model, then one gets a counter-model of the latter one by associating $u,v,w$ with pairwise distinct elements.
\end{example}

\begin{example}
Rule $\dec$ performs a case analysis on $x \iseq y$.
It is essential to allow further applications of Rule $\imi$ in some cases.
Consider  the sequent 
$u \mapsto (x,x) * p(x) \vdashrV{\emptyset} q(u,y)$, with the rules
$q(u,y) \Leftarrow u \mapsto (y,z) * p(z)$, 
and
$q(u,y) \Leftarrow (u \mapsto (z,z) * p(z)) \swedge z \not \iseq y$.
Rule $\imi$ does not apply because 
there is no substitution $\sigma$ with domain $\{ z \}$ such that
$(x,x) = (y,z)\sigma$, and
${u \mapsto (x,x) * p(x) \not\!\Emodels x \not \iseq y}$.
The rule $\dec$ yields 
$u \mapsto (y,y) * p(y) \vdashrV{\emptyset} q(u,y)$
and
$(u \mapsto (x,x) * p(x)) \swedge x \not \iseq y \vdashrV{\emptyset} q(u,y)$. Then the rule $\imi$ applies on both sequents, yielding the premisses
${u \mapsto (y,y) * p(y) \vdashrV{\emptyset} u \mapsto (y,y) * p(y)}$
and
${(u \mapsto (x,x) * p(x)) \swedge x \not \iseq y \vdashrV{\emptyset} 
u \mapsto (x,x) * p(x)}$.
\end{example}

We have the following facts, which can be verified by an inspection of the inference rules:
\begin{proposition}\label{prop:rule-facts}
	Consider an 
	\eqfree sequent $\asheap \vdashrV{V} \asheapB$ that is the conclusion of an inference rule, of which $\asheap' \vdashrV{V'} \asheapB'$ is a premise.
	\begin{enumerate}
		\item No rule can introduce any equality to $\asheap' \vdashrV{V'} \asheapB'$. 
		\item If $x\in V'\setminus V$, then the inference rule is either {\eli} or \dec.\label{it:rem:v} 
		\item If $\alloc{\asheapB'}\subsetneq \alloc{\asheapB}$ then the inference rule is either {\sep} or \dec.\label{it:rem:alloc} 
		\item The only inference rules that can introduce new variables to $\asheapB'$ are 
		$\imi$ and $\dec$.\label{it:intro:eq} 
		\item 
		No rule introduces a disequation between terms of a sort distinct from $\addr$.\label{it:rf:intro:dis}
		\item 
		The only rule that introduces a predicate atom 
		to the right-hand side of a premise is \imi.\label{it:pt:imi}
		\item If 
		$v\in (\alloc{\asheapB} \cup V)\setminus (\alloc{\asheapB'} \cup V')$, then the inference rule must be {\dec}.\label{it:rem:var:alloc}
	
	\end{enumerate}
\end{proposition}

\begin{proof}
The first six facts are straightforward to verify. Fact \ref{it:rem:var:alloc} holds because Rule {\eq} cannot apply if no equality occurs; if rule {\sep} is applied then the variables in $\alloc{\asheapB} \setminus \alloc{\asheapB'}$ must occur in $V'$; Rule {\imi} deletes a predicate atom but introduces a points-to atom with the same root and rule {\eli} cannot be applied on variables occurring in $\vars{\asheapB}$. 
\end{proof}

We  establish additional basic properties of the inference rules. All the rules are sound and invertible, except for rule \sep\  that is only sound.
The results follow easily from the semantics, except for the invertibility of $\imi$, which crucially depends on the fact that rules are \deterministic. Indeed, $\imi$ unfolds one atom on the right-hand side by selecting one specific inductive rule. In our case, at most one rule can be applied, which ensures that  equivalence is preserved. This is a crucial point  because otherwise one would need to consider disjunctions of formulas on the right-hand side of the sequent (one disjunct for each 
possible rule), which would make the procedure much less efficient.

\begin{example}
Consider the sequent $x \mapsto (y) * y \mapsto (z) \vdashrV{\emptyset} p(x,z)$, with the  rules
\[ \begin{tabular}{rclrcl}
$p(x,z)$ & $\Leftarrow$ & $x \mapsto (y) * q(y,z)$ \qquad &
$p(x,z)$ & $\Leftarrow$ & $x \mapsto (y) * q'(y,y)$ \\
$q(y,z)$ & $\Leftarrow$ & $y \mapsto (z) \swedge y \not \iseq z$ &
$q'(y,z)$ & $\Leftarrow$ & $y \mapsto (z)$ \\
\end{tabular}
\]
It is clear that the rules are not \deterministic, as there is an overlap between the two rules associated with $p$. Applying rule \imi\ on the above sequent yields either 
$x \mapsto (y) * y \mapsto (z) \vdashrV{\emptyset} x \mapsto (y) * q(y,z)$
or
$x \mapsto (y) * y \mapsto (z) \vdashrV{\emptyset} x \mapsto (y) * q'(y,y)$. None of these two possible premises is valid, although the initial sequent is valid. This shows that \imi\ is not invertible (although it is still sound) when $\asid$ is not \deterministic. The intuition is that it is not possible to decide which rule must be applied on $p$ before deciding whether $z$ is equal to $y$ or not.
\end{example}

\begin{lemma}
\label{lem:rules}
The rules \eq, \noteq, \unf, \wea, \eli\ and \dec\ and {\imi} are sound and invertible.
More specifically, for all heaps $\aheap$, the conclusion of the rule admits a counter-model $(\astore,\aheap)$ iff one of the premises admits a counter-model  $(\astore',\aheap)$. 
\end{lemma}
\begin{proof}
We consider each rule separately (we refer to the definitions of the rules for notations) and establish the equivalence of the lemma.
We recall (Definition~\ref{def:cmodel}) that a counter-model of a sequent is a structure $(\astore,\aheap)$ that validates the antecedent, falsifies the conclusion, and is such that the store is injective on $V$ and no variable in $V$ is allocated. 
\begin{itemize}
\item[\eq]{Assume that $(\astore,\aheap) \modelsr \asform \swedge (x \iseq t \wedge \apform)$, that
$(\astore,\aheap) \not \modelsr \asheapB$, that $\astore(V) \cap \dom{\aheap} = \emptyset$ and that $\astore$ is injective on $V$.
Then we have $\astore(x) = \astore(t)$, thus 
$(\astore,\aheap) \modelsr \repl{\asform}{x}{t} \swedge \repl{\apform}{x}{t}$ and
$(\astore,\aheap) \not \modelsr \repl{\asheapB}{x}{t}$.
For all $y \in \repl{V}{x}{t}$, we have either $y\in V$ and 
$\astore(y) \not \in \dom{\aheap}$, or $x \in V$ and $y = t$, thus $\astore(y) = \astore(t) = \astore(x) \not \in \dom{\aheap}$.
 Finally, 
 assume (for the sake of contradiction) that $\{ u,v \} \subseteq \repl{V}{x}{t}$ with 
$\astore(u) = \astore(v)$. Then there exist variables $u',v'$ such that
$\{ u',v' \} \subseteq V$, 
with $\repl{u'}{x}{t} = u$ and $\repl{v'}{x}{t} = v$. If $u'$ and $v'$ are both equal to $x$ or both distinct from $x$   then we have $\astore(u') = \astore(v')$, which contradicts the fact that $\astore$ is injective on $V$. If $u'=x$ and $v'\neq x$, then we have $\astore(v') = \astore(t) = \astore(x)$, again contradicting the fact that $\astore$ is injective on $V$. The proof is similar when $u'\neq x$ and $v' = x$.
Consequently, $(\astore,\aheap)$ is also a counter-model of the premise.


Conversely, assume that $(\astore,\aheap) \modelsr \repl{\asform}{x}{t} \swedge \repl{\apform}{x}{t}$; 
$\astore(\repl{V}{x}{t}) \cap \dom{\aheap} = \emptyset$; 
the store $\astore$ is injective on $\repl{V}{x}{t}$; and $(\astore,\aheap) \not \modelsr \repl{\asheapB}{x}{t}$.
 We consider the store $\astore'$ such that $\astore'(x) = \astore(t)$ and $\astore'(y) = \astore(y)$ if $y \not = x$. It is clear that $\astore' \modelsr x \iseq t$, 
$(\astore',\aheap) \modelsr \asform \swedge \apform$
and $(\astore',\aheap) \not \modelsr \asheapB$. For all $y \in V$, if $y \not = x$ then $y \in \repl{V}{x}{t}$ and $\astore'(y) = \astore(y) \not \in \dom{\aheap}$; otherwise $y = x$  and  $t \in \repl{V}{x}{t}$, thus
$\astore'(y) = \astore'(x) = \astore(t) \not \in \dom{\aheap}$.
There only remains to check that $\astore'$ is injective on $V$, which is done by contradiction: if this is not the case then there exists
$\{ u,v\}  \subseteqm V$ such that $\astore'(u) = \astore'(v)$. Hence 
$\{ \repl{u}{x}{t}, \repl{v}{x}{t} \} \subseteqm \repl{V}{x}{t}$, and we have
$\astore'(\repl{u}{x}{t}) = \astore(\repl{u}{x}{t})$
and
$\astore'(\repl{v}{x}{t}) = \astore(\repl{v}{x}{t})$. This contradicts the fact that $\astore$ is injective on $\repl{V}{x}{t}$.

}

\item[\noteq]{

Let $(\astore,\aheap)$ be a structure such that 
$(\astore,\aheap) \modelsr \asheap$;
$\astore(V) \cap \dom{\aheap} = \emptyset$; the store $\astore$ is injective on $V$;
and
$(\astore,\aheap) \not \modelsr \asformB  \swedge (\apformB \wedge \apformB')$.
By the application condition of the rule we have
$\asheap \Emodels \apformB'$, thus
by Lemma~\ref{lem:emodel}, we deduce that $(\astore,\aheap) \modelsr \apformB'$.
Therefore $(\astore,\aheap) \not \modelsr \asformB  \swedge \apformB$.
Conversely, it is clear that any counter-model of 
$\asheap \vdashr \asformB  \swedge \apformB$ is a counter-model of 
$\asheap \vdashr \asformB  \swedge (\apformB \wedge \apformB')$.
}

\item[\unf]{
Assume that $(\astore,\aheap) \modelsr (p(\vec{t}) * \asform) \swedge \apform$,
$\astore(V) \cap \dom{\aheap} = \emptyset$,  $\astore$ is injective on $V$
and
${(\astore,\aheap) \not \modelsr  \asheapB}$.
By definition, $\astore \models \apform$, and there are disjoint heaps $\aheap_1,\aheap_2$ such that 
${(\astore,\aheap_1) \modelsr p(\vec{t})}$, 
${(\astore,\aheap_2) \modelsr \asform}$ and $\aheap = \aheap_1 \union \aheap_2$.
We get that $p(\vec{t}) \unfold \asform' \swedge \apform'$ and
$(\astore',\aheap_1) \modelsr \asform' \swedge \apform'$, for some \extension{\astore}{\vars(\asform' \swedge \apform') \setminus \vars(p(\vec{t}))}. 
We assume that
${\vars(\asform' \swedge \apform') \cap {\vars((p(\vec{t}) * \asform) \swedge \apform) \subseteq \vec{t}}}$ (by  $\alpha$-renaming).
We get
$\astore' \models \apform'$ 
and 
$(\astore',\aheap_1 \union \aheap_2) \modelsr \asform' * \asform$, hence
$(\astore',\aheap) \modelsr (\asform' * \asform) \swedge (\apform' \wedge \apform)$. 
Furthermore, the formula $\asform' \swedge \apform'$ occurs in $\{ \asform_i' \swedge \apform_i' \mid i =1,\dots,n\}$, up to a renaming of the  variables not occurring in $\vec{t}$, by definition of $\unfoldall$.
Thus $(\astore',\aheap)$ is a counter-model of 
a sequent $(\asform'_i * \asform) \swedge (\apform'_i \wedge \apform) \vdashr \asheapB$, for some $i = 1,\dots,n$.

Conversely, let $(\astore,\aheap)$ be a structure such that
$(\astore,\aheap) \modelsr (\asform'_i * \asform) \swedge (\apform'_i \wedge \apform)$,
$\astore(V) \cap \dom{\aheap} = \emptyset$,  $\astore$ is injective on $V$
and
$(\astore,\aheap) \not \modelsr  \asheapB$.
We deduce that $\astore \models \apform'_i \wedge \apform$ and 
there exist disjoint heaps $\aheap_1,\aheap_2$ such that
$(\astore,\aheap_1) \modelsr \asform'_i$
and
$(\astore,\aheap_2) \modelsr \asform$.
By the application condition of the rule we have 
$p(\vec{t}) \unfold \asform'_i \swedge \apform'_i$, thus 
$(\astore,\aheap_1) \modelsr p(\vec{t})$ by definition of the semantics of predicate atoms (since $\astore$ is an extension of itself).
Consequently, $(\astore,\aheap_1 \union \aheap_2) \modelsr p(\vec{t}) * \asform$, hence
 $(\astore,\aheap) \modelsr (p(\vec{t}) * \asform) \swedge \apform$, and
$(\astore,\aheap)$ is a counter-model of $(p(\vec{t}) * \asform) \swedge \apform \vdashr \asheapB$.
}

\item[\wea]{
Assume that $(\astore,\aheap) \modelsr \asform \swedge \apform$, 
$\astore(V) \cap \dom{\aheap} = \emptyset$,  $\astore$ is injective on $V$
and 
$(\astore,\aheap) \not \modelsr \asheapB$.
If $\apform'$ is $x \not \iseq y$, with $\asform \Emodels x \not \iseq y$, then, by Lemma~\ref{lem:emodel}, we have 
$\astore(x) \not = \astore(y)$, thus $(\astore,\aheap) \modelsr \asform \swedge (\apform \wedge \apform')$, hence
$(\astore,\aheap)$ is a counter-model of the conclusion of the rule.
Assume that $\apform' = \bigwedge_{i=1}^n x \not \iseq y_i$ with $x \not \in \vars(\asform) \cup \vars(\apform) \cup \vars(\asheapB) \cup \{ y_1,\dots,y_n \}$.
Let $\astore'$ be a store such that $\astore'(y) = \astore(y)$ if $y \not = x$ and $\astore'(x)$ is an arbitrarily chosen location not occurring in $\dom{\aheap} \cup \{ \astore(y_i) \mid i = 1,\dots,n \}$ (such a location exists since $\aheap$ is finite and $\univ_{\addr}$ is infinite).
By definition, we have $\astore'(V) \cap \dom{\aheap} = \emptyset$, 
$\astore' \models \bigwedge_{i=1}^n x \not \iseq y_i$,
$(\astore',\aheap) \modelsr \asform \swedge \apform$ (since $\astore'$ and $\astore$ coincide on $\vars(\asform) \cup \vars(\apform)$)
and
$(\astore',\aheap) \not \modelsr \asheapB$ (since $\astore'$ and $\astore$ coincide on $\vars(\asheapB)$).
Thus $(\astore',\aheap)$ is a counter-model of $\asform \swedge (\apform \wedge \bigwedge_{i=1}^n x \not \iseq y_i) \vdashr \asheapB$.

Conversely, it is clear that 
any counter-model of $\asform \swedge (\apform \wedge \apform') \vdashr \asheapB$
is a counter-model of $\asform \swedge \apform \vdashr \asheapB$.

}

\item[\dec]{
Let $(\astore,\aheap)$ such that $(\astore,\aheap) \modelsr \asform \swedge \apform$, 
$\astore(V) \cap \dom{\aheap} = \emptyset$,  $\astore$ is injective on $V$ 
and 
$(\astore,\aheap) \not \modelsr \asheapB$.
We distinguish two cases.
\begin{itemize}
\item{If $\astore(x) = \astore(y)$ then 
$(\astore,\aheap)  \modelsr \repl{\asform}{x}{y}$,
$(\astore,\aheap)  \modelsr \repl{\apform}{x}{y}$
and
$(\astore,\aheap) \not \modelsr \repl{\asheapB}{x}{y}$.
Moreover, $\astore(\repl{V}{x}{y}) \cap \dom{\aheap} = 
\astore(V) \cap \dom{\aheap} = \emptyset$ and $\astore$ must be injective on $\repl{V}{x}{y}$, hence 
$(\astore,\aheap)$ falsifies  
$\repl{\asform}{x}{y} \swedge \repl{\apform}{x}{y} \vdashrV{\repl{V}{x}{y}} \repl{\asheap}{x}{y}$.}
\item{Otherwise, we have $\astore \models x \not \iseq y$, thus 
$(\astore,\aheap) \modelsr \asform \swedge (\apform \wedge x \not \iseq y)$, and 
$(\astore,\aheap)$ is a counter-model of 
$\asform \swedge (\apform \wedge x \not \iseq y) \vdashr \asheap$.}
\end{itemize}

Conversely, if $(\astore,\aheap)$ is a counter-model of
$\asform \swedge (\apform \wedge x \not \iseq y) \vdashr \asheapB$ then it is clear that it is also a counter-model of 
$\asform \swedge \apform \vdashr \asheapB$.
If $(\astore,\aheap)$ is a counter-model of 
$\repl{\asform}{x}{y} \swedge \repl{\apform}{x}{y} \vdashrV{\repl{V}{x}{y}} \repl{\asheapB}{x}{y}$, then consider the store $\astore'$ such that $\astore'(x) = \astore(y)$ and 
$\astore'(z) = \astore(z)$ if $z\not = x$.
By definition, 
$(\astore',\aheap) \modelsr \asform \swedge \apform$
and 
$(\astore',\aheap) \not \modelsr \asheapB$.
The set ${\astore'(V) \cap \dom{\aheap} = \astore(\repl{V}{x}{y})  \cap \dom{\aheap}}$ is empty,
and $\astore'$ is injective on $V$, since $\astore$ is injective on $\repl{V}{x}{y}$
and by definition $\astore'(u) = \astore(\repl{u}{x}{y})$, for all variables $u$.
Therefore $(\astore',\aheap)$ is a counter-model of $\asform \swedge \apform \vdashr \asheapB$.
}

\item[\eli]{
Let $(\astore,\aheap)$ be a counter-model of 
$\asheap \vdashrV{V} \asheapB$.
Then 
$(\astore,\aheap) \modelsr \asheap$
and $(\astore,\aheap)  \not\modelsr \asheapB$.
Let $\astore'$ be a store such that $\astore'(x)\not \in \dom{\aheap} \cup \astore(V)$ and
$\astore'(y) = \astore(y)$, if $y \not = x$.
Since $x\not \in \vars(\asheap) \cup \vars(\asheapB)$ we have
$(\astore',\aheap) \modelsr \asheap$,
and $(\astore',\aheap) \not \modelsr \asheapB$. Moreover $\astore'(x) \not \in \dom{\aheap}$ and $\astore'$ is injective on $V \cup \{ x \}$ (since $\astore$ is injective on $V$ and $\astore'(x)\not \in V$), thus
$(\astore',\aheap)$ is a counter-model of
$\asheap \vdashrV{V \cup \{ x\}} \asheapB$.

Conversely, it is clear that any counter-model of $\asheap \vdashrV{V \cup \{ x \}} \asheapB$ is also a counter-model
of $\asheap \vdashr \asheapB$.
}

\item[\imi]{
Let $(\astore,\aheap)$ be a counter-model of 
$(x \mapsto (y_1,\dots,y_k) * \asform) \swedge \apform \vdashr (p(x,\vec{z}) * \asformB) \swedge \apformB$.
By definition $(\astore,\aheap)\modelsr (x \mapsto (y_1,\dots,y_k) * \asform)  \swedge \apform$,
$(\astore,\aheap) \not \modelsr (p(x,\vec{z}) * \asformB) \swedge \apformB$,
 the set ${\astore(V) \cap \dom{\aheap}}$ is empty and $\astore$ is injective on $V$.
  By the application conditions of the rule and Lemma~\ref{lem:emodel}, we have $\astore\models \apformB'\sigma$.
Assume 
for the sake of contradiction that $(\astore,\aheap)$ is not a counter-model of 
$(x \mapsto (y_1,\dots,y_k) * \asform) \swedge \apform \vdashr (x \mapsto (y_1,\dots,y_k) * \asformB'\sigma * \asformB) \swedge \apformB$.
This entails that 
${(\astore,\aheap) \modelsr (x \mapsto (y_1,\dots,y_k) * \asformB'\sigma * \asformB) \swedge \apformB}$, hence 
$\astore \models \apformB$, and there exist disjoint heaps $\aheap_1,\aheap_2$ and $\aheap_3$ such that
$\aheap = \aheap_1 \union \aheap_2 \union \aheap_3$,
$(\astore,\aheap_1) \modelsr x \mapsto (y_1,\dots,y_k)$,
$(\astore,\aheap_2) \modelsr \asformB'\sigma$
and
$(\astore,\aheap_3) \modelsr   \asformB$.
Let
$\asheap = x \mapsto (u_1,\dots,u_k) * \asformB') \swedge \apformB'$, so that $p(x,\vec{z}) \unfold \asheap$. Since 
all the rules in $\asid$ are \prules, necessarily $\vars(\asheap) \setminus \vars(p(x,\vec{z})) \subseteq \{ u_1,\dots,u_k \}$.
Let $\astore'$ be the store defined by:
$\astore'(y) = \astore(\sigma(y))$, for all $y \in \vars$.
By the application condition of the rule, $\dom{\sigma}$ is a subset of $\{ u_1,\dots,u_k\} \setminus (\set{x} \cup \vec{z})$, hence $\astore'$ coincides with $\astore$ on all variables in $x,\vec{z}$. 
We deduce that
$\astore'$ is an \extension{\astore}{\{ u_1,\dots,u_k \} \setminus \vars(p(x,\vec{z})) }.
We have 
${(\astore,\aheap_1) \modelsr x \mapsto (y_1,\dots,y_k)}$, with $\sigma(x) = x$ and $\sigma(u_i) = y_i$, thus
$(\astore,\aheap_1) \modelsr (x \mapsto (u_1,\dots,u_k))\sigma$.
Moreover, we also have $(\astore,\aheap_2) \modelsr \asformB'\sigma$, hence
$(\astore,\aheap_1 \union \aheap_2) \modelsr \asheap\sigma$. 
By Proposition~\ref{prop:subst_store}, we get
$(\astore',\aheap_1 \union \aheap_2) \modelsr \asheap$.
Since $p(x,\vec{z}) \unfold \asheap$, we deduce that
$(\astore,\aheap_1 \union \aheap_2) \modelsr p(x,\vec{z})$, hence
$(\astore,\aheap) \modelsr p(x,\vec{z}) *  \asformB$.
Thus $(\astore,\aheap)  \modelsr (p(x,\vec{z}) * \asformB) \swedge \apformB$, which contradicts our hypothesis.

Conversely, 
assume that $(x \mapsto (y_1,\dots,y_k) * \asform) \swedge \apform \vdashr (p(x,\vec{z}) * \asformB) \swedge \apformB$ is valid, and let 
$(\astore,\aheap)$ be an \rmodel of $(x \mapsto (y_1,\dots,y_k) * \asform) \swedge \apform$
such that $\astore(V) \cap \dom{\aheap} = \emptyset$ and $\astore$ is injective on $V$.
We deduce that $(\astore,\aheap) \modelsr (p(x,\vec{z}) * \asformB) \swedge \apformB$, thus $\astore \models \apformB$, and there exist disjoint heaps 
$\aheap_1,\aheap_2$ such that
$(\astore,\aheap_1) \modelsr (p(x,\vec{z})$
and
$(\astore,\aheap_2) \modelsr \asformB$.
This entails that there exists a symbolic heap $\asheap$ and a \namedextension{\astore'}{\astore}{\vars(\asheap) \setminus \vars(p(x,\vec{z}))} such that $p(x,\vec{z}) \unfold \asheap$, and
$(\astore',\aheap_1) \modelsr \asheap$.
Since the rules in $\asid$ are \prules, $\asheap$ is of the form 
$(x \mapsto (v_1,\dots,v_m) * \asformB'') \swedge \apformB''$.
Moreover, it is clear that $\aheap(\astore(x)) = (\astore(y_1),\dots,\astore(y_k))$, so that $m = k$ and 
$\astore'(v_i) = \astore(y_i)$, for all $i = 1,\dots,k$.
Let $\sigma'$ be the substitution mapping every variable in $\{ v_1,\dots,v_k\}$ not occurring in $x,\vec{z}$ to the first variable $y_j$ such that $\astore'(v_i) = \astore(y_j)$.
By definition, we have $\astore' = \astore \circ \sigma'$, thus by Proposition~\ref{prop:subst_store}, we get 
$(\astore,\aheap_1) \modelsr \asheap\sigma'$.

Assume 
for the sake of contradiction that the inductive rule used to derive
$\asheap$ is distinct from the one 
used to derive the formula $\asheap' = (x \mapsto (u_1,\dots,u_k) * \asformB') \swedge \apformB'$ in the application 
condition of the rule. 
We may assume by renaming that $\vars(\asheap') \cap \vars(\asheap) \subseteq \vars(p(x,\vec{z})$.
Let $\astore''$ be a store coinciding with $\astore'$ on all constants and on all variables in 
$\vars(\asheap)$, and such that, for all variables $y\in \vars(\asheap') \setminus \vars(p(x,\vec{z}))$,
$\astore''(y) = \astore(\sigma(y))$.
Since $\astore' \models \apformB''$ we have 
$\astore'' \models \apformB''$.
By the application condition of the rule $\sigma(u_i) = y_i$, thus $\astore'' \models u_i \iseq v_i$, for all $i = 1,\dots,k$.
Since $\astore \models \apform$ and $\apform \models \apformB'\sigma$ (still by the application condition of the rule),
we get $\astore \models \apformB'\sigma$, and by Proposition~\ref{prop:subst_store} we deduce that 
$\astore \models \apformB'$.
Thus $(u_1,\dots,u_k) \iseq (v_1,\dots,v_k) \wedge \apformB' \wedge \apformB''$ is  satisfiable, which contradicts the fact that $\asid$ is \deterministic.

This entails that the rules applied to derive 
$\asheap$ and $\asheap'$ are the same, and by renaming we may assume in this case that $(u_1,\dots,u_k) = (v_1,\dots,v_k)$ (which entails that ${(x \mapsto (v_1,\dots,v_k))\sigma} = {x \mapsto (y_1,\dots,y_k)}$),
with $\asformB' = \asformB''$ and $\apformB' = \apformB''$.
Using the fact that $\sigma(u_i) = y_i$ and $\astore'(v_i) = \astore(y_i)$, for all $i = 1,\dots,k$, it is easy to check that 
$\astore'(y) = \astore(\sigma(y))$, for all variables $y$.
Since $(\astore',\aheap_2) \modelsr \asheap$
we get by Proposition~\ref{prop:subst_store},
$(\astore,\aheap_2) \modelsr \asheap\sigma$, i.e., 
$(\astore,\aheap_2) \modelsr (x \mapsto (y_1,\dots,y_k) * \asformB'\sigma)$.
Therefore $(\astore,\aheap) \modelsr (x \mapsto (y_1,\dots,y_k) * \asformB'\sigma * \asformB) \swedge \apformB$.

}

\end{itemize}
\end{proof}

Rule \sep\ is sound but in contrast with the other rules, it is not invertible in general (intuitively, there is no guarantee that the decomposition of the left-hand side of the sequent corresponds to that of the right-hand side).
\begin{example}
Consider the (valid) sequent $x \mapsto (y) * y \mapsto (x) \vdashrV{\emptyset} p(x,y) * p(y,x)$ with the rule $p(u,v) \Leftarrow u \mapsto (v)$. Rule \sep\ applies, yielding the valid premises $x \mapsto (y) \vdashrV{\emptyset} p(x,y)$ and
$y \mapsto (x) \vdashrV{\emptyset} p(y,x)$. However, since the rules apply modulo  commutativity of $*$ we may also get the premises:
$x \mapsto (y) \vdashrV{\emptyset} p(y,x)$ and
$y \mapsto (x) \vdashrV{\emptyset} p(x,y)$ which are not valid.
\end{example}

\begin{lemma}
\label{lem:sep_sound}
Rule \sep\ is sound. More specifically, if $(\astore,\aheap)$ is a counter-model of 
the conclusion, then one of the premises admits a counter-model $(\astore,\aheap')$, where $\aheap'$ is a proper subheap of $\aheap$.
\end{lemma}
\begin{proof}
Let $(\astore,\aheap)$ be an \rmodel of $(\asform_1 * \asform_2) \swedge (\apform_1 \wedge \apform_2)$, where $\astore(V) \cap \dom{\aheap}= \emptyset$ and $\astore$ is injective on $V$.
Assume that the premises admit no counter-model of the form $(\astore,\aheap')$ with $\aheap' \subset \aheap$.
By definition, there exist disjoint heaps 
$\aheap_1,\aheap_2$ such that 
$\aheap = \aheap_1 \union \aheap_2$,
and $(\astore,\aheap_i) \modelsr \asform_i$, for $i = 1,2$.
Since $\asform_i \not = \emp$,
$\asform_i$ contains at least one predicate atom, with a root $x_i$. By Lemma~\ref{lem:alloc}, necessarily $\astore(x_i) \in \dom{\aheap_i}$, so $\aheap_i$ is not empty and $\aheap_i \subset \aheap$ for $i = 1,2$.
Still by Lemma~\ref{lem:alloc}, $\astore(\alloc{\asform_i}) \subseteq \dom{\aheap_i}$, and  since $\aheap_1$ and $\aheap_2$ are disjoint, we deduce that
$\astore(\alloc{\asform_{3-i}}) \cap \dom{\aheap_i} = \emptyset$ for $i=1,2$.
Thus $\astore(V \cup \alloc{\asform_{3-i}}) \cap \dom{\aheap_i} = \emptyset$.
By Corollary \ref{cor:heapunsat} $\astore$ is injective on $\alloc{\asform_1 * \asform_2}$, hence since $\astore$ is injective on $V$, we deduce that $\astore$ is injective on $V \cup \alloc{\asform_{3-i}}$.
Since $(\astore,\aheap_i)$ cannot be a counter-model of the premises because $\aheap_i \subset \aheap$, this entails that 
$(\astore,\aheap_i) \modelsr  \asformB_i \swedge \apformB_i$ for $i = 1,2$, thus 
$(\astore,\aheap) \modelsr (\asformB_1 * \asformB_2) \swedge (\apformB_1 \wedge \apformB_2)$.
\end{proof}

\subsection{Axioms and Anti-Axioms}

We define two sets of syntactic criteria on sequents that  allow to quickly conclude that such sequents are respectively valid or non-valid. This will be useful to block the application of the inference rules in these cases. Axioms (i.e., necessarily valid sequents) are defined as follows.
\begin{definition}
\label{def:axiom}
 An {\em axiom} is a sequent that is of one of the following forms modulo AC:
 \begin{enumerate}
 \item{$\asform \swedge (\apform \wedge \apform') \vdashr \asform \swedge \apform$;}
 \item{$\asform \swedge (\apform \wedge x \not \iseq x) \vdashr \asheapB$;}
 \item{$\asform \swedge \apform \vdashr \asheapB$ where $\asform$ is \heapunsat;}
 \item{$\asform \swedge \apform \vdashr \asheapB$ where either $\alloc{\asform} \cap V \not = \emptyset$ or $V$ contains two occurrences of the same variable.}
 \end{enumerate}
\end{definition}
Intuitively, a sequent is valid  if the right-hand side is a trivial consequence of the left-hand side, if the left-hand side is (trivially) unsatisfiable, or if $V$ contains a variable that is allocated by the left-hand side or two occurrences of the same variable (since by hypothesis counter-models must be injective on $V$).
\begin{lemma}
\label{lem:ax}
All axioms are valid.
\end{lemma}
\begin{proof}
We consider each case separately (using the same notations as in Definition~\ref{def:axiom}):
\begin{enumerate}
\item{It is clear that every model  of  $\asform \swedge (\apform \wedge \apform')$ is a model of 
$\asform \swedge \apform$.}
\item{By definition, $\asform \swedge (\apform \wedge x \not \iseq x)$ has no model, hence
$\asform \swedge (\apform \wedge x \not \iseq x) \vdashr \asheapB$ has no counter-model.}
\item{If $\asform$ is \heapunsat then $\alloc{\asform}$ contains two occurrences of the same variable, which by Corollary \ref{cor:heapunsat}, entails that $\asform$ has no model. Thus
$\asform \swedge \apform \vdashr \asheapB$ has no counter-model.}
\item{
Let $(\astore,\aheap)$ be a counter-model of $\asform \swedge \apform \vdashr \asheapB$. By definition $\astore$ is injective on $V$ hence we cannot have $\{ x,x \} \subseteqm V$.
Also, by definition, if $x\in V$ then we cannot have $\astore(x) \in \dom{\aheap}$, and if $x\in \alloc{\asform}$ then  if $x\in \alloc{\asform} \cap V$ then $\astore(x) \in \dom{\aheap}$ by Lemma~\ref{lem:alloc}. We conclude that it is impossible to have $x\in \alloc{\asform} \cap V$ either.
}
\end{enumerate}
\end{proof}

\newcommand{\antiaxiom}{anti-axiom\xspace}

We also introduce the notion of an {\antiaxiom}, which is a sequent satisfying some syntactic conditions that prevent it from being valid. 
\begin{definition}
\label{def:antiax}
A sequent $\asform \swedge \apform \vdashr \asformB \swedge \apformB$ is an {\em \antiaxiom} if it is not an axiom, $\apform$ contains no equality, $\apformB = \true$ and one of the following conditions holds:
\begin{enumerate}
\item{$\alloc{\asformB} \not \subseteq \alloc{\asform}$; \label{anti:alloc}}
\item{$\asformB = \emp$ and $\asform \not = \emp$.}
\item{There exists a variable $x \in \alloc{\asform} \setminus \alloc{\asformB}$, such that 
$y \not \pathto{\asform}^* x$ holds, for all $y \in \alloc{\asformB}$;\label{anti:connect}}
\item{$V \cap (\vars(\asform) \setminus \vars(\asformB))$ is not empty.\label{anti:V}
}
\item{
$\vars_{\addr}(\asform) \setminus (\vars(\asformB) \cup \alloc{\asform})$ is not empty. 
\label{anti:ref}}
\end{enumerate}
\end{definition}

We provide examples illustrating every case in Definition~\ref{def:antiax}:
\begin{example}
The following sequents (where $p$ is some arbitrary predicate) are {\antiaxiom}s:
\[ \begin{array}{llllcllll}
1: & p(x,y) & \vdashrV{\emptyset} &  p(y,x) & \quad &
2: & p(x,y) & \vdashrV{\emptyset} & \emp \\
3: & p(x,y) * p(z,y) & \vdashrV{\emptyset} & q(x,y) & &
4: & p(x,y) & \vdashrV{\{ y \}} & r(x) \\
5: & p(x,y) & \vdashrV{\emptyset} & r(x)
\end{array} 
\]
Intuitively, $1$ cannot be valid because there exist models of $p(x,y)$ in which $y$ is not allocated whereas $y$ is allocated in all models of $p(y,x)$ 
Note that by Assumption \ref{assume:productive}, all predicates are productive, hence $p(x,y)$ admits at least one model. 
Furthermore, a predicate cannot allocate any of its arguments other than the root, for instance rules of the form $p(x,y) \Leftarrow x \mapsto (y) * p(y,x)$, indirectly allocating $y$, are not allowed. 
$2$ cannot be valid because the models of $p(x,y)$ allocate at least $x$.
For $3$, assuming that all variables are associated with distinct locations, one can construct a model of $p(x,y) * p(z,y)$
in which there is no path from $x$ to $z$  and  by Proposition~\ref{prop:connect} all locations occurring in the heap of any model of $q(x,y)$ must be reachable from $x$.
For $4$ and $5$, we can construct a counter-model by considering any structure in which $y$ occurs in the heap but is not allocated, and by Lemma~\ref{lem:establish}, all the locations occurring in the heap of any model of $r(x)$ must be allocated.
\end{example}

To show that all {\antiaxiom}s admit counter-models, we use the following lemma, which will also 
play a key r\^ole in the completeness proof.
It states that all the formulas that are \heapsat admit a model satisfying some particular properties:
 \newcommand{\smallU}{U}
 
 \begin{lemma}
 \label{lem:sat}
 Let $\asform$ be a spatial formula, containing a variable $x$ of sort $\addr$.
 Let $\astore$ be a store that is injective on $\alloc{\asform}$.
 Let $\smallU$ be an infinite subset of $\univ_{\addr}$ such that ${\smallU \cap \astore(\vars_{\addr})} = \emptyset$. 
 If $\asform$ is \heapsat, then it admits an $\asid$-model of the form $(\astore,\aheap)$, where 
 $\astore(x) \in \locs{\aheap}$, the set $\dom{\aheap}$ is a subset of $\smallU \cup \astore(\alloc{\asform})$
 and ${\locs{\aheap} \subseteq \smallU \cup \astore(\vars(\asform))}$.
 Moreover, if $\astore$ is injective then $(\astore,\aheap)$ 
 is a \pathcompatible model of $\asform$.
 \end{lemma}
 \begin{proof}
 The proof is by induction on the formulas. Note that we cannot have $\asform = \emp$ since $x\in \vars(\asform)$.
  \begin{itemize}
 \item{If $\asform = y_0 \mapsto (y_1,\dots,y_n)$ then we set: $\aheap \isdef \{ (\astore(y_0),\dots,\astore(y_n)) \}$. By definition, ${(\astore,\aheap) \modelsr \asform}$. Moreover, $x \in \{ y_0,\dots,y_n\}$, and 
 	$\locs{\aheap} = \{ \astore(y_i) \mid i = 0,\dots,n \text{\ and\ } y_i \in \vars_{\addr} \}$, hence $\astore(x) \in \locs{\aheap}$ and
 $\locs{\aheap} \subseteq \astore(\vars(\asform))$.
 Furthermore, $\dom{\aheap} = \{ \astore(y_0) \}$ and $\alloc{\asform} = \{ y_0 \}$ thus
 $\dom{\aheap} \subseteq \astore(\alloc{\asform})$.
 Finally, assume that $\astore$ is injective. We have by definition $\connect{\aheap} = \{ (\astore(y_0),\astore(y_i)) \mid i = 1,\dots,n \}$, thus if $\astore(u) \connect{\aheap}^* \astore(v)$ for $u,v \in \vars(\asheap)$ then 
 we must have either $\astore(u) = \astore(v)$, so that $u = v$ because $\astore$ is injective, in which case it is clear that $u \pathto{\asform}^* v$;
 or $\astore(u) = \astore(y_0)$
 and $\astore(v) = \astore(y_i)$ for some $i = 1,\dots,n$.
 Since $\astore$ is injective this entails that $u = y_0$, $v = y_i$, thus 
 $u \pathto{\asform} v$ by definition of $\pathto{\asform}$. 
 } 
 \item{Assume that $\asform = p(y_0,\dots,y_n)$ is a predicate atom. 
 Then, since by Assumption \ref{assume:productive} every predicate symbol is productive, there exists a symbolic heap $\asheapB$ such that
 $\asform \unfold \asheapB$.
If $x = y_0$ then since the rules in $\asid$ are \prules, $\asheapB$ contains a points-to-atom with root $x$. Otherwise, by Assumption \ref{assume:no_useless_variable}, $x = y_i$ for some $i \in \outparam{p}$, hence there exists a rule application
 $\asform \unfold \asheapB$ such that $x$ occurs in some predicate atom in $\asheapB$. Thus in both cases we may assume that $x$ occurs in a spatial atom in $\asheapB$. 
 Note that $\asheapB$ must be \heapsat, since all considered rules are \prules and by Definition~\ref{def:simple} the roots of the predicate symbols in $\asheapB$ are pairwise distinct existential variables, thus also distinct from the root $y_0$ of the points-to atom.
Furthermore, $\asheapB$ is of the form $\asformB \swedge \apformB$, where $x \in \vars(\asformB)$ and 
 $\apformB$ is a conjunction of disequations $u \not \iseq v$, with $u \not = v$.
 
 Let $\astore'$ be an \extension{\astore}{\vars(\asheapB) \setminus \vars(\asform)} mapping the variables in $\vars(\asheapB) \setminus \vars(\asform)$ to pairwise distinct locations in $\smallU$. Since by hypothesis  $U \cap \astore(\vars) = \emptyset$, $\astore'$ is injective on any set of the form $E \cup (\vars(\asheapB) \setminus \vars(\asform))$ when $\astore$ is injective on $E$.
 Let $\smallU' \isdef \smallU \setminus \astore'(\vars(\asheapB) \setminus \vars(\asform))$.
 By the induction hypothesis, there exists a heap $\aheap$ such that $(\astore',\aheap) \modelsr \asformB$,  with $\astore'(x) \in \locs{\aheap}$, ${\dom{\aheap} \subseteq \smallU' \cup \astore'(\alloc{\asformB})}$
 and $\locs{\aheap} \subseteq \smallU' \cup \astore'(\vars(\asformB))$.
 Now if $\astore$ is injective then (as $\astore'$ is also injective in this case), 
$(\astore',\aheap)$ is a \pathcompatible \rmodel of $\asformB$.
We show that $(\astore,\aheap)$ fulfills all the properties of the lemma.

\begin{itemize}
\item{ 
  Since $\astore'$ maps the variables in $\vars(\asheapB) \setminus \vars(\asform)$ to pairwise distinct locations in $\smallU$ and ${\astore(\vars) \cap \smallU}$ is $\emptyset$, necessarily $\astore' \modelsr \apformB$, thus $(\astore',\aheap) \modelsr \asheapB$ which entails that $(\astore,\aheap) \modelsr \asform$. We also have $\astore(x) = \astore'(x) \in \locs{\aheap}$.
}

\item{
Let $\ell \in \locs{\aheap}$. We show that $\ell \in \smallU \cup \astore(\vars(\asform))$.
 By the induction hypothesis, we have $\ell \in \smallU' \cup \astore'(\vars(\asformB))$.
If $\ell \in \smallU' \subseteq \smallU$ then the proof is completed, otherwise we have $\ell = \astore'(y)$ for some $y \in \vars(\asformB)$. If $y \in \vars(\asform)$ then $\astore(y) = \astore'(y)$ thus 
$\ell \in \astore(\vars(\asform))$.
Otherwise, we must have $y \in  \vars(\asheapB) \setminus \vars(\asform)$, thus $\astore'(y) \in \smallU$ by definition of $\astore'$ and the result holds.
}
\item{
Let $\ell \in \dom{\aheap}$, we show that $\ell \in \smallU \cup \astore(\alloc{\asform})$.
By the induction hypothesis we have $\ell \in \smallU' \cup \astore'(\alloc{\asformB})$.
If $\ell \in \smallU' \subseteq \smallU$ then the proof is completed.
Otherwise, $\ell = \astore'(y)$ with $y \in \alloc{\asformB}$.
  Since the rules in $\asid$ are \prules, $y$  is either the root $y_0$ of $\asform$, in which case we have $y \in \alloc{\asform}$ and $\astore(y) = \astore'(y)$, thus the result holds;
  or $y$ occurs in 
$\vars(\asheapB) \setminus \vars(\asform)$, in which that we have $\astore'(y)\in \smallU$, by definition of $\astore'$. 
}
\item{
 There  remains to show that $(\astore,\aheap)$ is a \pathcompatible \rmodel of $\asform$, in the case where $\astore$ is injective.
 Assume that $\astore(u) \connect{\aheap}^* \astore(v)$. If $\astore(u) = \astore(v)$ then  $u = v$ by injectivity of $\astore$, hence $u \pathto{\aheap}^* v$.
 Otherwise, we must have $\{ \astore(u),\astore(v) \} \subseteq \locs{\aheap} \subseteq \smallU' \cup \astore'(\vars(\asformB))$, and since $\smallU' \subseteq \smallU$ and $\smallU \cap \astore(\vars) = \emptyset$, we have
 $\{ \astore(u),\astore(v) \} \subseteq \astore'(\vars(\asformB))$.
We also have $\astore'(\vars(\asheapB) \setminus \vars(\asform)) \subseteq \smallU$, which entails that $\{ \astore(u),\astore(v) \} \subseteq \astore'(\vars(\asform))$. 
Since $\astore'$ is an \extension{\astore}{\vars(\asheapB) \setminus \vars(\asform)}, $\astore$ and $\astore'$ coincide on all variables in $\vars(\asform)$ and we deduce that
$\{ \astore(u),\astore(v) \} \subseteq \astore(\vars(\asform))$. 
Because $\astore$ is injective, this entails that 
$u,v \in \vars(\asform)$, so that
$\astore(u) = \astore'(u)$ and
 $\astore(v) = \astore'(v)$. 
 By hypothesis $(\astore',\aheap)$ 
 	 is a \pathcompatible \rmodel of $\asformB$, and we deduce that $u \pathto{\asformB}^* v$. 
 	 Since $u \not = v$, necessarily $u \in \alloc{\asformB}$ (by definition of $\pathto{\asformB}$), and since the rules in $\asid$ are \prules, and $u \in \vars(\asform)$, this entails that $u = \roots{\asform}$. 
 Since $v \in \vars(\asform)$,  by Assumption \ref{assume:no_useless_variable} we have
 $u \pathto{\asform}^* v$.
 }
 \end{itemize}
 }
 
 \item{
 Assume that $\asform  = \asform_1 * \asform_2$, with $\asform_i \not = \emp$. Let $\smallU_1,\smallU_2$ be disjoint infinite subsets of $\smallU$.
 Let $\{ x_1,x_2\}$ be some arbitrary chosen variables such that $x_i \in \vars(\asform_i)$ and $x\in \{ x_1,x_2\}$ (it is easy to check that such a pair of variables always exists).
 By the induction hypothesis, there exist
 heaps $\aheap_i$ such that $(\astore,\aheap_i) \models \asform_i$  
 where 
 $\astore(x_i) \in \locs{\aheap_i}$, ${\dom{\aheap_i} \subseteq \smallU_i \cup \astore(\alloc{\asform_i})}$
 and $\locs{\aheap_i} \subseteq \smallU_i \cup \astore(\vars(\asform_i))$.
 Moreover, if $\astore$ is injective, then $(\astore,\aheap_i)$ is an \pathcompatible $\asid$-model of $\asform_i$.
 We first show that $\aheap_1$ and $\aheap_2$ are disjoint.
 Assume for the sake of contradiction that $\ell \in \dom{\aheap_1} \cap \dom{\aheap_2}$. Since $\smallU_1 \cap \smallU_2 = \emptyset$, necessarily $\ell = \astore(y_i)$ (for $i =1,2$), with $y_i \in \alloc{\asform_i}$. Since 
 $\astore$ is injective on $\alloc{\asform}$, we deduce that $y_1 = y_2$. We have $\{ y_1,y_2 \} \subseteqm \alloc{\asform}$, hence $\asform$ is \heapunsat, which contradicts the hypotheses of the lemma.
 Thus $\aheap_1$ and $\aheap_2$ are disjoint. 
 
 Let $\aheap = \aheap_1 \union \aheap_2$.
 We have $\astore(x) \in \{ \astore(x_1),\astore(x_2) \} \subseteq \locs{\aheap_1} \cup \locs{\aheap_2} = \locs{\aheap}$.
 Moreover, $\dom{\aheap} = \dom{\aheap_1} \cup \dom{\aheap_2} \subseteq \smallU_1 \cup \smallU_2 \cup \astore(\alloc{\asform_1}) \cup \astore(\alloc{\asform_2}) {\subseteq \smallU \cup \astore(\alloc{\asform})}$,
 and $\locs{\aheap} = \locs{\aheap_1} \cup \locs{\aheap_2} \subseteq \smallU_1\cup \smallU_2 \cup \astore(\vars(\asform_i)) \cup \astore(\vars(\asform_2))  {\subseteq \smallU \cup \astore(\vars(\asform))}$.
Furthermore, $(\astore,\aheap) \modelsr \asform_1 * \asform_2 = \asform$.

There only remains to prove that $(\astore,\aheap)$ is a \pathcompatible \rmodel of $\asform$ when $\astore$ is injective.
Assume that this is not the case, and let $u,v$ be variables such that $\astore(u) \connect{\aheap}^* \astore(v)$ and 
$u \not \pathto{\asform}^* v$. This entails that $u \not = v$.
By definition, 
there exist $\ell_0,\dots,\ell_m$ such that 
$\ell_0 = \astore(u)$,
$\ell_m = \astore(v)$,
and $\forall i = 1,\dots,m,\, \ell_{i-1} \connect{\aheap} \ell_i$.
We assume, w.l.o.g., that $m$ is miminal, i.e., that there is no sequence
$\ell_0',\dots,\ell_k'$ and no variables $x_0,x_k$ such that $k <m$, $\ell_0' = \astore(x_0)$,
$\ell_k' = \astore(x_k)$ and $x_0 \not \pathto{\asform}^* x_k$.
We may also assume, by symmetry, that $\ell_0 \in \dom{\aheap_1}$.
If all the locations $\ell_1,\dots,\ell_{m-1}$  occur in $\dom{\aheap_1}$ then 
$\astore(u) \connect{\aheap_1}^* \astore(v)$, thus
$u \pathto{\asform_i}^* v$ because $(\astore,\aheap_i)$ is an \pathcompatible $\asid$-model of $\asform_i$, which entails that 
$u \pathto{\asform}^* v$ since by definition $\pathto{\asform_i} \subseteq \pathto{\asform}$, contradicting our assumption.
Otherwise, let $j$ be the smallest index  in $\ell_1,\dots,\ell_{m-1}$ such that $\ell_j \not \in \dom{\aheap_1}$.
Since $\ell_j \in \dom{\aheap}$ (as $\ell_{j} \connect{\asheap} \ell_{j+1}$) we deduce that 
$\ell_j \in \dom{\aheap_2} \subseteq \locs{\aheap_2}$, and $\ell_j \in \locs{\aheap_1}$. Since $\smallU_1 \cap \smallU_2 = \emptyset$, we get $\ell_j = \astore(\vars(\asform_i))$, for some $i = 1,2$.
Thus $\ell_j = \astore(u')$ with $u'\in \vars(\asform)$.
Since $u \not \pathto{\asform}^* v$, we have by transitivity either
$u \not \pathto{\asform}^* u'$
or
$u' \not \pathto{\asform}^* v$,
 which contradicts the minimality of $m$ (one of 
the sequences 
$\ell_1,\dots,\ell_j$ or
$\ell_j,\dots,\ell_m$ satisfies the conditions above and has a length strictly less than $m+1$).
  }
 \end{itemize}
 \end{proof}

\begin{lemma}
\label{lem:antiax}
All {\antiaxiom}s admit counter-models.
\end{lemma}
\begin{proof}
Consider an \antiaxiom $\asform \swedge \apform \vdashr \asformB \swedge \apformB$, with the same notations as in Definition~\ref{def:antiax}.
Let $\astore$ be an injective store such that $\univ_{\addr} \setminus \astore(\vars)$ is infinite, such a store always exists since $\univ_{\addr}$ is infinite. 
First assume that $\asform \not = \emp$. In this case $\asform$ necessarily contains at least one spatial atom $\anatom$, hence at least one variable $x = \rootof{\anatom}$ of sort $\addr$.
Since $\asform \swedge \apform \vdashr \asformB \swedge \apformB$ is not an axiom,
$\asform$ is \heapsat, thus by Lemma~\ref{lem:sat} (applied with $\smallU \isdef \univ_{\addr} \setminus \astore(\vars)$), $\asform \swedge \apform$ admits a \pathcompatible model $(\astore,\aheap)$, 
such that $\dom{\aheap} \subseteq \smallU \cup \astore(\alloc{\asform})$.
Since $\smallU \cap \astore(\vars) = \emptyset$, we have, for all $u\in \vars$, $\astore(u) \in \dom{\aheap} \iff \astore(u)\in \astore(\alloc{\asform})$. Since $\astore$ is injective, we deduce that 
\[\astore(u) \in \dom{\aheap}  \iff u\in \alloc{\asform}\quad (\dagger).\]
Note that if $\asform = \emp$, then the structure $(\astore,\aheap)$ with $\aheap = \emptyset$ also satisfies $(\dagger)$, since $\pathto{\aheap}$, $\connect{\asform}$, $\dom{\aheap}$ and $\alloc{\asform}$ are all empty in this case.

We show that $(\astore,\aheap)$ is a counter-model of  $\asform \swedge \apform \vdashr \asformB \swedge \apformB$.
Since $\astore$ is injective, in particular $\astore$ must be injective on $V$, because otherwise $V$ would contain
two occurrences of the same variable, hence $\asform \swedge \apform \vdashr \asformB \swedge \apformB$ would be an axiom, contradicting Definition~\ref{def:antiax}.
If $V$ contains a variable $y$ such that $\astore(y) \in \dom{\aheap}$ then $(\dagger)$
entails that $y \in \alloc{\asform}$, which is impossible since $\asform \swedge \apform \vdashr \asformB \swedge \apformB$  would then be an axiom.
We prove that $(\astore,\aheap) \not \modelsr \asformB \swedge \apformB$ by considering each case from Definition~\ref{def:antiax} separately.
\begin{enumerate}
\item{If $\alloc{\asformB} \setminus \alloc{\asform}$ contains a variable $y$, then by $(\dagger)$ we have $\astore(y) \not \in \dom{\aheap}$, which by Lemma~\ref{lem:alloc} entails that 
$(\astore,\aheap) \not \modelsr \asformB$.}

\item{If $\asformB = \emp$ and $\asform \not = \emp$, then  $\asform$ contains at least one atom, hence $\alloc{\asform} \not = \emptyset$.
By Lemma~\ref{lem:alloc}, $\astore(\alloc{\asform}) \subseteq \dom{\aheap}$, thus $\aheap \not = \emptyset$
and  $(\astore,\aheap) \not \modelsr \asformB \swedge \apformB$.}

\item{Assume that 
there exists a variable $x' \in \alloc{\asform} \setminus \alloc{\asformB}$ such that 
$y \not \pathto{\asform} x'$, for all ${y \in \alloc{\asformB}}$, and that 
$(\astore,\aheap)\modelsr \asformB$.
By Lemma~\ref{lem:alloc} we have $\astore(x') \in \dom{\aheap} \subseteq \locs{\aheap}$, thus 
by Proposition~\ref{prop:connect} there exists $y\in \alloc{\asformB}$ such that 
$\astore(y) \connect{\aheap}^* \astore(x)$.
Since $(\astore,\aheap)$ is a \pathcompatible model 
of $\asform$ necessarily $y \pathto{\asform}^* x'$, which contradicts the above assumption.}

\end{enumerate}
For the remaining cases, we will apply Lemma~\ref{lem:sat} to obtain a heap, using a variable $x$ that is in $V \cap (\vars(\asform) \setminus \vars(\asformB))$ (for Condition \ref{anti:V}) or in $\vars_{\addr}(\asform) \setminus (\vars(\asformB) \cup \alloc{\asform})$ (for Condition \ref{anti:ref}).  Note that by Lemma~\ref{lem:sat} this entails in particular that $\astore(x) \in \locs{\aheap}$.
\begin{enumerate}
\setcounter{enumi}{3}
\item{Assume that $x \in V \cap \vars(\asform)$ and $x\not \in \vars(\asformB)$.
Then $\astore(x) \in \locs{\aheap}$ and 
since $x \in V$ we get ${\astore(x) \not \in \dom{\aheap}}$. By Lemma  \ref{lem:establish}, 
if $(\astore,\aheap) \modelsr \asformB$ then we have $\astore(x) \in \astore(\vars(\asformB))$, and $x \in \vars(\asformB)$ since $\astore$ is injective. Thus we get a contradiction.
}
\item{Assume that $x\in \vars_{\addr}(\asform)$, $x \not \in \vars(\asformB)$ and
$x \not \in \alloc{\asform}$. By $(\dagger)$, we deduce that $\astore(x) \not \in \dom{\aheap}$.
Assume that $(\astore,\aheap) \modelsr \asformB$.
Since $x \not \in \vars(\asformB)$ and $\astore$ is injective, we get $\astore(x) \not \in \astore(\vars(\asformB))$. By Lemma~\ref{lem:establish}, we deduce that $\astore(x) \not \in \locs{\aheap}$, a contradiction.}
\end{enumerate}
\end{proof}

\subsection{Proof Trees}

\newcommand{\fexpanded}{fully expanded\xspace}

 \newcommand{\descendant}{descendant\xspace}
  \newcommand{\direct}{\imi-free\xspace}
 
\newcommand{\ireducible}{\imi-reducible\xspace}
 \newcommand{\qireducible}{quasi-\ireducible}
\newcommand{\varmap}[1]{{\cal V}_{\mapsto}(#1)}

A {\em proof tree} is a (possibly infinite) tree with nodes labeled by sequents, such that if 
a node labeled by $\aseq$ has successors labeled by 
$\aseq_1,\dots,\aseq_n$ then there exists a rule instance 
of the form $\frac{\aseq_1 \dots \aseq_n}{\aseq}$. 
We will usually identify the nodes in a proof tree 
with the sequents labeling them.
A {\em path} from $\aseq$ to $\aseq'$ in a proof tree is a finite sequence $\aseq_0,\dots,\aseq_n$ such that $\aseq = \aseq_0$, $\aseq_n = \aseq'$ and for all $i = 1,\dots,n$, $\aseq_i$ is a successor of $\aseq_{i-1}$.
A proof tree is {\em \fexpanded} 
if all its leaves are axioms.
It is {\em rational} if it contains a finite number of subtrees, up to a renaming of variables. Note that rational trees may be infinite, but they can be represented finitely.
A sequent $\aseq'$ is a {\em \descendant} of a sequent $\aseq$
if there exists a proof tree with a path from $\aseq$ to $\aseq'$.



\begin{example}
Let $\asid$ be the following set of rules, where $\cstA,\cstB$ denote constant symbols and $u$ is a variable of the same sort as $\cstB$:
\[
\begin{tabular}{ccc}
$p(x,y)$ & $\Leftarrow$ & $x \mapsto (\cstA,y,z) * p(z,y)$ \\
$p(x,y)$ & $\Leftarrow$ & $x \mapsto (\cstB)$ \\
$r(x)$ & $\Leftarrow$ & $x \mapsto (\cstA,y,z) * r(z)$ \\
$r(x)$ & $\Leftarrow$ & $x \mapsto (u)$
\end{tabular}
\] 
The sequent $p(x,y) \vdashrV{\emptyset} r(x)$ admits the following rational proof tree (the sequent $p(z,y) \vdashrV{\emptyset} r(z)$ is identical to the conclusion up to a renaming):
\begin{prooftree}
\AxiomC{}
  \RightLabel{axiom}
\UnaryInfC{$x \mapsto (\cstA,y,z) \vdashrV{\{ z \}} x \mapsto (\cstA,y,z)$}
\AxiomC{$p(z,y) \vdashrV{\emptyset} r(z)$}
  \RightLabel{$\eli$}
\UnaryInfC{$p(z,y) \vdashrV{\{ x \}} r(z)$}
  \RightLabel{$\sep$}
\BinaryInfC{$x \mapsto (\cstA,y,z) * p(z,y) \vdashrV{\emptyset} x \mapsto (\cstA,y,z) * r(z)$}
  \RightLabel{$\imi$}
\UnaryInfC{$x \mapsto (\cstA,y,z) * p(z,y) \vdashrV{\emptyset} r(x)$}
\AxiomC{}
  \RightLabel{axiom}
\UnaryInfC{$x \mapsto (\cstB) \vdashrV{\emptyset} x \mapsto (\cstB)$}  
  \RightLabel{$\imi$}
\UnaryInfC{$x \mapsto (\cstB) \vdashrV{\emptyset} r(x)$}  
  \RightLabel{$\unf$}
  \BinaryInfC{$p(x,y) \vdashrV{\emptyset} r(x)$}
\end{prooftree}
\end{example}

\begin{remark}
Note that, since infinite proof trees are allowed, the fact that each rule is sound does not imply that the procedure itself is sound,
i.e., that the root of every \fexpanded proof tree  is valid.
The latter property holds  only for the strategy introduced in the next section, see Theorem~\ref{theo:sound}.
\end{remark}

\subsection{The Strategy}

\label{sect:strat}

\newcommand{\admissible}{admissible\xspace}


 \newcommand{\narrowvar}[1]{{\cal V}^\dagger(#1)}
 \newcommand{\narrow}{narrow\xspace}
 
In this section, we introduce a strategy to restrict the application of the inference rules, which will ensure both the soundness and efficiency of the proof procedure. To this purpose, we define a set of variables $\narrowvar{\aseq}$ which denotes the variables occurring at non-root positions on the right-hand side of the considered sequent, but not in the set of non-allocated  variables $V$. As we shall see, the strategy will handle the sequents in different ways depending on the number of such variables. Recall that $\maxr{\asid}$ denotes the maximal arity of the predicate symbols and tuples occurring in $\asid$.
\begin{definition}
 For every sequent $\aseq = \asheap \vdashr \asheapB$, we denote by 
 $\narrowvar{\aseq}$ the set $\vars_{\addr}(\asheapB) \setminus (V \cup \alloc{\asheapB})$.
 A sequent $\aseq$ is {\em \narrow} if it is \eqfree and $\card{\narrowvar{\aseq}} \leq \maxr{\asid}$.
 \end{definition}
Note that in particular, if $\asheapB$ is a spatial atom, then $\aseq$ is necessarily \narrow since in this case we must have $\card{\vars(\asheapB)} \leq \maxar{\asid} \leq \maxr{\asid}$. 
The inference rules are applied with the following strategy:

\begin{definition}
\label{def:strat}
\label{def:admissible} 
We assume a {\em selection function} is given, mapping every non\-empty finite set of expressions (i.e., formulas or rule applications) $S$ to an element of $S$; this element is said to be \hbox{\em selected in $S$}. We also assume that this function can be computed in polynomial time.
A proof tree is {\em \admissible} if all the rule applications occurring in it fulfill the following conditions 
(see Figure \ref{fig:rules} for the notations):
  \begin{enumerate}
  \item{No rule is applied on an axiom or an an \antiaxiom, and no rule is applied if one of the premises is an \antiaxiom. \label{strat:noax}}

 \item{\imi\ is applied only if $\apformB = \true$ and $\asformB = \emp$. \label{strat:ini}}
 \item{\unf\ is applied only if  
 $\asheapB$ is of one of the forms $q(x_1,\vec{y}) \swedge \true$ or $x_1 \mapsto \vec{y} \swedge \true$, where $x_1$ is the first component of the vector $\vec{t}$ (as defined in the rule). \label{strat:unf}}
 \item{\dec\ is applied only if $\asheapB$ is of the form $\asformB \swedge \true$, 
 $x \in \alloc{\asform}$, $y \in \vars(\asheapB) \setminus (\alloc{\asform} \cup V)$, 
 $\asformB$ is a predicate atom 
and $x \not \iseq y$ does not already occur in $\apform$. \label{strat:dec}}
 \item{\sep\ is applied only if $\apform_1 = \apform_2$, $\apformB_1 = \apformB_2 = \true$,
 and $\asformB_1$ is selected in the set of atoms in 
  $\asformB_1 * \asformB_2$. 
 Furthermore,
if the conclusion is not \narrow, then 
the rule is applied only if the left-premise  is valid 
(this will be tested by applying the decision procedure recursively).\label{strat:sep}}
 \item{The rules are applied with the following priority (meaning that no rule can be applied if a rule with a higher priority is also applicable):  \label{strat:priority}
 	\[\wea\ > \eli\ > \eq\ > \noteq\  >  \sep\ >  \unf\ >  \dec\ > \imi.\]
	}
 \item{For all the rules other than \sep, and for all applications of $\sep$ on a sequent that is not \narrow, if several applications of the rule are possible and all fulfill the conditions above, then  
 	only the one that is selected  in the set of possible rule applications can be applied. \label{strat:nocare} }
 
 \end{enumerate}
\end{definition}

From now on, we assume that all the considered proof trees are \admissible (the soundness\footnote{The soundness of each individual rule, as stated by Lemmata \ref{lem:rules} and \ref{lem:sep_sound}, does not depend on Definition~\ref{def:admissible}, but the global soundness of the procedure, as stated in Theorem~\ref{theo:sound} holds only for \admissible proof trees.} and completeness proofs below, as well as the complexity analysis hold with this requirement).
No assumption is made on the selection function, for instance one may assume that $\asformB_1$ is the leftmost atom in $\asformB_1 * \asformB_2$ 
(i.e., the atoms are ordered according to the order in which they occur in the formula) and that the first detected rule application is applied, except for $\sep$ if the conclusion is  \narrow.

Many of the conditions in Definition~\ref{def:strat} are quite natural, and simply aim at pruning the search space by avoiding irrelevant rule applications. For instance, applying a rule on an axiom or on an \antiaxiom is clearly useless. The priority order is chosen in such a way that the rules with the minimal branching factor are applied first, postponing the most computationally costly rules. Similarly, the rules operating on spatial atoms ($\imi$ or $\unf$) are postponed until all pure formulas have been handled. The unfolding of the atoms on the left-hand side of a sequent (Rule $\unf$) is postponed until 
the right-hand side has been fully decomposed (using rule \sep), yielding a unique spatial atom. This permits to guide the choice of the atom to be unfolded on the left-hand side: we only  unfold the atom with the same root as the one on the right-hand side.
Rule $\dec$ may cause a costly explosion if applied blindly, thus the application conditions are carefully designed. They are meant to ensure that the rule application is really useful, in the sense that it permits further applications of rule $\imi$.
Condition \ref{strat:nocare} in Definition~\ref{def:strat}
is meant to ensure that the considered rules are applied with a ``don't care'' strategy, 
meaning that if several rule applications are possible then one of them (the selected one) is chosen arbitrarily and the others are ignored. This is justified by the fact that these rules are invertible, hence exploring all possibilities is useless. Such a strategy is crucial for proving that the procedure runs in polynomial time.
In contrast, if the conclusion of $\sep$  is \narrow, then the rule must applied with a ``don't know'' indeterminism, i.e., all possible applications must be considered. This is necessary for completeness, due to the fact that the rule is {\em not} invertible in this case. Of course, all these possible rule applications must be taken into account for the complexity analysis.
In other words, one must cope with two different kinds of branching: an ``and-branching'' due to the fact that a given rule may have several premises, and an ``or-branching'' due to the fact that there may be several ways of applying a given rule on a given sequent. By Condition \ref{strat:nocare}, the latter branching occurs only when the rule $\sep$ is applied on  a \narrow sequent; in all other cases, only the selected rule application can be considered hence there is no or-branching.

The requirement that $\asformB_1$ must be selected in Condition \ref{strat:sep} ensures that only one decomposition $\asformB_1 * \asformB_2$ can be considered for the right-hand side of a given sequent. This avoids for instance the exponential blow-up that would occur if the same decomposition was performed in different orders.
Note that, to check whether the left premise is valid in Condition \ref{strat:sep}, it is necessary  to  recursively invoke the proof procedure. As we shall see below, this is feasible because this premise and all its {\descendant}s are necessarily \narrow: 


\begin{proposition}
\label{prop:narrow}
Let $\aseq_1$ be an \eqfree sequent. For every successor $\aseq_2$ of $\aseq_1$, the relation $\narrowvar{\aseq_2} \subseteq \narrowvar{\aseq_1}$ holds.
Consequently, if a sequent is \narrow then all its {\descendant}s 
are \narrow.
\end{proposition}
\begin{proof}
Let $\aseq_i = \asheap_i \vdashrV{V_i} \asheapB_i$ for $i = 1,2$.
Assume for the sake of contradiction that $\narrowvar{\aseq_2} \setminus \narrowvar{\aseq_1}$ contains a variable $u$.
By definition, $u\in \vars(\asheapB_2)$ and $u\not \in \alloc{\asheapB_2} \cup V_2$.
First assume that $u\not \in \vars(\asheapB_1)$, i.e., that $u$ was introduced by the application of an inference rule to $\aseq_1$.
By Proposition~\ref{prop:rule-facts} (\ref{it:intro:eq}), the only rule that can introduce new variables to the right-hand side of a sequent
is {\imi}. Indeed, since we assume that all inferences are \admissible, by Condition \ref{strat:dec} of Definition~\ref{def:admissible}, $\dec$ must replace a variable $x$ by a variable $y$ occurring in $\vars(\asheapB_1)$. 
Rule \imi\ replaces a predicate atom $p(x,\vec{z})$ that occurs in $\asheapB_1$ by a formula of the form  $\asheapB_1'\sigma$, where 
$p(x,\vec{z}) \unfold \asheapB_1'\swedge \apformB'$.
Since the rules in $\asid$ are \prules, by Condition \ref{it:prog2} of Definition~\ref{def:simple}, all the variables occurring in $\asheapB_1'\sigma$ but not in $p(x,\vec{z})$ must occur as roots in $\asheapB_1'\sigma$. This entails that $\vars(\asheapB_2) \setminus \vars(\asheapB_1) \subseteq \alloc{\asheapB_2}$, which contradicts the fact that $u\not \in \alloc{\asheapB_2} \cup V$.
We deduce that $u \in \vars(\asheapB_1)$.

Assume that $u\in V_1$. 
Since $u\notin V_2$, the inference rule applied to $\aseq_1$ must have deleted a variable from $V_1$.
The only rules that can delete variables from $V_1$ are \eli\ and {\dec} by Proposition~\ref{prop:rule-facts} (\ref{it:rem:v}). Rule \eli\  applies only on a variable $x\not \in \vars(\asheapB_1)$, thus we cannot have $x = u$.
If Rule \dec\ is applied and replaces a variable $x$ by $y$, then $x$ cannot occur in $\aseq_2$, 
hence $u \not = x$. Thus we must have $u \not \in V_1$.

Finally assume that $u\in \alloc{\asheapB_1}$. 
Since $u \notin \alloc{\asheapB_2}$, the inference rule applied to $\aseq_1$ must have deleted a variable from $\alloc{\asheapB_1}$.
The only rules that can delete variables from $\alloc{\asheapB_1}$ are 
\sep\ and {\dec} by Proposition~\ref{prop:rule-facts} (\ref{it:rem:alloc}). 
Again, if rule \dec\ replaces a variable $x$ by $y$, then $x$ cannot occur in $\aseq_2$, hence $u \not = x$.
Now, consider  rule {\sep} and let 
$\aseq_2' \isdef \asheap_2' \vdashrV{V_2'} \asheapB_2'$ be the other premise of the rule.
Since $u\in \alloc{\asheapB_1}$ and $\aseq_1$ is not an \antiaxiom, we must have
$u\in \alloc{\asheap_1}$, which entails by definition of the rule
that $u\in \alloc{\asheap_2} \cup V_2$ and
$u\in \alloc{\asheap_2'} \cup V_2'$. 
Still by definition of rule {\sep}, if a variable is the root of a spatial atom occurring on the right-hand side of the conclusion $\asheapB_1$, then this spatial atom must occur on the right-hand side of one of the premises $\asheapB_2$ or $\asheapB_2'$. This entails that $u \in \alloc{\asheapB_2} \cup \alloc{\asheapB_2'}$.
Furthermore, we must have
$\alloc{\asheap_2} \cap \alloc{\asheap_2'} = \emptyset$ because otherwise $\aseq_1$ would contain 
two atoms with the same root, hence would be an axiom.
Since $u \not \in V_2$ we deduce that $u\in \alloc{\asheap_2}$, and that 
$u\in V_2'$.
If $u \not \in \alloc{\asheapB_2}$ then necessarily
$u \in \alloc{\asheapB_2'}$. Since $\aseq_2'$ is not an \antiaxiom we deduce that
$u \in \alloc{\asheap_2'}$, which entails that $\aseq_2'$ is an axiom since $\alloc{\asheap_2'} \cap V_2'\not = \emptyset$, a contradiction.

The second part of the proposition follows by an immediate induction, using the fact that no rule can introduce any equality in its premises.
\end{proof}

\newcommand{\asuccessor}{auxiliary successor\xspace}


\vspace{-1ex}

We call {\em {\asuccessor}s} the sequents whose validity must be tested to check whether rule \sep\ is applicable or not, according to Definition~\ref{def:admissible}:

\begin{definition}
\label{def:asucc}
A sequent $\aseq$ is an {\em \asuccessor} of a sequent $\aseq'$ if:
\begin{itemize}
	\item $\aseq$ is not an \antiaxiom and is not \narrow,
	\item $\aseq'$ is not an axiom,
	\item $\aseq'$ and $\aseq$ are of the form $\asform' \swedge \apform \vdashrV{V'} \asformB' \swedge \true$
	and $(\asform' * \asform) \swedge \apform \vdashrV{V} (\asformB' * \asformB) \swedge \true$ respectively, where $\asformB'$ is a spatial atom, and $V' = V \cup \alloc{\asform}$.
\end{itemize}    
\end{definition}
If the sequent $\aseq$ in Definition~\ref{def:asucc} is valid, then it is the left premise of an application of \sep\ on $\aseq'$. The following proposition gives an upper-bound on the number of {\asuccessor}s.

 \newcommand{\anum}{\kappa}
\newcommand{\anumB}{\kappa}

\begin{proposition}
\label{prop:sepbranching}
Assume that 
$\maxar{\asid} \leq \anum$ 
(i.e., that the maximal arity of the predicates is bounded by $\anum$).
Then every sequent $\aseq$ has at most $2^\anum$ {\asuccessor}s, 
and each of these {\asuccessor}s can be computed in polynomial time.
\end{proposition}
\vspace{-1.5ex}
\begin{proof}
If $\aseq$ admits an \asuccessor, then $\aseq$ is necessarily of the form
 $\asform \wedge \apform \vdashr \asformB_1 * \asformB_2$ where $\asformB_1$ is the selected  predicate atom in $\asformB_1*\asformB_2$ according to Definition~\ref{def:admissible}. 
 To get an \asuccessor of~$\aseq$, one has to decompose $\asform$ into $\asform_1 * \asform_2$, in such a way that the obtained premises
 $\aseq_i = \asform_i \swedge \apform \vdashrV{V_i} \asformB_i$ (with $V_i = V \cup \alloc{\asform_{3-i}}$) are not {\antiaxiom}s. 
 For each atom $\anatom$ in $\asform$ such that $\rootof{\anatom} \in \vars(\asformB_1)$, we first choose whether $\anatom$ occurs in $\asform_1$ or $\asform_2$. 
There is at most one such atom for each variable, because otherwise $\asform$ would be \heapunsat and $\aseq$ would be an axiom, and $\asformB_1$ contains at most $\anum$ variables, thus there are at most $2^\anum$ possible choices. 

We show that, once this choice is performed, 
the decomposition $\asform = \asform_1 * \asform_2$ is also fixed. 
We denote by $E$ the set of variables $\vars(\asformB_1) \setminus \alloc{\asform_1}$. 
Let $\anatom$ be a predicate atom with root $x$ in $\asform$, and let $y_0$ be the root of $\asformB_1$.
 If $x = y_0$ then necessarily $\anatom$ occurs in $\asform_1$, since otherwise 
$x\in V_1$ and $\aseq_1$ would be an anti-axiom.
 If for all paths $y_0 \pathto{\asform} y_1 \pathto{\asform} \dots \pathto{\asform} y_n \pathto{\asform} x$, there exists $i = 1,\dots,n$ such that $y_i \in E$, then necessarily 
  $y \not \pathto{\asform_1}^* x$ and thus $\anatom$ cannot occur in $\asform_1$, as otherwise
 $\aseq_1$ would be an \antiaxiom.
We finally show that for every atom $\anatom'$ in $\asform$ with root $x'$, 
if there exists a path  $y_0 \pathto{\asform} y_1 \pathto{\asform} \dots \pathto{\asform} y_n \pathto{\asform} x'$ such that $\{ y_1,\dots,y_n,x' \} \cap E = \emptyset$, then 
  $\anatom'$ occurs in $\asform_1$ (thus, in particular, $\anatom$ occurs in $\asform_1$ 
  if the previous conditions are not satisfied, since $x\not \in E$, as $x \not \in \vars(\asformB_1)$). 
  We assume, w.l.o.g., that the considered path is the minimal one not satisfying the property, so that the atoms with roots $y_0,\dots,y_n$ all occur in 
  $\asform_1$. This entails that 
 $y_0 \pathto{\asform_1}^* x'$. If $x' \not \in \vars(\asformB_1)$ 
 then 
$\anatom$ must occur in $\asform_1$, otherwise $\aseq_1$ would be an \antiaxiom, because we would have $x'\in V_1$, as $x' \in \alloc{\asform_2}$.
Otherwise, since $x' \not \in E$, we have $x' \in \alloc{\asform_2}$, by definition of $E$.

To sum up, assuming that $\aseq  = \asform \wedge \apform \vdashr \asformB$ is neither an axiom nor an \antiaxiom, the {\asuccessor}s of $\aseq$ are computed as follows. We first compute the selected atom $\asformB_1$ in $\asformB$ (which can be done in polynomial time by the assumption in Definition \ref{def:strat}). 
The atoms $\anatom$ in $\asform$ are added either to $\asform_1$ or to $\asform_2$  using the following algorithm. Initially, $\asform_1$ and $\asform_2$ are both empty.
If $\rootof{\anatom} = \rootof{\asformB_1}$ then $\anatom$ is moved from $\asform$ to $\asform_1$ (there is exactly one atom with this property).
For each atom $\anatom$ in $\asform$ with  $\rootof{\anatom} \in \vars(\asformB_1)$ and $\rootof{\anatom} \not = \rootof{\asformB_1}$, we nondeterministically  add $\anatom$ to either $\asform_1$ or $\asform_2$ (which yields at most $2^\anum$ possible choices) and remove $\anatom$ from $\asform$. 
Then, for each atom $\anatom$ in $\asform_1$, all the remaining atoms $\anatom'$ in $\asform$ such that $\rootof{\anatom'} \in \vars(\anatom)$ are also moved from $\asform$ to $\asform_1$.
This rule is applied recursively until no such atom exists.
Afterwards, all the atoms still remaining in $\asform$ are added to $\asform_2$. 
It is clear that all these operations can be performed in polynomial time w.r.t.\ the size of $\aseq$ (in particular, the number of applications of the previous rule is bounded by the number of atoms in $\asform$).
 \end{proof}

\section{Properties of the Proof Procedure}

\label{sect:prop}

 \subsection{Soundness}

 We prove that the proof procedure is sound, in the sense that the root of every \fexpanded proof tree 
 fulfilling the conditions of Definition \ref{def:strat} is valid.
 As infinite proof trees are allowed, this does not follow 
 from the fact that all the rules are sound.
We show that every infinite branch contains infinitely many  applications of the rule $\sep$. As we shall see, this is sufficient to ensure that 
no counter-model exists, since otherwise the size of the smallest counter-model would be  decreasing indefinitely  along some infinite branch.
We recall that the rules are meant to be applied bottom-up: a rule is applicable on some sequent $\aseq$ if it admits an instance with conclusion $\aseq$, yielding some premises $\aseq_1,\dots,\aseq_n$.
We focus on sequents on which rule \imi\ is applied (which we call \ireducible), and we establish some useful properties.
\begin{definition}
\label{def:reduced}
A sequent is {\em \ireducible} if 
rule \imi\ can be applied on it; 
this entails that no other rule is applicable, as all other rules have priority over \imi.
A sequent $\aseq$ is {\em \qireducible} if it is of the form $(x \mapsto \vec{y} * \asform) \swedge \apform \vdashr \asformB \swedge \true$, where $\apform$ is a conjunction of disequations, $\asformB$ is a predicate atom with root $x$, and $\aseq$ is not an axiom. If $\aseq$ is \qireducible then we denote by $\varmap{\aseq}$ the set of variables occurring in $\vec{y}$.
\end{definition}
Note that the definition of $\varmap{\aseq}$ 
is unambiguous because the left-hand side of $\aseq$ cannot contain more than one points-to atom with root $x = \rootof{\asformB}$, otherwise it would be \heapunsat and $\aseq$ would be an axiom.
It is easy to check that every \ireducible sequent is \qireducible (see Definition~\ref{def:admissible}, Condition \ref{strat:ini}).
The converse does not hold in general, because rules \wea, \eli\ or \dec\ may be applicable on the considered sequent, and they have priority over \imi.
We observe that the application of \imi\ is necessarily interleaved with that of \sep:

\begin{lemma}
\label{lem:sep}
Every path in a proof tree between two distinct \ireducible sequents contains an application of rule \sep.
\end{lemma} 
\begin{proof}
Let $\aseq_1$ be an \ireducible sequent and assume that $\aseq_2$ is a \descendant of $\aseq_1$.
By definition, the only 
rule that applies on $\aseq_1$ is \imi, yielding a sequent 
$\aseq_1'$. Since the rules in $\asid$ are \prules, necessarily the right-hand side of $\aseq_1'$ contains
exactly one points-to atom. Since $\aseq_2$ is \ireducible, \imi\ applies on $\aseq_2$, which entails that 
no points-to atom can occur on the right-hand side of $\aseq_2$ (since \imi\ applies only when the right-hand side is a predicate atom). The only rule that can remove a points-to atom from the right-hand side of a sequent is \sep, thus \sep\ necessarily applies along the path from $\aseq_1$ to $\aseq_2$.
\end{proof}

\newcommand{\companion}{companion\xspace}

\newcommand{\specvar}[1]{{\cal V}^\star(#1)}



To analyze the termination behavior of the rules, we define a new set of variables $\specvar{\aseq}$, which is similar to $\narrowvar{\aseq}$: it contains variables that occur on the right-hand side of the considered sequent but are not allocated on the left-hand side, and that do not occur either in the set of non-allocated variables $V$.
\begin{definition}	
For every sequent $\aseq = \asheap \vdashr \asheapB$, we denote by $\specvar{\aseq}$ the
set $\vars_{\addr}(\asheapB) \setminus (V \cup \alloc{\asheap})$.
\end{definition}
Note that,  since the considered inferences are \admissible, if rule \dec\ applies on variables $x,y$ (with the notations of the rule), then necessarily $y \in \specvar{\aseq}$ and $x\in \vars(\asheap)$ 
(see Condition \ref{strat:dec} in Definition~\ref{def:strat}). 
The next proposition states that the rules (except possibly $\eq$) cannot add new variables in the set $\specvar{\aseq}$.
\begin{proposition}
\label{prop:specvar}
Let $\aseq_1$ be a sequent containing no equality. For every successor $\aseq_2$ of $\aseq_1$, we have $\specvar{\aseq_2} \subseteq \specvar{\aseq_1}$.
\end{proposition}
\begin{proof}
The proof is similar to that of Proposition~\ref{prop:narrow}.
Let $\aseq_i = \asheap_i \vdashrV{V_i} \asheapB_i$ (for $i = 1,2$).
Assume for the sake of contradiction that $v \in \specvar{\aseq_2} \setminus \specvar{\aseq_1}$.
By definition of $\specvar{\aseq_i}$, $v$ is of sort $\addr$, 
$v \in \vars(\asheapB_2)$, $v \not \in \alloc{\asheap_2} \cup V_2$, and either 
$v \not \in \vars(\asheapB_1)$ or
$v \in \alloc{\asheapB_1} \cup V_1$.
First assume that $v \not \in \vars(\asheapB_1)$.
Since $\aseq_1$ contains no equality and $v\notin \specvar{\aseq_1}$, by Proposition~\ref{prop:rule-facts} (\ref{it:intro:eq}), the only rule that can introduce new variables to the right-hand side of a sequent is \imi. 
If Rule \imi\ is applied, then we must have $v = y_i\sigma$ for some $i = 1,\dots,n$, with the notations used in the definition of the rule.
Since all the rules in $\asid$ are \prules, 
 $\asheapB_2$ necessarily contains an atom with root $v$ and $v\in \alloc{\asheapB_2$}.
Since $v \not \in \alloc{\asheap_2}$, this entails by Definition~\ref{def:antiax} (\ref{anti:alloc}) that $\aseq_2$ is an \antiaxiom, which contradicts Condition \ref{strat:noax} in Definition~\ref{def:admissible}. 
Thus we necessarily have $v \in \vars(\asheapB_1)$.
Now assume that $v \in \alloc{\asheapB_1} \cup V_1$. Then by Proposition~\ref{prop:rule-facts} (\ref{it:rem:var:alloc}), the only rule that can 
	remove a variable $v$ from $\alloc{\asheapB_1} \cup V_1$ is \dec, setting $x = v$. 
However, \dec\ replaces $x$ by another variable in the entire sequent, hence we have $x\not \in \vars(\asheapB_2)$, thus $x$ cannot be $v$.
\end{proof}

We  show (Lemma~\ref{lem:direct}) that the rules terminate in polynomial time if $\imi$ is not applied. To this purpose we prove that both the depth of the proof tree and the number of branches are polynomial. We first introduce the following definition:
\begin{definition}
A path $\aseq_0,\dots,\aseq_n$ in a proof tree from $\aseq$ to $\aseq'$ is {\em \direct} if there is no application of the rule {\imi} on $\aseq_0,\dots,\aseq_{n-1}$.
A \descendant $\aseq'$ of $\aseq$ is called {\em \direct} if the path from $\aseq$ to $\aseq'$ is \direct.
\end{definition}
The first goal is thus to prove that the length of all \direct paths is bounded. Then we will derive a bound on the total number of \direct descendants. 
To this purpose, we analyze the rules that can be applied on \qireducible sequents, and derive a number of easy but useful technical results. 
\begin{proposition}
\label{prop:qred}
If a sequent $\aseq$ is \qireducible then the only rules that can be applied on $\aseq$
are \imi, \wea, \dec, or \eli.
Moreover, for every \direct \descendant $\aseq'$ of $\aseq$, the sequent
$\aseq'$ is \qireducible and $\varmap{\aseq} \cap \vars(\aseq') \subseteq \varmap{\aseq'}$.
\end{proposition}
\begin{proof}
The proof is by an inspection of the different rules. Note that if $\unf$ applies on $(x \mapsto \vec{y} * \asform) \swedge \apform \vdashr \asformB \swedge \true$ (with the notations of Definition~\ref{def:reduced}), then,  by Definition~\ref{def:strat} (Condition \ref{strat:unf}), $\asform$ 
must contain an atom with the same root as $\rootof{\asformB} = x$, 
hence $(x \mapsto \vec{y} * \asform)$ is \heapunsat and the sequent is an axiom.
For the second part, it is clear that the only rule among \wea, \dec\ or \eli\ that can delete a variable $x$ from 
$\varmap{\aseq}$ is \dec, and this rule entirely removes $x$ from the sequent.
\end{proof}
\begin{proposition}
\label{prop:unf_qred}
The premises of \unf\ are \qireducible.
\end{proposition}
\begin{proof}
We use the notations of the  rule.
By Definition~\ref{def:strat} (Condition \ref{strat:unf}), $\asheapB$ is a predicate atom with the same root as $p(\vec{t})$ 
and $\apform$ is a conjunction of disequations.
Since the rules in $\asid$ are \prules, every formula $\asform_i$  contains a points-to atom with root $\rootof{p(\vec{t})} = \rootof{\asheapB}$. Thus the premises of \unf\ are necessarily \qireducible. 
\end{proof}
\begin{corollary}
\label{cor:onlyoneunf}
There is at most one application of \unf\ along a path containing no \ireducible sequent.
\end{corollary}
\begin{proof}
The result follows immediately from Propositions \ref{prop:qred} and \ref{prop:unf_qred}.
\end{proof}
We eventually derive (Lemma \ref{lem:lenpath}) the result concerning the length of the \direct paths. To this aim we first introduce a new notation: 
\newcommand{\diseq}[1]{\mathtt{NbDisEq}(#1)}
\begin{definition}
For every sequent $\aseq = \asheap \vdashr \asheapB$, we denote by 
$\diseq{\aseq}$ the number of disequations $x \not \iseq y$ not occurring in $\asheap$ such that $x,y \in \vars(\aseq)$.
\end{definition}
\begin{lemma}
\label{lem:lenpath}
The length of every \direct path from $\aseq$ is 
at most 
$\bigO{\size{\aseq}^2}$.
\end{lemma}
\begin{proof}
By Corollary \ref{cor:onlyoneunf}, an \direct path  contains at most one application of \unf.
Rule \eq\ applies at most $\card{\vars(\aseq)^2}$ times, with highest priority, yielding an 
\eqfree sequent, and afterward \eq\ can no longer be applied, since no inference rules introduce any equality to a sequent.
Thus it is sufficient to prove the result for \direct paths containing no application of rules \eq\ or \unf.
Let $\aseq_1,\dots,\aseq_n$ be a path with no application of \eq, \unf\ or \imi, where $\aseq_1 = \aseq$.
Let $\aseq_i = \asform_i \swedge \apform_i \vdashrV{V_i} \asformB_i \swedge \apformB_i$.
An inspection of the rules shows that we have $5\cdot\diseq{\aseq_i} + \size{\aseq_i} > 
5\cdot\diseq{\aseq_{i+1}} + \size{\aseq_{i+1}}$, for all $i = 1,\dots,n$. Indeed, none of the considered rules can add new variables
to $\aseq_i$ by Proposition~\ref{prop:rule-facts} (\ref{it:intro:eq}) and Condition \ref{strat:dec} of Definition~\ref{def:admissible}, 
all the rules (except possibly \dec) decrease $\size{\aseq_i}$ and \wea\ cannot 
remove disequations between variables occurring in predicate atoms. Rule \dec\ may add 
a disequation $x\not \iseq y$ in $\aseq_{i+1}$ (increasing the size by $4 = \size{* x \not \iseq y}$), but simultaneously removes the disequation $x \not \iseq y$ from $\diseq{\aseq_{i+1}}$.
Then the proof follows from the fact that $5\cdot\diseq{\aseq_1} + \size{\aseq_1} = \bigO{\size{\aseq_1}^2}$.
\end{proof}
To derive the result about the total number of \direct descendants, knowing the {\em length} of the paths is of course not sufficient: it is also necessary to estimate the {\em number} of such paths, which depends in particular on the number of applications of rule \dec. To this purpose, we prove the following result:
\begin{lemma}
\label{lem:dec}
Let $\aseq$ be a sequent, with 
$\card{\specvar{\aseq}} = \anum$.
The 
exhaustive application of \dec\ on $\aseq$ yields at most $\anum\cdot\diseq{\aseq}^{\anum}$ different branches.
\end{lemma}
\begin{proof}
The proof is by induction on the pair $(\anum,\diseq{\aseq})$.
By definition,  every application of \dec\ on $\aseq$ yields two sequents
$\aseq_1$ and $\aseq_2$, where $\aseq_1$ is obtained from $\aseq$ by replacing a variable $x$ by $y$
and $\aseq_2$ is obtained by adding the disequation $x \not \iseq y$.
By the application condition of the rule 
(Definition~\ref{def:strat}, Condition \ref{strat:dec}) necessarily $y \in \specvar{\aseq}$. 
Since $x$ is replaced by $y$ and $x$ is the root of an atom from the left-hand side of $\aseq$, we must have $y \not \in \specvar{\aseq_1}$, thus 
$\card{\specvar{\aseq_1}} = \anum-1$.
Now $\diseq{\aseq_1} \leq \diseq{\aseq}$, 
$\diseq{\aseq_2} = \diseq{\aseq} - 1$ 
(since by Definition~\ref{def:strat}, Condition \ref{strat:dec}, $x,y \in \vars(\aseq)$ and $x \not \iseq y$ does not occur in $\aseq$) and $\specvar{\aseq_2} = \specvar{\aseq}$, hence $\card{\specvar{\aseq_2}} = \anum$.
By the induction hypothesis,
the application of \dec\ generates at most $(\anum - 1)\cdot\diseq{\aseq_1}^{\anum-1} \leq \anum\cdot\diseq{\aseq}^{\anum-1}$ branches on $\aseq_1$
and
$\anum\cdot(\diseq{\aseq}-1)^{\anum} \leq \anum\cdot\diseq{\aseq}^{\anum-1}\cdot(\diseq{\aseq}-1) = \anum\cdot\diseq{\aseq}^{\anum} - \anum\cdot\diseq{\aseq}^{\anum-1}$ branches on $\aseq_2$.
Thus 
the total number of branches is at most $\anum\cdot\diseq{\aseq}^{\anum}$.
\end{proof}

\begin{lemma}
\label{lem:direct}
Assume that 
$\maxr{\asid} \leq \anumB$, 
for some fixed $\anumB \in \mathbb{N}$ 
(i.e., that both the maximal arity of the symbols and the number of record fields are bounded by $\anumB$).
The number of \direct\ {\descendant}s of any sequent $\aseq$ is polynomial w.r.t.\ 
$\size{\aseq} + \size{\asid}$ (but it is exponential w.r.t.\ $\anumB$).
\end{lemma}
\begin{proof}
By Lemma~\ref{lem:lenpath}, it is sufficient to prove that the total number of paths occurring in a proof tree with no application of \imi\ is polynomial w.r.t.\ $\size{\aseq} + \size{\asid}$.

The rules \eq, \wea, \eli, \noteq\ apply with the highest priority 
and yield only one branch, yielding a (unique) sequent $\asform \swedge \apform \vdash \asformB \swedge \apformB$.
If $\apformB \not = \true$, then by Definition~\ref{def:admissible}, no other rule can be applied because these other rules apply only when the right-hand side is of the form $\asformB \swedge \true$.
Otherwise, the rule \sep\ may apply 
(possibly several times), and eventually transforms the sequent  into sequents 
$\asform_i \swedge \apform_i \vdashrV{V_i} \asformB_i \swedge \true$, where each $\asformB_i$ is a predicate atom, 
$\asform = \bigast_{i=1}^n \asform_i$ and
$\asformB = \bigast_{i=1}^n \asformB_i$
(we may have $\apform_i \not = \apform$ 
since the rule may interleaved 
with applications of rule \wea). 

If the considered sequent is not \narrow, then by definition of the strategy (Condition \ref{strat:nocare} in Definition~\ref{def:admissible}),  there is at most one  application 
of  rule \sep;   
indeed, if several rule applications are possible then only the one that is selected is considered.   We thus obtain at most $\size{\aseq}$ branches, each ending with a \narrow sequent (because $\asformB_i$ is a predicate atom, thus $\asform_i \swedge \apform_i \vdashrV{V_i} \asformB_i$ is necessarily \narrow).

Now assume that $\aseq$ is \narrow and let $x_i$ be the root of $\asformB_i$. We have $\alloc{\asformB_i} = \{ x_i \}$, because $\asformB_i$ is a predicate atom.
Since $\asform_i * \apform_i \vdashrV{V_i} \asformB_i$ 
is not an axiom, 
we have $x_i \in \alloc{\asform_i}$. By definition of  rule \sep, this entails that $x_i \in V_j$, for all $j \in \{ 1,\dots,i-1,i+1,\dots, n \}$.
For {\sep} to be applied, it is only necessary to choose how to decompose the 
 formula $\asform$ into $\bigast_{i=1}^n \asform_i$. 
To this aim, we only have to associate each atom in $\asform$ with a formula $\asform_i$, hence to associate each variable $y \in \alloc{\asform}$ (that are by definition the roots of the spatial atoms in $\asform$) to an index $i = 1,\dots,n$.
By definition, the variable $x_i$ must be associated with index $i$, as otherwise we would have $x_i \not \in \alloc{\asform_i}$ and the premise
$\asform_i * \apform_i \vdashrV{V_i} \asformB_i$  would be an \antiaxiom.
We then arbitrarily choose the image of each variable occurring in $\narrowvar{\aseq}$.
Afterward, the image of the other variables are fixed inductively as follows.
Let $y \in \alloc{\asform}$ and assume that $y \not \in \narrowvar{\aseq} \cup \{ x_1,\dots,x_n \}$.
If $y$ occurs in some atom with root $y'$ in $\asform$ and $y'$ has already been associated with $i$, 
then we also associate $y$ with $i$.
For the sake of contradiction, assume that $y$ is associated with an index $j \not = i$.
By definition of rule \sep, this entails that $y \in V_i$, and since
$\asform_i * \apform_i \vdashrV{V_i} \asformB_i$ is not an \antiaxiom, we deduce that $y \in \vars(\asformB_i)$, hence
$y \in \vars(\asformB) \setminus (\alloc{\asformB} \cup V) = \narrowvar{\aseq}$, which contradicts our assumption (we cannot have $y \in V$ since $\aseq$ would then be an axiom).
Note that, for all variables $y \in \alloc{\asform} \setminus \{ x_1,\dots,x_n \}$, 
we have $x_i \pathto{\asform}^* y$, for some $i = 1,\dots,n$, because otherwise $\aseq$ would be an \antiaxiom.
Hence all such variables $y$ must be eventually associated with some indice $i = 1,\dots,n$.
Since $\aseq$ is \narrow, $\card{\narrowvar{\aseq}} \leq \anumB$. Consequently, 
there exist at most $\size{\aseq}^\anumB$ possible applications of rule \sep, each yielding $n \leq \size{\aseq}$ branches. We thus get a total of 
at most $\size{\aseq}^{\anumB+1}$ branches.

Afterwards,  rule \unf\ applies, yielding at most $\size{\asid}$ premises in each branch.
Then \dec\ applies on each leaf sequent $\aseq'$, and by Lemma~\ref{lem:dec}
we get at most  $\size{\asid}\cdot\diseq{\aseq'}^\anum$ branches.
The variables in $\aseq'$ either occur in $\aseq$ or are introduced by the rule \unf.
Since each application of \unf\ introduces at most $\maxk{\asid}$ variables, we get $\diseq{\aseq'} \leq (\size{\aseq} + \maxk{\asid})^2$.
\end{proof}

We derive the main result of this section:

\begin{theorem}
\label{theo:sound}
The root sequent of every (possibly infinite) \fexpanded proof tree is valid.
\end{theorem}
\begin{proof}
Let $\aseq$  be a non-valid sequent, with a counter-model $(\astore,\aheap)$.
By Lemma~\ref{lem:ax}, $\aseq$ is not an axiom, hence by definition of a \fexpanded proof tree, it admits successors. 
By Lemmata \ref{lem:rules} and \ref{lem:sep_sound}, one of these successors must be non-valid and must admit a counter-model $(\astore,\aheap')$, with $\card{\aheap'} \leq \card{\aheap}$, and if the rule that applies on $\aseq$ is \sep, then $\card{\aheap'} < \card{\aheap}$.
Starting from the root of the tree, 
we thus obtain (if this root is not valid) an infinite path $\aseq_0,\dots,\aseq_n,\dots,$ such that 
all the $\aseq_i$ admit a counter-model $(\astore_i,\aheap_i)$, $\card{\aheap_{i+1}} \leq \card{\aheap_{i}}$ and 
if \sep\ is applies on $\aseq_i$ then $\card{\aheap_{i+1}} < \card{\aheap_{i}}$.
Since $\card{\aheap}$ is finite  this entails that 
there exists $i\geq 0$ such that rule \sep\ does not apply on $\aseq_j$, for all $j \geq i$.
By Lemma~\ref{lem:sep}, this entails that there exists $k$ such that \imi\ does not apply on $\aseq_j$, for all $j \geq k$. Thus
$\aseq_k$ admits an infinite number of \direct\ {\descendant}s, which contradicts Lemma~\ref{lem:direct}.
\end{proof}

 \subsection{Completeness}
 
\newcommand{\amap}{\nu}
\newcommand{\renaming}{$\univ$-mapping\xspace}
\newcommand{\rename}[3]{(#1\circ#2,#1(#3))}

 \newcommand{\prA}{{\cal P}}

 We now establish completeness, i.e., we prove that every valid sequent admits a \fexpanded tree.
To this aim, we  prove that for every valid sequent, there exists a rule application yielding  valid premises.
Lemma~\ref{lem:sep_applies} handles the case of  rule $\sep$, 
and Lemma~\ref{lem:comp} handles all the other rules. We begin by establishing several preliminary results.
First, we note that the truth value of a formula in a structure is not dependent on the {\em names} of the locations: these locations can be freely renamed, provided the relations between the variables are preserved, and provided allocated locations are not mapped to the same image. This result will be useful to construct copies of models when needed in forthcoming proofs. 
More formally, we introduce a notion of a {\em \renaming} and state some conditions ensuring that the application of such {\renaming}s on a structure preserves the truth value of a formula.
\begin{definition}
A {\em \renaming} is a function $\amap$ mapping every element of $\univ_\asort$ to an element of $\univ_{\asort}$.
Let $\aheap$ be a heap. 
If $\amap$ is injective on $\dom{\aheap}$, then we denote by $\amap(\aheap)$ 
the heap with domain $\amap(\dom{\aheap})$, such that for all $\ell_0 \in \dom{\aheap}$, with 
$\aheap(\ell_0) = (\ell_1,\dots,\ell_n)$, 
we have 
$\amap(\aheap)(\amap(\ell_0)) = (\amap(\ell_1),\dots,\amap(\ell_n))$.
\end{definition}

\begin{lemma}
\label{lem:ren}
Let $\asheap$ be a symbolic heap, let $(\astore,\aheap)$ be an \rmodel of $\asheap$ and let $\amap$
be a \renaming that is injective on $\astore(\vars(\asheap)) \cup \dom{\aheap}$. Then $\rename{\amap}{\astore}{\aheap} \modelsr \asheap$.
\end{lemma}
\begin{proof}
The proof is  by induction on the satisfiability relation. 
The result is established also for spatial formulas and pure formulas.
\begin{itemize}
\item{If $\asheap$ is an equation $x \iseq y$, then $\astore(x) = \astore(y)$, hence 
$\amap(\astore(x)) = \amap(\astore(y))$ and
$\rename{\amap}{\astore}{\aheap} \modelsr \asheap$.}
\item{If $\asheap$ is a disequation $x \not \iseq y$, then $\astore(x) \not = \astore(y)$, and since $\amap$ is injective on $\astore(\vars(\asheap))$ we get
$\amap(\astore(x)) \not = \amap(\astore(y))$ and
$\rename{\amap}{\astore}{\aheap} \modelsr \asheap$.}

\item{If $\asheap = y_0 \mapsto (y_1,\dots,y_n)$ then we must have $\aheap = \{ (\astore(y_0),\dots,\astore(y_n)) \}$, 
which entails that 
$\amap(\aheap)$ is ${\{ (\amap(\astore(y_0)),\dots,\amap(\astore(y_n))) \}}$
and that
$\rename{\amap}{\astore}{\aheap} \modelsr \asheap$.}

\item{If $\asheap = \asform \swedge \apform$ (or $\asheap = \asform \wedge \apform$) then 
we have $(\astore,\aheap) \modelsr \asform$
and
$(\astore,\aheap) \modelsr \apform$.
By the  induction hypothesis, we get
$\rename{\amap}{\astore}{\aheap} \modelsr \asform$
and
$\rename{\amap}{\astore}{\aheap} \modelsr \apform$, thus
$\rename{\amap}{\astore}{\aheap} \modelsr \asheap$.
}

\item{If $\asheap = \asform_1 * \asform_2$ then there exists disjoint heaps 
$\aheap_i$ (for $i = 1,2$) such that 
$\aheap = \aheap_1 \union \aheap_2$ and
$(\astore,\aheap_i) \modelsr \asform_i$. Since $\astore(\vars(\asform_i)) \cup \dom{\aheap_i} \subseteq \astore(\vars(\asheap)) \cup \dom{\aheap}$, by the  induction hypothesis we have
$\rename{\amap}{\astore}{\aheap_i} \modelsr 
 \asform_i$ (for $i = 1,2$).
 But $\amap(\aheap_1)$ and $\amap(\aheap_2)$ must be disjoint since 
 $\dom{\aheap_1} \cap \dom{\aheap_2} = \emptyset$ and $\amap$ is injective, therefore
$\rename{\amap}{\astore}{\aheap} \modelsr 
 \asheap$.
}

\item{
If $\asheap$ is a predicate atom, then we have 
$\asheap \unfold  \asheapB$ and
$(\astore',\aheap) \modelsr \asheapB$, for some \namedextension{\astore'}{\astore}{\vars(\asheapB) \setminus \vars(\asheap)}.
Since the rules in $\asid$ are  \prules, for all variables $x \in \vars(\asheapB) \setminus \vars(\asheap)$, we have $x \in \alloc{\asheapB}$ thus $\astore'(x) \in \dom{\aheap}$ by Lemma~\ref{lem:alloc}.
This entails that $\amap$ is injective on $\astore'(\vars(\asheapB)) \cup \dom{\aheap}$, and
by the induction hypothesis we get 
$\rename{\amap}{\astore'}{\aheap} \modelsr \apform$.
Moreover, it is clear that $\amap(\astore')$ is an \extension{\amap(\astore)}{\vars(\asheapB) \setminus \vars(\asheap)}, thus
$\rename{\amap}{\astore}{\aheap} \modelsr \asheap$.
}
\end{itemize}
\end{proof}
This result entails that for every set $\smallU$ that is sufficiently large, counter-models may always be renamed, so that all locations occur in $\smallU$:
\begin{corollary}
\label{cor:ren}
Let $\smallU$ be an infinite subset of $\univ_{\asort}$.
Any non-valid and \eqfree sequent $\asheap \vdashr  \asheapB$ admits a counter-model
$(\astore,\aheap)$ such that $\dom{\aheap} \cup \astore(\vars_\addr) \subseteq \smallU$.
\end{corollary}
\begin{proof}
It suffices to consider any counter-model of $\asheapB$ and apply a bijective  \renaming 
(or any \renaming that is injective on $\astore(\vars(\asheap)) \cup \dom{\aheap}$)
that maps all locations $\ell \in \univ_\addr$ to an element in $\smallU$.
\end{proof}

The first step toward establishing completeness is to show that rule \sep\ always applies (on a valid sequent) if the right-hand side of the sequent contains several spatial formulas 
(assuming that the rules of higher priority do not apply). This is essential because the strategy in Definition~\ref{def:strat} forbids the application of the other rules  in this case. The difficulty is that, of course, \sep\ cannot be applied arbitrarily: we are required to obtain valid premises. To this purpose we show that the heap decomposition of the left-hand side ($\asform_1 * \asform_2$) matches that of the right-hand side ($\asformB_1 * \asformB_2$). This will be proven thanks to the next lemma. Intuitively, this lemmas states that, under some additional conditions, if a structure $(\astore,\aheap_1 \union \aheap_2)$ validates a formula $\asform_1* \asform_2$, then we may deduce that 
each structure $(\astore,\aheap_i)$ validates $\asform_i$, when all the variables allocated by $\asform_i$ are in the domain of $\aheap_i$.

\begin{lemma}
\label{lem:disjoint}
Let $\asform_i,\asformB_i$ (for $i = 1,2$) be spatial formulas, with $\alloc{\asformB_1 * \asformB_2} \subseteq \alloc{\asform_1 * \asform_2}$.
Let $\aheap_1$ and $\aheap_2$ be disjoint heaps such that $(\astore,\aheap_1 \union \aheap_2) \modelsr \asform_1 * \asform_2$ and $(\astore,\aheap_i) \modelsr \asformB_i$, for  $i = 1,2$.
If $(\dom{\aheap_1} \cup \dom{\aheap_2}) \cap \astore(\vars) \subseteq \astore(\alloc{\asform_1 * \asform_2})$
and $\astore(\alloc{\asform_i}) \subseteq \dom{\aheap_i}$ for  $i = 1,2$,
then $(\astore,\aheap_i) \modelsr \asform_i$, for $i = 1,2$.
\end{lemma}
\begin{proof}
By definition, since  $(\astore,\aheap_1 \union \aheap_2) \modelsr \asform_1 * \asform_2$, there exist disjoint heaps $\aheapB_1,\aheapB_2$ such that
$(\astore,\aheapB_i) \modelsr \asform_i$, for all $i = 1,2$, 
and $\aheap_1 \union \aheap_2 = \aheapB_1 \union \aheapB_2$.
We prove that $\aheapB_i = \aheap_i$ for $ i = 1,2$.
Since $\aheapB_1 \union \aheapB_2 = \aheap_1 \union \aheap_2$
and $\aheapB_1$ and $\aheapB_2$ are disjoint,
it is sufficient to prove that $\dom{\aheap_i} \subseteq \dom{\aheapB_i}$, for  $i = 1,2$.
By symmetry, we prove the result for $i = 1$.
Assume, for the sake of contradiction, that $\ell \in \dom{\aheap_1}$ and 
$\ell \not \in \dom{\aheapB_1}$.
By Proposition~\ref{prop:connect}, since $(\astore,\aheap_1) \modelsr \asformB_1$
and $\ell \in \dom{\aheap_1} \subseteq \locs{\aheap_1}$, there exists $y \in \alloc{\asformB_1}$ such that 
$\astore(y) \connect{\aheap_1}^* \ell$, thus there is a sequence of locations 
$\ell_0,\dots,\ell_n$ with $\ell_0 = \astore(y)$, $\ell = \ell_n$ and
$\ell_i \connect{\aheap_1} \ell_{i+1}$, for all $i = 0,\dots,n-1$.
Let $k$ be the smallest index such that $\ell_k \not \in \dom{\aheapB_1}$.
Note that  $\ell_i \in \dom{\aheap_1}$, for all $i \leq k$ and that
 $y \in \alloc{\asform_1}$. 
 Indeed, $y \in \alloc{\asformB_1} \subseteq \alloc{\asform_1} \cup \alloc{\asform_2}$ by the hypothesis of the lemma, and if 
 $y \in \alloc{\asform_2}$, then we get (again by the hypothesis of the lemma)
 $\astore(y) \in \dom{\aheap_2}$, which contradicts the fact that $\ell_0 = \astore(y) \in \dom{\aheap_1}$, as $\aheap_1$ and $\aheap_2$ are disjoint.
By Lemma~\ref{lem:alloc}, we deduce that  $\ell_0 = \astore(y) \in \dom{\aheapB_1}$ since  $y \in \alloc{\asform_1}$ and $(\astore,\aheapB_1) \modelsr\asform_1$. Thus $k > 0$, and necessarily $\ell_k \in \locs{\aheapB_1}$ (since $\ell_{k-1}\in \dom{\aheapB_1}$, by minimality of $k$).
We deduce that $\ell_k \in \locs{\aheapB_1} \setminus \dom{\aheapB_1}$.
By Lemma~\ref{lem:establish}, this entails that $\ell_k = \astore(x)$ for some $x\in \vars(\asform_1)$, as $(\astore,\aheapB_1) \models \asform_1$.
By the hypothesis of the lemma, we may assume that $x\in \alloc{\asform_1* \asform_2}$, since $\ell_k \in \dom{\aheap_1}$ and $(\dom{\aheap_1} \cup \dom{\aheap_2}) \cap \astore(\vars) \subseteq \astore(\alloc{\asform_1 * \asform_2})$.
However, we cannot have $x \in \alloc{\asform_1}$ because otherwise we would have $\astore(x) \in \dom{\aheapB_1}$ by Lemma~\ref{lem:alloc}, as $(\astore,\aheapB_1) \modelsr \asform_1$. We cannot have  
$x \in \alloc{\asform_2}$ either,
since otherwise we would have $\astore(x) \in \dom{\aheap_2}$, as  $\astore(\alloc{\asform_2}) \subseteq \dom{\aheap_2}$ by the hypothesis of the lemma, hence $\astore(x) \not \in \dom{\aheap_1}$ because $\aheap_1$ and $\aheap_2$ are disjoint.
Thus we obtain a contradiction.
\end{proof}
We then derive the result about rule \sep. The main issue is that we have to take into account the fact that the strategy ``blocks'' some applications of \sep, if the sequent is not \narrow (see Condition~\ref{strat:sep} in Definition~\ref{def:strat}). We show that blocked applications do not yield  valid premises. To illustrate the difficulties that arise when applying \sep, we provide the examples below. The first one illustrates the importance of the fact that each predicate allocates at most one parameter.
\begin{example}
\label{ex:sep1}
Consider the sequent $p(x,y) \vdashrV{\emptyset} q(x,y) * r(y)$, with the rules:
\[ 
\begin{tabular}{ccccccc}
$p(x,y)$  & $\Leftarrow$ & $x \mapsto (y) * p'(y)$ & \qquad & $q(x,y)$ & $\Leftarrow$ & $x \mapsto (y)$ \\
$p'(y)$  & $\Leftarrow$ & $y \mapsto ()$ & & $r(y)$ & $\Leftarrow$ & $y \mapsto ()$ \\
\end{tabular}
\]
Note that the rules of $p$ are {\em not} \prules, as both $x$ and $y$ are allocated by $p(x,y)$. Here, although the sequent is indeed valid, no application of \sep\ yields valid premises: to show that the sequent is valid, one has first to unfold $p(x,y)$ once, before applying \sep.
In our context, such a situation cannot arise as each predicate atom allocates only one of its parameters, namely its root.
As we shall see, this entails that every decomposition of the right-hand side of the sequent necessarily corresponds to some (purely syntactic) decomposition of the left-hand side.
\end{example}
The next example illustrates the importance of  determinism.
 \begin{example}
 \label{ex:sep2}
Consider the sequent $\aseq: p(x,y) * q(y) * p(z,y) \vdashrV{\emptyset} 
p'(x,y) * q'(z,y)$, with the rules:
\[
\begin{tabular}{rclrcl}
$p(x,y)$ & $\Leftarrow$ & $ x \mapsto (y)$ \qquad &
$q(y)$ & $\Leftarrow$ & $ y \mapsto (z) * q(z)$ \\
$q(y)$ & $\Leftarrow$ & $ y \mapsto ()$ &
$q(y)$ & $\Leftarrow$ & $y \mapsto (y)$ \\
$p'(x,y)$ & $\Leftarrow$ & $x \mapsto (z) * p'(z,y)$ &
$q'(x,y)$ & $\Leftarrow$ & $x \mapsto (z) * q'(z,y)$ \\
$p'(x,y)$ & $\Leftarrow$ &  $x \mapsto (y)$ &
$q'(x,y)$ & $\Leftarrow$ & $x \mapsto (y)$ \\
$p'(x,y)$ & $\Leftarrow$ & $x \mapsto ()$ &
$q'(x,y)$ & $\Leftarrow$ & $x \mapsto (x)$ 
\end{tabular}
\]
Intuitively, $q(y)$ denotes a list starting at $y$ and ending either with an empty tuple or with a loop on its last element, whereas $p'(x,y)$ (resp.\ $q'(x,y)$) denotes a list starting at $x$ and ending either with $y$ or with an empty tuple (resp.\ with a loop on the last element). These rules are not deterministic (for instance there is an overlap between the rules $p'(x,y) \Leftarrow x \mapsto (z) * p'(z,y)$  and $p'(x,y) \Leftarrow x \mapsto (y)$). 
Although the sequent $\aseq$ is valid, there is no application of \sep\ that yields valid premises. Indeed, none of the sequents $p(x,y) * q(y) \vdashrV{\emptyset} 
p'(x,y)$, or $q(y) * p(z,y) \vdashrV{\emptyset} 
q'(z,y)$ are valid.
Here \sep\ cannot be applied before the entire list $q(y)$ is unfolded, as the decision to group $q(y)$ with $p(x,y)$ or $p(y,z)$ cannot be made before the last cell in $q(y)$ is known. This would yield an infinite proof tree (with infinitely many branches).
The fact that the rules are \deterministic prevents such a behavior to occur.
 \end{example}

\begin{lemma}
 \label{lem:sep_applies}
 Let $\aseq: \asform \swedge \apform \vdashr (\asformB_1 * \asformB_2) \swedge \true$  be a valid sequent
 where $\asformB_i \not = \emp$ for $i = 1,2$ and $\apform$ is a conjunction of disequations.
 If $\aseq$ is not an axiom, then 
 there exists an application of \sep\ with conclusion $\aseq$
 for which all the premises are valid.
  \end{lemma}
 \begin{proof}
 To prove that \sep\ applies on $\aseq$, it is necessary to show that $\asform$ is of the form $\asform_1 * \asform_2$. However $\asform_i$ cannot be chosen arbitrarily, since the premises have to be valid.
To construct $\asform_i$, we shall construct a model of $\asform$ of a particular form, get a  decomposition of the heap of this model from the formula $\asformB_1 * \asformB_2$ on the right-hand side of the sequent, and use this decomposition to compute suitable formulas $\asform_1$ and $\asform_2$.
We first notice that, since $\aseq$ is not an axiom, $\asform$ must be \heapsat.
 Furthermore, since $\aseq$ is valid, it cannot be an \antiaxiom by Lemma~\ref{lem:antiax}, thus
 $\alloc{\asformB_i} \subseteq \alloc{\asform}$ for $i=1,2$. Since $\asformB_i \not = \emp$, this entails that $\alloc{\asform}$ is not empty.
 Let $\smallU_1,\smallU_2$ be disjoint infinite subsets of $\univ_{\addr}$ and let $\astore$ be an injective store such that
 $\astore(\vars) \cap \smallU_i = \emptyset$, for $i = 1,2$.
  Note that such $\smallU_1,\smallU_2$ and $\astore$ exist since $\univ_{\asort}$ is infinite, for all $\asort \in \sorts$.
  By Lemma~\ref{lem:sat} applied with $\smallU = \smallU_1$ and with any variable $x \in \alloc{\asform}$, there exists an \rmodel $(\astore,\aheap)$ of $\asform$ such that 
  $\dom{\aheap} \subseteq \smallU_1 \cup \astore(\alloc{\asform})$.
Since $\astore$ is injective and $\apform$ is a conjunction of disequations, necessarily $\astore \models \apform$ (observe that $\apform$ contains no disequation of the form $u \not \iseq u$ as otherwise $\aseq$ would be an axiom).
If 
there exists $x\in V$ such that $\astore(x) \in \dom{\aheap}$, then since $\dom{\aheap} \subseteq \smallU_1 \cup \astore(\alloc{\asform})$ and $\astore(\vars) \cap \smallU_1 = \emptyset$, 
 there must exist $x'\in \alloc{\asform}$ such that 
$\astore(x) = \astore(x')$. 
Since $\astore$ is injective, necessarily
$x = x'$, which entails that $\aseq$ is an axiom, contradicting the hypotheses of the lemma.
We deduce that $\astore(V) \cap \dom{\aheap} =  \emptyset$.
Since $\astore$ is injective, it is injective on $V$, as otherwise $V$ would contain two occurrences of the same variable and $\aseq$ would be an axiom.
Because $\aseq$ is valid, 
we deduce that $(\astore,\aheap) \modelsr  (\asformB_1 * \asformB_2)$, hence there exist disjoint heaps $\aheap_1$ and $\aheap_2$ such that $\aheap = \aheap_1 \union \aheap_2$ and 
$(\astore,\aheap_i) \modelsr \asformB_i$, for $i = 1,2$.  
Let $\asform_i$ be the separating conjunction of all the spatial atoms 
$\anatom$ occurring in $\asform$ such that $\astore(\alloc{\anatom}) \subseteq \dom{\aheap_i}$.
By Lemma~\ref{lem:alloc}, 
$\astore(\alloc{\asform}) \subseteq \dom{\aheap} = \dom{\aheap_1} \cup \dom{\aheap_2}$, hence every atom $\anatom$
occurs in either $\asform_1$ or $\asform_2$, and since $\aheap_1$ 
and $\aheap_2$ are disjoint, no atom $\anatom$ can occur both in $\asform_1$ and $\asform_2$.
Therefore, $\asform = \asform_1 * \asform_2$.

%
We now prove that $(\astore,\aheap_i) \modelsr\asform_i$ for  $i =1,2$. To this aim we use Lemma~\ref{lem:disjoint}, thus we verify that all the hypotheses of the lemma are satisfied. Since $\dom{\aheap} \subseteq \smallU_1 \cup \astore(\alloc{\asform})$ and $\smallU_1 \cap \astore(\vars) = \emptyset$, we have $\dom{\aheap} \cap \astore(\vars) \subseteq \astore(\alloc{\asform})$, i.e. $(\dom{\aheap_1} \cup \dom{\aheap_2}) \cap \astore(\vars) = \astore(\alloc{\asform_1 * \asform_2})$.
 Furthermore, $\astore(\alloc{\asform_i}) \subseteq \dom{\aheap_i}$ for $i = 1,2$ by definition of $\asform_i$, $(\astore,\aheap_1 \union \aheap_2) = (\astore,\aheap) \modelsr \asform = \asform_1 * \asform_2$, $(\astore,\aheap_i) \modelsr \asformB_i$ (for all $i = 1,2$) and $\alloc{\asformB_1 * \asformB_2} \subseteq  \alloc{\asform_1 *\asform_2}$.
By Lemma~\ref{lem:disjoint}, 
we deduce that $(\astore,\aheap_i) \modelsr\asform_i$ for  $i =1,2$.

It is clear that there is an application of \sep\ on $\aseq = (\asform_1 * \asform_2) \swedge \apform \modelsr (\asformB_1 * \asformB_2) \swedge \true$ yielding 
the premises $\asform_i \swedge \apform \vdashrV{V_i} \asformB_i$ for $i = 1,2$, where $V_i = V \cup \alloc{\asform_{3-i}}$. There remains to prove that these premises are valid and that the application of the rule is allowed by the strategy.
Assume that one of these premises, say $\asform_1 \swedge \apform \vdashrV{V_1} \asformB_1$, is not valid.
By Corollary \ref{cor:ren} applied with $\smallU = \smallU_2$, $\asform_1 \swedge \apform \vdashrV{V_1} \asformB_1$ admits a counter-model
$(\astore',\aheap_1')$, where $\dom{\aheap_1'} \cup \astore'(\vars_\addr) \subseteq \smallU_2$.
By construction we have 
$\astore'(V_1) \cap \dom{\aheap_1'} = \emptyset$, thus 
$\astore'(V) \cap \dom{\aheap_1'} = \emptyset$ since $V \subseteq V_1$.
Moreover, $\astore'$ is injective on $V_1 = V \cup \alloc{\asform_2}$.
By Lemma~\ref{lem:sat}, $\asform_2$ therefore admits a model $(\astore',\aheap_2')$ with 
$\dom{\aheap_2'} \subseteq \smallU_1 \cup \astore'(\alloc{\asform_2})$.
We show that $\aheap_1'$ and $\aheap_2'$ are disjoint. 
Assume for the sake of contradiction that $\ell \in \dom{\aheap_1'} \cap \dom{\aheap_2'}$. Since $\smallU_1 \cap \smallU_2 = \emptyset$,
we have 
$\dom{\aheap_1'} \cap \dom{\aheap_2'} \subseteq \astore'(\alloc{\asform_2})$, 
thus
$\ell \in \astore'(\alloc{\asform_2})$.
But since $\astore'(V_1) \cap \dom{\aheap_1'} = \emptyset$ and $\alloc{\asform_2} \subseteq V_1$
we also have $\ell \not \in \astore'(\alloc{\asform_2})$, a contradiction.
Thus $\aheap_1'$ and $\aheap_2'$ are disjoint
and we have $(\astore',\aheap_1' \union \aheap_2') \modelsr (\asform_1 * \asform_2) \swedge \apform = \asform \swedge \apform$.
If $\astore'(V) \cap \dom{\aheap_2'}$ contains a location $\ell$ then since $\smallU_1 \cap \smallU_2 = \emptyset$, necessarily $\ell = \astore'(x)$, for some $x \in \alloc{\asform_2}$, and since $\astore'$ is injective on $V_1 = V \cup \alloc{\asform_2}$ and $\ell \in \astore'(V)$, necessarily, $x \in V$ and $\aseq$ is an axiom, which contradicts the hypotheses of the lemma.
We deduce that $\astore'(V) \cap \dom{\aheap_2'} = \emptyset$, and 
since $\astore'(V) \cap \dom{\aheap_1'} = \emptyset$,  that 
 $\astore'(V) \cap \dom{\aheap_1' \union \aheap_2'} = \emptyset$.
 Furthermore, by construction, $\astore'$ is injective on $V$.
Since 
$\aseq$ is valid, this entails that $(\astore',\aheap_1' \union \aheap_2') \modelsr (\asformB_1 * \asformB_2)$, hence
$\aheap_1' \union \aheap_2' = \aheap_1'' \union \aheap_2''$, where $(\astore',\aheap_i'') \modelsr \asformB_i$.
By Lemma~\ref{lem:unique}, since 
$(\astore',\aheap_2') \modelsr \asformB_2$ and $\asid$ is \deterministic, we get 
$\aheap_2'' = \aheap_2'$, and therefore $\aheap_1'' = \aheap_1'$, thus $(\astore',\aheap_1') \modelsr \asformB_1$.
This contradicts the fact that $(\astore',\aheap_1')$ is a counter-model of 
$\asform_1 \swedge \apform \vdashrV{V_2} \asformB_1$.

If $\aseq$ is \narrow, then the proof is completed, since 
all the applications of rule $\sep$ are considered by the strategy in this case.
If $\aseq$ is not \narrow, then the application of the rule may be blocked by Condition~\ref{strat:nocare} of Definition~\ref{def:strat}. 
In this case, this means that there exist $\asform_1',\asform_2'$ such that $\asform = \asform_1' * \asform_2'$
and $\asform_1' \swedge \apform \vdashrV{V_1'} \asformB_1$ is valid, where 
$V_1' = V \cup \alloc{\asform_2'}$.
We assume that the other premise, 
$\asform_2' \swedge \apform \vdashrV{V_2'} \asformB_2$ (with $V_2' = V \cup \alloc{\asform_1'}$)
is not valid, and we derive a contradiction.
By Corollary \ref{cor:ren}, this premise admits 
a counter-model 
$(\astore'',\aheap''_2)$ such that $\dom{\aheap''_2} \cup \astore''(\vars_{\addr}) \subseteq \smallU_1$.
We have $(\astore'',\aheap_2'') \modelsr \asform_2'$, $\astore'' \models \apform$,
$\astore''(V_2') \cap \dom{\aheap_2''} = \emptyset$,
$\astore''$ is injective on $V_2'$
and $(\astore,\aheap_2'') \not \modelsr \asformB_2$.
By Lemma~\ref{lem:sat}, 
$\asform_1'$ admits a model
$(\astore'',\aheap_1'')$ such that
 $\dom{\aheap_1''} \subseteq \smallU_2 \cup \astore''(\alloc{\asform_1'})$.

Note that the heaps
 $\aheap_1''$ and $\aheap_2''$ are disjoint. Indeed, if $\ell \in \dom{\aheap_1''} \cap \dom{\aheap_2''}$ then, since 
 ${\dom{\aheap_1''} \subseteq \smallU_2 \cup \astore''(\alloc{\asform_1'})}$ and $\dom{\aheap''_2} \cup \astore''(\vars_{\addr}) \subseteq \smallU_1$, with $\smallU_1 \cap \smallU_2 = \emptyset$, we obtain ${\ell \in \astore''(\alloc{\asform_1'})}$, so that 
 $\ell \in \astore''(V_2')$, contradicting the fact that $\astore''(V_2') \cap \dom{\aheap_2''} = \emptyset$.

 We prove that  $\astore''$ is injective on $V_1'$, and that $\dom{\aheap_1''} \cap \astore''(V_1) = \emptyset$, which entails that 
 $(\astore'',\aheap_1'') \modelsr \asformB_1$, since $\asform_1' \swedge \apform \vdashrV{V_1'} \asformB_1$ is valid.
 If there exist $x_1,x_2$ such that $\{ x_1,x_2 \} \subseteqm V_1' = V \cup \alloc{\asform_2'}$ and $\astore''(x_1) = \astore''(x_2)$ then, as 
 $\astore''$ is injective on $V_2' = V \cup \alloc{\asform_1'}$, necessarily at least one of the variables $x_1,x_2$ -- say $x_1$ --
 is in $\alloc{\asform_2'}$. As $(\astore'',\aheap_2'') \models \asform_2'$, we get 
 $\astore''(x_1) \in \dom{\aheap_2''}$ (by Lemma~\ref{lem:alloc}), and 
$\{ x_1,x_2 \} \not \subseteqm \alloc{\asform_2'}$ (by Corollary  \ref{cor:heapunsat}). 
This entails that 
 $\astore''(x_1)  \not \in \astore''(V_2')$ (as $\astore''(V_2') \cap \dom{\aheap_2''} = \emptyset$) and $x_2 \in V \subseteq V_2'$, yielding a contradiction, since $\astore''(x_1) = \astore''(x_2)$. 
 Now assume that there exists $y\in V_1' = V \cup \alloc{\asform_2'}$ such that $\astore''(y) \in \dom{\aheap_1''}$. As $\dom{\aheap_1''} \subseteq \smallU_2 \cup \astore''(\alloc{\asform_1'})$ and $\dom{\aheap''_2} \cup \astore''(\vars_{\addr}) \subseteq \smallU_1$, we get
 $\astore''(y) \in \astore''(\alloc{\asform_1'})$, i.e., there exists $y' \in \alloc{\asform_1'}$ such that
 $\astore''(y) = \astore''(y')$. If $y \in V$, then $\{ y,y' \} \subseteqm V_2'$, contradicting the fact that $\astore''$ is injective on $V_2'$.
 Thus $y \in \alloc{\asform_2'}$, so that $\astore(y) \in \dom{\aheap_2''}$ (by Lemma~\ref{lem:alloc}), contradicting the fact that $\aheap_1''$ and $\aheap_2''$ are disjoint. 

Thus
 $(\astore'',\aheap_1'' \union \aheap_2'') \modelsr \asform_1' * \asform_2' = \asform$.
 As $\astore'' \models \apform$ and ${\astore''(V) \cap \dom{\aheap_i''} \subseteq \astore(V_i) \cap \dom{\aheap_i''}}$ is empty, $\astore''$ is injective on $V$, and $\aseq$ is valid, 
 we get $(\astore'',\aheap_1'' \union \aheap_2'') \modelsr  \asformB_1 * \asformB_2$.
 Therefore, there exist disjoint heaps $\aheapB''_i$ (for $i = 1,2$), such that 
 $(\astore'',\aheapB_i'') \modelsr  \asformB_i$  and
 $\aheap_1'' \union \aheap_2'' = \aheapB_1'' \union \aheapB''_2$.
 By Lemma~\ref{lem:unique}, since $(\astore'',\aheap_1'') \modelsr \asformB_1$, we deduce that $\aheap_1'' = \aheapB_1''$, thus $\aheap_2'' = \aheapB_2''$, hence
 and $(\astore'',\aheap_2'') \modelsr \asformB_2$,
 which contradicts the fact that $(\astore'',\aheap_2'')$ is a counter-model of $\asform_2' \swedge \apform \vdashrV{V_2'} \asformB_2$.
 \end{proof}
 
\newpage

We handle the case of the other rules, more precisely we show that there is always an inference rule that applies on a valid sequent (if it is not an axiom): 
 \begin{lemma}
 \label{lem:comp}
For any valid sequent $\asform \swedge \apform \vdashr \asformB \swedge \apformB$ that is not an axiom, there exists
an application of the inference rules that yields valid premises.
 \end{lemma}
 \begin{proof} 
Since $\asform \swedge \apform \vdashr \asformB \swedge \apformB$ is not an axiom, $\asform$ is \heapsat, $V \cap \alloc{\asform} = \emptyset$ and $V$ contains at most one occurrence of each variable.
Note that if an invertible rule applies on $\asform \swedge \apform \vdashr \asformB \swedge \apformB$
then the result holds, since the premises are  necessarily valid. Thus we may assume that 
$\asform \swedge \apform \vdashr \asformB \swedge \apformB$  is irreducible by all invertible rules, hence 
(by Lemma~\ref{lem:rules}) by all 
the rules except $\sep$.

Let $\astore$ be an injective store such that the set 
$\smallU =\univ_{\addr} \setminus \astore(\vars)$ is infinite.
If $\asform \not = \emp$, then $\asform$ contains at least one variable of sort $\addr$, and by Lemma~\ref{lem:sat}, we deduce that 
$\asform$ admits a model $(\astore,\aheap)$ such that
$\dom{\aheap} \subseteq \smallU \cup \astore(\alloc{\asform})$.
If $\asform = \emp$ then the same property trivially holds, with $\aheap = \emptyset$.
Note in particular that the assertion $\dom{\aheap} \subseteq \smallU \cup \astore(\alloc{\asform})$ entails that 
$\astore(V) \cap \dom{\aheap} = \emptyset$ since 
$\astore(V) \cap \smallU = \emptyset$, $V \cap \alloc{\asform} = \emptyset$ and $\astore$ is injective.

 If $\apform$ contains an equation then rule \eq\ applies, which contradicts our assumption that the sequent is irreducible by all invertible rules; thus $\apform$ is a conjunction of disequations. Furthermore, none of these disequations may be of the form $x \not \iseq x$ because otherwise $\asform \swedge \apform \vdashr \asformB \swedge \apformB$ would be an axiom.
This entails that $\astore \models \apform$. Since $\asform \swedge \apform \vdashr \asformB \swedge \apformB$  is valid, $\astore(V) \cap \dom{\aheap} = \emptyset$ and $\astore$ is injective (hence injective on $V$ because $V$ contains at most one occurrence of each variable), we deduce that 
\[(\astore,\aheap) \modelsr \asformB \swedge \apformB. \quad\quad (\dagger)\]
  
  Assume that $\apformB$ contains an equation $x  \iseq y$. If $x = y$ then Rule
 $\noteq$ applies, and
 otherwise, we necessarily have $\astore(x) \not = \astore(y)$, hence $(\astore,\aheap) \not \modelsr \asformB \swedge \apformB$ 
 which contradicts ($\dagger$). 
Thus $\apformB$ contains no equation.
Now assume that $\apformB$ contains a disequation $x \not  \iseq y$. 
If $\{ x,y\} \subseteq \csts$ 
(i.e., both $x$ and $y$ are constants) or $\{ x,y \} \subseteqm \alloc{\asform} \cup V$ then  $\asform \swedge \apform \Emodels x \not \iseq y$, and \noteq\ applies, which is impossible by hypothesis.
If $x = y$ then $x \not  \iseq y$ is unsatisfiable, which contradicts  ($\dagger$). 
Let $\astore'$ be a store verifying $\astore'(x) = \astore'(y)$, that maps all other variables to pairwise distinct
elements of the appropriate sort, also distinct from $\astore(x)$, and such that $\smallU' = \univ_{\addr} \setminus \astore'(\vars)$ is infinite. Such a store exists since $\univ_{\addr}$ is infinite and $\{ x,y \} \not \subseteq \csts$. If $\astore' \not \models \apform$,
then since $\apform$  contains no equation, $x \not \iseq y$ must occur in $\apform$.
Then the rule \noteq\ applies, which contradicts our assumption that the sequent is irreducible by all invertible rules.
Thus $\astore' \models \apform$.
Since $\{ x,y \}  \not \subseteqm \alloc{\asform}$, $\astore'$ is injective on $\alloc{\asform}$. 
By Lemma~\ref{lem:sat}, we deduce that 
$\asform$ admits a model of the form $(\astore',\aheap')$ with $\dom{\aheap'} \subseteq \smallU' \cup\astore'(\alloc{\asform})$. 
If $\astore'(V) \cap \dom{\aheap'}$ contains a location $\ell$, then 
$\ell = \astore'(z) = \astore'(z')$, for some $z \in \alloc{\asform}$ and $z'\in V$.
Since $\alloc{\asform} \cap V = \emptyset$ and $\astore'$ is injective on all pair of variables except on $\{ x, y \}$, necessarily we have either $x = z$ and $y = z'$, or $x = z'$ and $y = z$, 
which means that $\{ x,y \} \subseteqm \alloc{\asform} \cup V$, hence \noteq\ applies, contradicting our assumption that the sequent is irreducible by all invertible rules.
We deduce that 
$\astore'(V) \cap \dom{\aheap'} = \emptyset$.
Since $\{ x,y \} \not \subseteqm V$, $\astore'$ is necessarily injective on $V$.
As $\asform \swedge \apform \vdashr \asformB \swedge \apformB$ is valid, we deduce that 
 $(\astore',\aheap') \modelsr \asformB \swedge \apformB$
 hence  $\astore'(x) \not = \astore'(y)$, which contradicts the definition of $\astore'$.
 Therefore, $\apformB = \true$.
 
If $\asformB = \emp$ then necessarily $\asform = \emp$, since otherwise by Corollary \ref{cor:noemp} we would have 
$\aheap \not = \emptyset$, which contradicts $(\dagger)$. Hence in this case, $\asform \swedge \apform \vdashr \asformB \swedge \apformB$ is of the form $\emp\swedge \apform \vdashr \emp$ and  is an axiom. 
  
If $\asformB$ is of the form $\asformB_1 * \asformB_2$, where $\asformB_i \not = \emp$ for $i = 1,2$, then the proof follows from Lemma~\ref{lem:sep_applies}.
Thus, we may assume that $\asformB$ is a spatial atom. Let $x$ be the root of $\asformB$.
If $x \not \in \roots{\asform}$, then we have 
$(\astore,\aheap) \not \modelsr \asformB$ by Lemma~\ref{lem:alloc}, which  contradicts  ($\dagger$). 
Therefore, $\asform$ contains a spatial atom with root $x$.
If this spatial atom is a predicate 
atom then \unf\ can be applied.
Now assume that $\asform$ contains a points-to atom $x \mapsto (y_1,\dots,y_n)$, i.e., is of the form
$x \mapsto (y_1,\dots,y_n) * \asform'$ (modulo AC).
Since $\asformB$ is a spatial atom with root $x$, there are two cases to consider:
\begin{enumerate}
\item{
$\asformB$ is of the form $x \mapsto (y_1',\dots,y_n')$.
We have $\aheap = \{ (\astore(x),\astore(y_1),\dots,\astore(y_n)) \} \union \aheap'$, where $\aheap'$ is a heap such that $(\astore,\aheap') \modelsr \asform'$ and $\astore(x) \not \in \dom{\aheap'}$.
By $(\dagger)$, we have $(\astore,\aheap) \modelsr \asformB$, thus
$\aheap = \{ (\astore(x),\astore(y_1'),\dots,\astore(y_n'))  \}$. This entails that 
$\aheap' = \emptyset$ (hence 
 $\asform' = \emp$ by Corollary \ref{cor:noemp}), and that $\astore(y_i) = \astore(y_i')$ for all $i = 1,\dots,n$. Since $\astore$ 
is injective, we have $y_i = y_i'$ for all $i = 1,\dots,n$. 
Therefore, $\asformB = \asform$, and $\asform \swedge \apform \vdashr \asformB \swedge \apformB$ is an axiom (because $\apformB = \true$), which contradicts the hypothesis of the lemma.}

\item{
$\asformB$ is a predicate atom.
By ($\dagger$), we have 
$(\astore,\aheap) \modelsr \asformB$ and by definition there exists a formula $\asheapB'$ such that
$\asformB \unfold \asheapB'$, 
where $(\astore',\aheap) \modelsr \asheapB'$ for some \namedextension{\astore'}{\astore}{\vars(\asheapB') \setminus \vars(\asformB)}, and in particular we have $\astore(x) = \astore'(x)$.
As the rules in $\asid$ are  \prules,
$\asheapB'$ is of the form ${(x \mapsto (u_1,\dots,u_m) * \asformB') \swedge \apformB'}$, with $\vars(\asheapB') \subseteq \vars(\asformB) \cup \{ u_1,\dots,u_m \}$,
and we have $\aheap(\astore(x)) = \aheap(\astore'(x)) = (\astore'(u_1),\dots,\astore'(u_m))$.
Since $\asform$ contains a points-to atom $x \mapsto (y_1,\dots,y_n)$ and $(\astore,\aheap) \modelsr \asform$, necessarily
$\aheap(\astore(x)) = (\astore(y_1),\dots,\astore(y_n))$. Consequently $n = m$ and $\astore(y_i) = \astore'(u_i)$ for all $i = 1,\dots,n$.
Let $\sigma$ be the substitution mapping every variable $u_i$ not occurring in $\vars(\asformB)$ to 
$y_i$. Note that $\sigma$ is well-defined: if $u_i = u_j$ then $\astore(y_i) = \astore(y_j)$ and $y_i = y_j$ since $\astore$ is injective. 
It is straightforward to check that $\astore'(z) = \astore(z\sigma)$ for all variables $z \in \vars$,  using the definition of $\sigma$ and the fact that $\astore'$ and $\astore$ coincide on all variables not occurring in $u_1,\dots,u_n$.
If $\imi$ applies with this substitution $\sigma$ then the proof is completed. Otherwise, 
we must have $\asform \swedge \apform \nEmodels \apformB'\sigma$, i.e., $\apformB'$ must contain a disequation $u' \not \iseq v'$ such that  
$\asform \swedge \apform \nEmodels u'\sigma \not \iseq v'\sigma$. Note that this entails that $u'\sigma$ and $v'\sigma$ cannot 
both be constants.
Consider a store $\astore''$ verifying $\astore''(u'\sigma) = \astore''(v'\sigma)$, and mapping all other variables to pairwise distinct elements of the appropriate sort. We may assume that $\astore''$ is chosen in such a way that the set $\smallU' = \univ \setminus \astore''(\vars)$ is infinite.
By definition,  $\astore'' \not \models \apformB'\sigma$, and since $\asform \swedge \apform \nEmodels u'\sigma \not \iseq v'\sigma$, we have $\{ u'\sigma,v'\sigma \} \not \subseteqm V \cup \alloc{\asform}$, 
thus
 $\astore''$
 is injective on $\alloc{\asform} \cup V$ and $\astore'' \models \apform$ (since $\apform$ is a conjunction of disequations not containing $u'\sigma \not \iseq v'\sigma$). 
 By Lemma~\ref{lem:sat}, $\asform$ admits a model $(\astore'',\aheap'')$ such that $\dom{\aheap''} \subseteq \smallU'' \cup \astore''(\alloc{\asform})$. This entails that $\astore''(V) \cap \dom{\aheap''} = \emptyset$.
Since  $\astore''$ is injective on $V$ and $\asform \swedge \apform \vdashr \asformB \swedge \apformB$ is valid, we deduce that
$(\astore'',\aheap'') \modelsr \asformB$, thus
$\asformB \unfold \asheapB''$, 
where $(\hat{\astore},\aheap'') \modelsr \asheapB''$, for some formula $\asheapB''$ and \namedextension{\hat{\astore}}{\astore''}{\vars(\asheapB'') \setminus \vars(\asformB)}.
Since the rules in $\asid$ are  \prules, 
$\asheapB''$ is of the form $(x \mapsto (u_1',\dots,u_k') * \asformB'') \swedge \apformB''$, 
and we have $\aheap''(\astore''(x)) = (\hat{\astore}(u_1'),\dots,\hat{\astore}(u_k'))$.
We may assume, by renaming, that  $\{ u_1',\dots,u_k' \} \cap \{ u_1,\dots,u_m \} \subseteq \vars(\asformB)$.
Since $\asform$ contains a points-to atom $x \mapsto (y_1,\dots,y_n)$ and $(\astore'',\aheap'') \modelsr \asform$, necessarily
$\aheap''(\astore''(x)) = (\astore''(y_1),\dots,\astore''(y_n))$, and thus we must have $k = n = m$ and for all $i = 1,\dots,n$,
$\astore''(y_i) = \hat{\astore}(u_i')$.
Let $\sigma'$ be the substitution mapping every variable $u_i'$ not occurring in $\vars(\asformB)$ to the first variable $y_j$ in $y_1,\dots,y_n$ such that $\astore''(y_j) = \hat{\astore}(u_i')$, and let $\theta$ be the substitution mapping every variable $y_i$ to $y_j$, where $j$ is the smallest index in $1,\dots,n$ such that $\astore'' \models y_i \iseq y_j$. 
Let $\hat{\sigma} = \sigma \cup \sigma'$ and $\theta' =  \hat{\sigma}\theta$; note that $\hat{\sigma}$ is well-defined since $\dom{\sigma} \cap \dom{\sigma'} = \emptyset$, as $\{ u_1',\dots,u_k' \} \cap \{ u_1,\dots,u_m \} \subseteq \vars(\asformB)$. 
By construction,
 $\theta'$ is a unifier of 
$(u_1,\dots,u_n)$ and 
$(u_1',\dots,u_n')$.
Since $\asid$ is \hdeterministic,
one of the two formulas $\apformB'$ or $\apformB''$ contains a disequation 
$v \not \iseq w$ with $v\theta' = w\theta'$ and $v,w \in \vars_{\addr}$. 
Note that $v\hat{\sigma} \not = w\hat{\sigma}$ since otherwise 
$\apformB'\sigma$ or $\apformB''\sigma'$ would be unsatisfiable.
We show that $v \not \iseq w$ occurs in $\apformB'$. Assume, for the sake of contradiction, that 
$v \not \iseq w$ occurs in $\apformB''$.
Then $v\hat{\sigma} = v\sigma'$ and $w\hat{\sigma} = w\sigma'$.
By definition of $\sigma'$ and $\theta'$, we have $z\sigma'  = z\theta'$ for all $z \in \dom{\sigma'}$, thus we get 
$v\sigma' = w\sigma'$, hence (by definition of $\sigma'$) $\hat{\astore}(v) = \hat{\astore}(w)$, which contradicts the fact that $(\hat{\astore},\aheap'') \modelsr \asheapB'' = (x \mapsto (u_1',\dots,u_k') * \asformB'') \swedge \apformB''$.
Thus necessarily 
 $v \not \iseq w$ occurs in $\apformB'$ 
 and we have $v\hat{\sigma} = v\sigma$ and $w\hat{\sigma} = w\sigma$.
As all rules are \prules, necessarily at least one of the two variables $v,w$ (say $v$) is an existential variable, and we must have $v \in \alloc{\asformB'}$.
Since $(\astore',\aheap) \modelsr \asheapB' = (x \mapsto (u_1,\dots,u_m) 
* \asformB') \swedge \apformB'$, we deduce by 
Lemma~\ref{lem:alloc} that 
$\astore'(v) \in \dom{\aheap}$, hence $\astore(v\sigma) \in \dom{\aheap}$.
By definition of $(\astore,\aheap)$, this entails that $\astore(v\sigma) \in \astore(\alloc{\asform})$, and since $\astore$ is injective we get $v\sigma \in \alloc{\asform}$. 

Assume that $w\sigma \in \alloc{\asform}$.
Then $\alloc{\asform}$ contains spatial atoms with respective roots 
$w\sigma$ and $v\sigma$. 
These atoms must be distinct since $v\sigma \not = w\sigma$.
By Lemma~\ref{lem:alloc} these locations are allocated in disjoint part of the heap $\aheap''$, hence necessarily $\astore''(w\sigma) \not = \astore''(v\sigma)$.
We deduce that $v\theta' \not = w\theta'$, which contradicts our assumption.
Thus $w\sigma \not \in \alloc{\asform}$, and since $\asform \swedge \apform \vdashr \asformB \swedge \apformB$
cannot be an \antiaxiom, necessarily $w\sigma \in \vars(\asformB)$.
We also have $w\sigma \not \in V$, since $\dom{\aheap''} \cap \astore''(V) = \emptyset$.
We may assume that  the sequent $\asform \swedge \apform \vdashr \asformB \swedge \apformB$ 
 is irreducible by the rule \dec, so that, 
 as $\asformB$ is a predicate atom, $\apform$ contains a disequation $v\sigma \not \iseq w\sigma$ (since $v\sigma \in \alloc{\asform}$, $w\sigma \in \vars(\asformB) \setminus (\alloc{\asform} \cup V)$).
But then 
we have $\astore'' \not \models \apform$, contradicting the definition of $\astore''$. 
 

}
\end{enumerate}
  \end{proof}

The completeness result follows immediately:
\begin{theorem}
\label{theo:comp}
If $\aseq$ is a valid sequent, then it admits a
\fexpanded (possibly infinite) proof tree.
\end{theorem}
\begin{proof}
The theorem follows from Lemma~\ref{lem:comp}.
\end{proof}

 \subsection{Termination and Complexity}

We now show that the proof procedure runs in polynomial time. To this purpose we provide (Definition~\ref{def:associate} and Lemma~\ref{lem:associate}) a characterization of the \ireducible sequents that may occur in a proof tree.


\begin{definition}
\label{def:associate}
Let $\anum  \in \mathbb{N}$. 
A sequent $\aseq$ is a  {\em $\anum$-\companion} of a symbolic heap 
$\asform \swedge \apform$
if $\aseq$  is of the form
$(\asform_1\sigma * \asform_2 * \asform_3) \swedge (\apform_0 \wedge \apform_1 \wedge \apform_2) \vdashr \asformB \swedge \true$, where:
\begin{enumerate}
\item{$\apform_0$ is the set of disequations occurring in $\apform$ that are of the form $x \not \iseq t$ (up to commutativity), with $x \in (\vars(\asform_1\sigma * \asform_2 * \asform_3) \cup \vars(\asformB)) \setminus \vars_\addr$ and $t \not = x$. \label{companion:diseq0}}
\item{$\apform_1 = \bigwedge_{x\in \specvar{\aseq}, y \in A} x \not \iseq y$, where $A= \alloc{\asform_1\sigma * \asform_2 * \asform_3}  \setminus \vars(\asformB)$.\label{companion:diseq1}}
\item{$\apform_2$ is a conjunction of disequations of the form $u \not \iseq v$, where $u,v \in \vars(\asformB)$.\label{companion:diseq2}}
\item{$\asformB$ is a predicate atom. \label{companion:pred}}
\item{$V \subseteq \vars(\asformB)$.\label{companion:V}}
\item{Either $\asform_1 = \emp$, or there exists a predicate atom $\anatom$ such that $\rootof{\anatom} = \rootof{\asformB}$ and 
$\anatom \unfold \asform_1$.\label{companion:unfold}}
\item{$\len{\asform_2} \leq \anum+1$.\label{companion:num}}
\item{$\sigma$ is a substitution such that $\card{\dom{\sigma}} \leq \anum$.\label{companion:sigma}}
\item{$\asform\sigma$ is of the form $\asform_3 * \asformC_1 * \asformC_2$ (modulo AC), where
$\len{\asformC_1} \leq 2\cdot\anum$ and:
\begin{enumerate}
\item{no atom $\anatomB$ occurring in $\asform_3$ is such that $\rootof{\asformB} \not \pathto{\asform_1\sigma * \asform_2 * \asform_3 * \asformC_2}^* \rootof{\anatomB}$; and}
\item{no atom $\anatomB$ occurring in $\asformC_2$ is such that $\rootof{\asformB} \pathto{\asform_1\sigma * \asform_2 * \asform_3 * \asformC_2}^* \rootof{\anatomB}$.}
\end{enumerate}

\label{companion:dec}}

\end{enumerate}
\end{definition}

Intuitively, we aim at describing sequents $\aseq$ that are \ireducible and occur in some proof tree, where the formula $\asform \swedge \apform$ is meant to denote the left-hand side  of the root sequent. It is clear that the left-hand side of $\aseq$ will contain some parts that are inherited from the formula $\asform \swedge \apform$, together with some other parts, that are introduced by the inference rules applied along the branch of $\aseq$ in the proof tree.
First, observe that Conditions \ref{companion:pred} (stating that the right-hand side of $\aseq$ is a predicate atom $\asformB$) and \ref{companion:V} (stating that the variables in $V$ must occur in $\asformB$) are easy consequences of the fact that $\aseq$ is \ireducible: rules \sep\ and \eli\ cannot be applied.
The substitution $\sigma$ is intended to encode the effect of the (cumulative) applications of rule \dec\ (in the left branch): some variables initially occurring in $\asform \swedge \apform$ are replaced by other variables. 
Note that Condition \ref{companion:sigma} bounds the size of this substitution, which is essential to avoid any exponential blow-up. As we shall see, it stems from the fact that rule \dec\ only applies, by Definition~\ref{def:strat} (\ref{strat:dec}), to a variable in $\specvar{\aseq}$ which restricts the number of applications of the rule.
The formula $\apform_0$ denotes the set of disequations between variables of a sort distinct from $\addr$, which are all inherited from $\apform$ (as the rules contain no disequations of this form).
On the other hand, $\apform_1$ denotes the set of disequations
that were introduced by previous applications of \dec\ (in the right branch). Note that since $\aseq$ is \ireducible, we may assume that all the possible applications of \dec\ have been applied, leading either to the elimination of a variable (encoded in $\sigma$) or to the addition of some disequation, this is why $\apform_1$ contains the set of {\em all} possible disequations. The formula $\apform_2$ denotes the set of disequations between variables in $\asformB$, which are arbitrary (this is not problematic as the number of such disequations is bounded, as $\asformB$ is an atom). 
The formula $\asform_1$ denotes the part of the formula added in the left-hand side of the sequent by the {\em last} application of \unf. The formula may be empty, because such application does not always exist (for instance if $\aseq$ is the root sequent). The atom $\anatom$ then denotes the predicate atom on which \unf\ was previously applied. Observe that the root of this atom is always the same as that of $\asformB$, by Definition~\ref{def:strat} (\ref{strat:unf}). Note also that the formula contains no atom obtained from any application of \unf\ occurring {\em before} the last one (i.e., on atoms other than $\anatom$): indeed, we shall prove that all such atoms must have been  eliminated by previous decomposition steps when \imi\ was applied.
The formula $\asform_1\sigma$ contains all spatial atoms inherited from the initial formula $\asform$. Note that some atoms in $\asform$ may have been deleted, by previous applications of \sep. 
The difficulty here is that  this could lead to an exponential blow-up, since in principle one could consider all possible subsets of $\asform$. 
To overcome this issue, we refine the criterion by taking into account all the additional reachability conditions that prevent $\aseq$ from being an axiom or an \antiaxiom.
We thus further split the conjunction of atoms that were deleted from $\asform$ into two parts $\asformC_1 * \asformC_2$, where $\asformC_1$ contains a {\em bounded} number of atoms that can be chosen in an {\em arbitrary} way, and $\asformC_2$ may contain an {\em arbitrary} number of atoms, but is fully determined by the choice of $\asform_1,\asform_2$ and $\asformC_1$.

The following lemma shows that a formula has only polynomially many $\anum$-{\companion}s, up to a renaming of variables.

\begin{lemma}
\label{lem:poly}
Let $\asform \swedge \apform$ be a spatial formula and
$E$ be a set of variables.
For every fixed number $\anum$, 
if $\maxar{\asid} \leq \anum$ 
(i.e., if the maximal arity of the predicate symbols in bounded by $\anum$), then
the number of sequents $\aseq$ that are $\anum$-{\companion}s of $\asform \swedge \apform$  and such that $\vars(\aseq) \subseteq E$
is polynomial w.r.t.\ $\size{\asform} + \size{\asid} + \card{E}$ (up to AC), and the size of $\aseq$ is polynomial w.r.t.\ $\size{\asform} + \size{\asid}$.
\end{lemma}
\begin{proof}
By Definition~\ref{def:associate}, $\aseq$ is of the form $(\asform_1\sigma * \asform_2 * \asform_3) \swedge (\apform_1 \wedge \apform_2) \vdashr \asformB \swedge \true$, and satisfies Conditions $1$-$9$.
Note that since $\asformB$ is a predicate atom, we have $\size{\asformB} \leq 1 + \maxar{\asid} \leq \anum+1$ and $\card{\vars(\asformB)} \leq \maxar{\asid} \leq \anum$.
 
We have $\len{\asform_1 * \asform_2 *\asform_3} = \len{\asform_1} + \len{\asform_2} + \len{\asform_3} \leq \len{\asform_1} + \anum + 1 + \len{\asform} \leq  
\size{\asid} + \size{\asform}+ \anum + 1$, 
thus
$\size{\asform_1\sigma * \asform_2 * \asform_3} \leq (\anum + 1)\cdot (\size{\asid} + \size{\asform}+ \anum + 1) = \bigO{\size{\asform}+\size{\asid}}$.
Since $\vars(\apformB_1) \cup \vars(\apformB_2)  \subseteq \vars(\asform_1\sigma * \asform_2 * \asform_3) \cup \vars(\asformB)$, 
we deduce that $\size{\apform_1 \wedge \apform_2} =  \bigO{(\size{\asform}+\size{\asid})^2}$ hence $\size{\aseq} =  \bigO{(\size{\asform}+\size{\asid})^2}$.
 
Let  $n = \card{E} + \size{\asform} + \card{\csts} + 1$. Note that $n \leq \card{E} + \size{\asform} + \size{\asid} + 1$. 
The number of possible predicate atoms $\asformB$ and $\anatom$ is 
at most 
$\card{\preds}\cdot n^\anum \leq \size{\asid} \cdot n^\anum$. Similarly, the number of formulas
$\asform_2$ is at most $\size{\asid}^{\anum+1} \cdot n^{(\anum+1)^2}$.
Once $\anatom$ is fixed, the number of formulas 
$\asform_1$ cannot be greater than the number of rules in $\asid$, up to a renaming of variables.
Since $\card{\dom{\sigma}} \leq \anum$, there are at most 
$n^{2\cdot\anum}$ 
possible substitutions $\sigma$ we only have to choose the set $\dom{\sigma}$,
i.e., at most $\anum$ elements among the variables occurring in $\asform$ or $\asform_1$, and the image of each variable in $\dom{\sigma}$ 
among the variables in $E$.
The number of formulas $\asform_3$ is at most $\len{\asform}^{2\cdot\anum}$ 
(since we only have to choose the formula $\asformC_1$ in $\asform$, as $\asformC_2$ is entirely determined by the choice of $\asform_1, \asform_2$ and $\asformC_1$). 
We have $\vars(\apform_2) \subseteq \vars(\asformB)$ thus 
$\card{\vars(\apform_2)} \leq \anum$, and there are $2^{\anum^2}$  possible
formulas $\apform_2$.
Since $V \subseteq \vars(\asformB)$, we have $\card{V} \leq \anum$ and the number of possible sets $V$ is $2^\anum$. 
Finally, the formulas $\apform_0$ and $\apform_1$ are fixed once $\asform$ and $\asform_1,\asform_2,\asform_3$ and $\sigma$ are fixed. 
\end{proof}

Lemma \ref{lem:associate} below  is used to prove that every sequent that is \ireducible and occurs in a proof tree is a \companion of the left-hand side of the root of the proof tree. Together with 
Lemma~\ref{lem:poly} and the result stated in Lemma~\ref{lem:direct}  on the number of \direct descendants of a sequent, this will entail that the size of the proof tree is polynomial w.r.t.\ the size of the root.
We need the following proposition, which is an easy consequence of the definition of the rules.

\newcommand{\sv}{V^{\star}}

\begin{proposition}
\label{prop:bothallocvar}
Let $\aseq_i = \asheap_i \vdashrV{V_i} \asheapB_i$ (for $i =1,2$) be sequents.
If $\aseq_2$ is a successor of $\aseq_1$, then 
$\alloc{\asheap_1} \cap \alloc{\asheapB_1} \cap \alloc{\asheap_2} \subseteq \alloc{\asheapB_2}$.
\end{proposition}
\begin{proof}
Assume that $x\in \alloc{\asheap_1} \cap \alloc{\asheapB_1} \cap \alloc{\asheap_2}$ and $x\not \in \alloc{\asheapB_2}$.
It is easy to check, by inspection of the rules, that the only rule that can remove the variable $x$ from 
$\alloc{\asheapB_1}$ is \sep\ (indeed, \imi\ deletes a predicate atom but introduce a points-to atom with the same root, 
and the rules \eq\ and \dec\ cannot replace $x$ as otherwise it would not occur in $\asheap_2$ ). 
By definition of the rule, $\aseq_1$ also has a successor
$\aseq_2' = \asheap_2' \vdashrV{V_2'} \asheapB_2'$ such that
$x \in V_2'$ (since $x \in \alloc{\asheap_2}$) and $x \in \alloc{\asheapB_2'}$ (since $x \not \in \alloc{\asheapB_2}$). This entails that $\aseq_2'$ is an \antiaxiom, contradicting Condition~\ref{strat:noax} in Definition~\ref{def:admissible}.
\end{proof}

\begin{lemma}
\label{lem:associate}
Assume that 
$\maxar{\asid} \leq \anumB$ for $\anumB \in \mathbb{N}$, and let $\aseq = \asform \swedge \apform \vdash \asformB \swedge \apformB$ be an \eqfree sequent.
Every \ireducible sequent occurring in a proof tree with root $\aseq$ 
is a $\anumB$-\companion of the left-hand side $\asform \swedge \apform$ of $\aseq$.
\end{lemma}
\begin{proof}  
Let $\aseq' = \asform' \swedge \apform' \vdashrV{V'} \asformB' \swedge \apformB'$ be an \ireducible sequent occurring in a proof tree with root $\aseq$. 
Let $\aseq_1,\dots,\aseq_n$, be a path from $\aseq$ to $\aseq'$, where 
$\aseq_1 = \aseq$, $\aseq_n = \aseq'$ and $\aseq_i = \asheap_i \vdashrV{V_i} \asheapB_i$ for $i=1,\ldots, n$.
We check that all the conditions in Definition~\ref{def:associate} are satisfied.
Note that by Condition~\ref{strat:noax} in Definition~\ref{def:admissible}, $\aseq'$ cannot be an axiom or an \antiaxiom.

\begin{itemize}
\item{{\bf Right-Hand Side.}
We first show that the right-hand side $\asformB'$ is indeed a predicate atom. By Condition~\ref{strat:ini} in Definition~\ref{def:admissible}, $\apformB' = \true$ and $\asformB'$ is a predicate atom, thus Condition~\ref{companion:pred} of Definition~\ref{def:associate} holds.
}
\item{{\bf Pure Formulas.}
We then prove that the pure formulas $\apform'$ fulfill Conditions \ref{companion:diseq0}, \ref{companion:diseq1} and \ref{companion:diseq2}.
If $\apform'$ contains an equality then the rule \eq\ applies, hence $\aseq'$ cannot be \ireducible because $\eq$ has priority over {\imi} (Condition~\ref{strat:priority} in Definition~\ref{def:admissible}).
Thus $\apform'$ is a conjunction of disequations, and it can be written as
$\apform_0 \swedge \apform_1 \wedge \apform_2$ where 
$\apform_0$ is the set of disequations in $\apform'$ between 
terms of a sort distinct from $\addr$, 
$\vars(\apform_2) \subseteq \vars(\asformB') \cap \vars_\addr$ 
and $\apform_1$ only contains disequations of the form $u \not \iseq v$ (up to commutativity of $\not \iseq$) with $u,v \in \vars_\addr$  and $u \not \in \vars(\asformB')$. 
Note that each disequation involves distinct terms because by hypothesis, $\aseq'$ is not an axiom.
Consider one such disequation $u\not \iseq v$ in $\apform_1$. If $u \not \in \vars(\asform')$ then rule {\wea} applies by the second application condition of the rule, which is impossible since this rule has priority over {\imi} and $\aseq'$ would not be {\ireducible}. We deduce that $u \in \vars(\asform')$, and since $\aseq'$ is not an \antiaxiom, we must have 
$u \in \alloc{\asform'}$.
If $v \not \in \vars(\asform') \cup \vars(\asformB')$ or $v \in \alloc{\asform'} \cup V'$ then rule \wea\ also applies (in the latter case, we have $\asform' \gEmodels{V_i} u \not \iseq v$ by Condition~\ref{it:emodel:neq:sub} of Definition~\ref{def:emodel}). 
Thus $v \in \vars(\asform') \cup \vars(\asformB')$ and $v \not \in \alloc{\asform'} \cup V$. This entails that 
$v  \in \vars(\asformB')$ since otherwise $\aseq'$ would be an \antiaxiom by Condition~\ref{anti:ref} of Definition~\ref{def:antiax} because $u\in \alloc{\asform'}$, hence $v$ is necessarily of sort $\addr$.
 We therefore have $v \in \specvar{\aseq'}$ and $u \in \alloc{\asform'} \setminus \vars(\asformB')$.
Conversely, if a disequation $u \not \iseq v$ with $v \in \specvar{\aseq'}$ and $u \in \alloc{\asform'} \setminus \vars(\asformB')$
does not occur in $\apform_1$ then by Condition~\ref{strat:dec} in Definition~\ref{def:admissible}, \dec\ would be applicable, which is impossible.
 Thus Conditions \ref{companion:diseq1} and \ref{companion:diseq2} in Definition~\ref{def:associate} are satisfied.
By Proposition~\ref{prop:rule-facts} (\ref{it:rf:intro:dis}) no rule introduces a disequation between terms of a sort distinct from $\addr$, all disequations in $\apform_0$ must occur in 
$\apform$. Furthermore, such a disequation cannot be of the form $t \not \iseq t$ (otherwise $\aseq'$ would be an axiom as explained above)
and must contain a variable occurring in $\asform'$ or $\asformB'$ (otherwise \wea\ would apply). Conversely, no rule can delete a disequation containing a variable in $(\vars(\asform') \cup \vars(\asformB')) \setminus \vars_\addr$ from the left-hand side of the sequent.
Thus Condition~\ref{companion:diseq0} holds.
}
\item{{\bf Non Allocated Variables.}
We now check that Condition \ref{companion:V} applies, namely that $V' \subseteq \vars(\asformB')$.
Assume, for the sake of contradiction, that $V'$ contains a variable $x \not \in \vars(\asformB')$. Since $\aseq'$ is not an \antiaxiom, we have $V' \cap (\vars(\asform') \setminus \vars(\asformB')) = \emptyset$, hence $x \not \in \vars(\asform')$. If $x \in \vars(\apform')$ then {\wea} 
applies, 
and otherwise \eli\ applies. In both cases, $\aseq'$ cannot be \ireducible (since \wea\ and \eli\ have priority over \imi), which contradicts our assumption. Thus Condition~\ref{companion:V} holds.
}
\item{{\bf Substitution $\sigma$.} We now define the substitution $\sigma$, and we show that Condition \ref{companion:sigma} holds.
Let $i$ be the least number such that the right-hand side $\asheapB_i$ of $\aseq_i$ is a spatial 
atom ($i$ must exist since $n$ fulfills this condition, by Condition~\ref{strat:ini} in Definition~\ref{def:strat}), and let $\sv = \specvar{\aseq_i}$. 
Note that $\asheapB_i$ must be a predicate atom. 
Indeed, the right-hand side of $\asheapB_n = \asheap_n \vdashrV{V_n} \asheapB_n$ is a predicate atom since $\aseq_n$ is \ireducible, and by Condition~\ref{it:pt:imi} of Proposition~\ref{prop:rule-facts}, the only rule that can introduce a predicate atom to the right-hand side of a premise is \imi, which by Condition~\ref{strat:ini} in Definition~\ref{def:strat}, applies only if the right-hand side of the conclusion is a predicate atom. 
Thus $\asheapB_j$ cannot be a points-to atom.
By Proposition~\ref{prop:specvar} we have 
$\specvar{\aseq_j} \subseteq \sv$, for all $j = i,\dots,n$, and $\card{\sv} \leq \maxar{\asid} \leq \anumB$ since $\sv \subseteq \vars(\aseq_i)$ and $\asheapB_i$ is a predicate atom.
Let $I \subseteq \{ 1,\dots,n\}$ be the set of indices $j$ such that rule $\dec$ applies on $\aseq_j$, and replaces a
variable $x_j$ by a variable $y_j$ (i.e., such that $\aseq_{j+1}$ is a left premise of an application of \dec\ on $\aseq_j$ with variables $x_j$ and $y_j$). 
By Condition~\ref{strat:dec} of Definition~\ref{def:admissible}, \dec\ applies only if the right-hand side of the conclusion is a predicate atom,
thus we must have $I \subseteq \{ i,\dots,n \}$.
Let $\sigma = \sigma_n$, where 
$\sigma_0 = \id$, $\sigma_{j+1} = \sigma_j\{ x_{j+1} \leftarrow y_{j+1} \}$ if $j+1 \in I$ and $\sigma_{j+1} = \sigma_j$ otherwise.
By Condition~\ref{strat:dec} in Definition~\ref{def:admissible}, 
for all $j\in I$, we have $y_j \in \specvar{\aseq_j} \subseteq \specvar{\aseq_i} = \sv$. Also,
$y_j \not \in \specvar{\aseq_{j+1}}$ because $x_j$ must be allocated in the left-hand side of $\aseq_j$, hence, after the replacement, $y_j$ is allocated in the left-hand side of $\aseq_{j+1}$. Thus by Proposition~\ref{prop:specvar}, $y_j \not \in \specvar{\aseq_{j'}}$, for all $j' > j$, which entails that 
the 
$y_j$ for $j \in I$ are pairwise distinct. 
Therefore, $\card{I} \leq \card{\sv} \leq \anumB$, hence $\card{\dom{\sigma}} \leq \anumB$ and Condition~\ref{companion:sigma} is satisfied. }
\item{{\bf Unfolded Atom.}
Now we define the unfolded predicate atom $\anatom$ and we prove that Condition~\ref{companion:unfold} holds.
If rule \unf\ is applied on a sequent $\aseq_k$ and $\aseq'$ is an \direct \descendant of $\aseq_k$, then we denote by $\anatom$ the predicate atom on which \unf\ is applied, and by $\asform_1 \swedge \apformC$ the formula occurring in $\aseq_{k+1}$ such that $\anatom \unfold \asform_1 \swedge \apformC$.
If no such application of \unf\ exists then we simply set $\asform_1 = \emp$.
Note that if $k$ and $\anatom$ exist then they must be unique, since 
we assumed that $\aseq'$ is an \direct \descendant of $\aseq_k$
and by Corollary \ref{cor:onlyoneunf}
 there must be an application of \imi\ along the path between any two applications of \unf.
Also, by Condition~\ref{strat:unf} in Definition~\ref{def:strat}, 
$\asheapB_k$ is a spatial 
atom and $\rootof{\anatom} = \rootof{\asheapB_k}$. By Proposition~\ref{prop:qred}, the only rules that can be applied on $\aseq_{k+1},\dots,\aseq_{n-1}$ are \wea, \dec, or \eli. Each of these rules only affects the right-hand side of its premise 
by instantiating it with $\sigma$
thus $\rootof{\anatom}\sigma = \rootof{\asformB'}$. 
First consider the case where $k$ exists. Since the rules in $\asid$ are  \prules, $\asheap_{k+1}$ must contain a points-to atom. By Condition~\ref{strat:unf} in Definition~\ref{def:admissible}, the root of this atom must be the same as that of $\asheapB_k = \asheapB_{k+1}$. Furthermore, $\asheap_{k+1}$ contains no equation, since otherwise this equation would occur in $\asheap_k$, rule $\eq$ would be applicable on $\aseq_k$ and this rule has priority over \unf. Therefore, $\aseq_{k+1}$ is \qireducible.
By Proposition~\ref{prop:qred}, this entails that only the rules \wea, \eli\ or \dec\ may be applied along the path 
$\aseq_{k+1},\dots,\aseq_n$. These rules either do not affect spatial formulas, or
uniformly replace a variable $x_j$ (for $j \in I$) by $y_j = x_j\sigma_j$ (hence eventually, every variable $x$ is replaced by $x\sigma$). 
Thus $\asform'$ 
is of the form $\asform_1\sigma * \asform_1'$. 
It is clear that the previous assertion trivially holds in the case where $k$ does not exist by letting 
$\asform_1 = \emp$. Thus Condition~\ref{companion:unfold} holds.}
\item{
{\bf Spatial Formula.}
We now prove that the spatial formula $\asform_1'$ is of the required form and fulfills Condition~\ref{companion:num}. The formula $\asform_1'$ may be written as
$\asform_2 * \asform_3 * \asform_4$, where:
\begin{itemize}
\item $\alloc{\asform_2} \subseteq (\sv \cup \alloc{\asformB'})$,
\item $\alloc{\asform_3} \subseteq \vars(\asform) \setminus (\sv \cup \alloc{\asformB'})$,
\item $\alloc{\asform_4} \cap (\vars(\asform) \cup \sv \cup \alloc{\asformB'}) = \emptyset$.
\end{itemize}
If $\asform_2$ contains two atoms with the same root then $\asform'$ is \heapunsat and $\aseq'$ is an axiom.
Thus $\len{\asform_2} \leq \card{\sv} + 1 \leq \anumB+1$ and Condition~\ref{companion:num} holds.

The only rule that can add new predicate symbols to a premise is \unf, replacing an atom
$\anatom'$ by a formula $\asheap$ such that $\anatom' \unfold \asheap$.
Since the rules in $\asid$ are  \prules, the root of every predicate atom in $\asheap$  must be a fresh variable $y$, not occurring in $\vars(\asform)$. If $y \in \dom{\sigma}$ then $y$ is eventually replaced by some variable in $\sv$, hence
the corresponding  predicate atom from $\asform_1'$ 
is in $\asform_2$.
If $y \not \in \dom{\sigma}$, then since this variable is never replaced, the corresponding atom must occur in $\asform_4$ by construction.
This entails that all the atoms in $\asform_3$ must be obtained from atoms in $\asform$ after application of the substitution $\sigma$, i.e., that
$\asform\sigma = \asform_3 * \asformC$ for some formula $\asformC$, up to AC.
The formula $\asformC$ can be written as $\asformC_1 * \asformC_2$ where 
$\alloc{\asformC_1} \subseteq \vars(\asformB') \cup \sv$ and 
$\alloc{\asformC_2} \cap (\vars(\asformB') \cup \sv) = \emptyset$. 
This entails that $\len{\asformC_1} \leq 2\cdot\anumB$, since $\card{\vars(\asformB')} \leq \anumB$ and
$\card{\sv} \leq \anumB$. 

We prove that necessarily $\asform_4 = \emp$.
Assume for the sake of contradiction that $\asform_4$ contains an atom $\anatom'$, and let $u = \rootof{\anatom'}$. By definition of $\asform_4$, $u \not \in \vars(\asform) \cup \sv \cup \alloc{\asformB'}$.
The variable $u$ must have been introduced 
by an application of  rule \unf\ on some sequent 
 $\aseq_{m}$, as \unf\ is the only rule that can introduce new variables to the left-hand side of a sequent. By Proposition~\ref{prop:unf_qred}, $\aseq_{m+1}$ is \qireducible, and we must have $m \not =k$ (if $k$ exists), since otherwise $\anatom'$ would occur in $\asform_1\sigma$.
  Since the rules in $\asid$ are  \prules, we have $u \in \varmap{\aseq_{m+1}}$.
 Since $m \not = k$, by Corollary \ref{cor:onlyoneunf}, there exists $m'$ such that
$m < m'$ 
and $\aseq_{m'}$ is \ireducible, i.e., rule \imi\ is applied on $\aseq_{m'}$. 
Let $m'$ be the minimal number satisfying this property. 
The path $\aseq_{m+1},\dots,\aseq_{m'}$ only contains 
applications of \wea, \eli\, and \dec, and $u \in \varmap{\aseq_{m'}}$ by Proposition~\ref{prop:qred}.
Note that if $k$ exists then necessarily $m' < k$, moreover,
none of these rules can introduce new variables to the right-hand side of a sequent. In particular, variable $u$, which is a fresh variable that does not occur in $\asheapB_{m}$, cannot occur on the right-hand side $\asheapB_{m'}$ of $\aseq_{m'}$.
Let $\theta$ denote the substitution used in the application of \imi\ on $\aseq_{m'}$.
As $u \in \varmap{\aseq_{m'}}$, there exists a variable $z \in \dom{\theta}$ such that $z\theta = u$.
Since the rules in $\asid$ are  \prules, this entails that $u\in \alloc{\asheapB_{m'+1}}$.
Consequently, $u \in \alloc{\asheap_{m'+1}} \cap \alloc{\asheapB_{m'+1}}$ and by (repeatedly) applying  Proposition~\ref{prop:bothallocvar}, we deduce that 
$u \in \alloc{\asformB'}$ (since  by hypothesis $u \in \alloc{\asform_4} \subseteq \alloc{\asform'}$), which contradicts the definition of $\asform_4$.}
\item{{\bf Reachability Conditions.}
Finally, we show that Condition~\ref{companion:dec} holds, i.e., that $\asformC_2$ contains exactly the atoms $\anatomB$ in $\asform_3 * \asformC_2$ such that $\rootof{\asformB'} \not \pathto{\asform_1\sigma * \asform_2 * \asform_3 * \asformC_2} \rootof{\anatomB}$.
Assume that $\asform_3$ contains an atom $\anatomB$ such that $\rootof{\asformB'} \not \pathto{\asform_1\sigma * \asform_2 * \asform_3 * \asformC_2} \rootof{\anatomB}$.
Since $\asform' = \asform_1\sigma * \asform_2 * \asform_3$, we get $\rootof{\asformB'} \not \pathto{\asform'} \rootof{\anatomB}$, $\rootof{\anatomB} \in \alloc{\asform'}$
and
$\rootof{\anatomB} \not \in \alloc{\asformB'}$ (because $\rootof{\anatomB} \not = \rootof{\asformB'}$ and $\alloc{\asformB'} = \{ \rootof{\asformB'} \}$), thus $\aseq'$ is an \antiaxiom.
Conversely, assume that $\asformC_2$ contains an atom $\anatomB$ such that $\rootof{\asformB'} \pathto{\asform_1\sigma * \asform_2 * \asform_3 * \asformC_2} \rootof{\anatomB}$.
Note that $\rootof{\anatomB} \not \in \vars(\asformB') \cup \sv$, by definition of $\asformC_2$, which entails that $\rootof{\anatomB} \not \in \codom{\sigma}$.
If $\rootof{\anatomB} \not \in \alloc{\asform'}$, then by considering the first variable $y \not \in \alloc{\asform'} = \alloc{\asform_1\sigma * \asform_2* \asform_3}$ along the path from $\rootof{\asformB'}$ 
to $\rootof{\anatomB}$, we have
$y \in \vars(\asform')$, $y \not \in \alloc{\asform'}$ and
$y \in \alloc{\asformC_2}$, hence $y \not \in \vars(\asformB')$. Thus $\aseq'$ is an \antiaxiom,  contradicting our hypothesis.
Therefore, $\rootof{\anatomB}  \in \alloc{\asform'} = \alloc{\asform_1\sigma} \cup \alloc{\asform_2} \cup \alloc{\asform_3}$. We distinguish the three cases:
\begin{itemize}
\item{Assume that $\rootof{\anatomB} \in \alloc{\asform_1\sigma}$ i.e., that $k$ exists, $\asform_1 \neq \emp$ and $\anatom \unfold \asform_1 \swedge \apformC$. Since $\rootof{\anatomB} \not \in \codom{\sigma}$, we also have $\rootof{\anatomB} \in \alloc{\asform_1}$. Since the rules in $\asid$ are  \prules, the variables in $\alloc{\asform_1}$ are either the root of $\anatom$ or fresh variables. Now $\rootof{\anatom}\sigma$ must be distinct from $\rootof{\anatomB}$, since $\rootof{\anatomB} \not \in \vars(\asformB')$ and $\rootof{\anatom}\sigma = \rootof{\asformB'}$. Furthermore, all fresh variables are  necessarily distinct from $\rootof{\anatomB}$. Indeed, they cannot occur in $\asform$, by definition, and $\rootof{\anatomB} \subseteq \alloc{\asform}$, as $\rootof{\anatomB} \subseteq \alloc{\asformC_2} \subseteq \alloc{\asformC} \subseteq \alloc{\asform\sigma}$ and $\rootof{\anatomB} \not \in \codom{\sigma}$. Thus this case cannot occur.}
\item{If $\rootof{\anatomB} \in \alloc{\asform_2}$, then $\rootof{\anatomB} \in \sv \cup \alloc{\asformB'}$ (by definition of $\asform_2$), which contradicts the fact that $\rootof{\anatomB} \not \in \vars(\asformB') \cup \sv$.}
\item{ 
If $\rootof{\anatomB} \in \alloc{\asform_3}$, then $\rootof{\anatomB}$ occurs twice in $\alloc{\asform\sigma}$ because $\asform\sigma = \asform_3 * \asformC_1 * \asformC_2$; hence, it also occurs twice in $\alloc{\asform}$ because  $\rootof{\anatomB} \not \in \codom{\sigma}$. This entails that $\asform$ is \heapunsat, a contradiction. }
\end{itemize}
Thus we get a contradiction, which entails that Condition~\ref{companion:dec} in Definition~\ref{def:admissible} is satisfied.
}
\end{itemize}
\end{proof}

\newcommand{\normalized}{normalized\xspace}
\newcommand{\normalization}{normalization\xspace}

We now introduce a restricted form of variable renaming\footnote{Note that we cannot assume that sequents are taken modulo variable renaming, since identity would then be difficult to test (it would be equivalent to the well-known graph isomorphism problem).}. The definition is useful to straightforwardly detect cycles in a proof tree, which is essential for efficiency.
For each sort $\asort$, we consider infinite (fixed) sequences of pairwise distinct variables  of sort $\asort$: $x_i^\asort$, with $i \in \mathbb{N}$ (not occurring in the root sequent). A sequent $\aseq'$ occurring in a proof tree with root $\aseq$ is {\em \normalized} if, for every variable $y$ of sort $\asort$ occurring in $\vars(\aseq') \setminus \vars(\aseq)$, there exists $i \leq \size{\aseq'}$ such that
$y = x_i^\asort$. It is clear that every sequent $\aseq'$ can be reduced in polynomial time to a \normalized sequent. 
Indeed, it suffices to rename all the variables $y_1,\dots,y_n$ of sorts $\asort_1,\dots,\asort_n$ occurring in $\aseq'$ but not in $\aseq$ by the variables 
$x_1^{\asort_1},\dots,x_{n}^{\asort_n}$, respectively. As $n$ cannot be greater than the size of $\aseq'$, the obtained sequent is \normalized.
Such an operation is called a {\em \normalization}.

\begin{lemma}
\label{lem:polydes}
Assume that 
$\maxar{\asid} \leq \anumB$, for some fixed $\anumB \in :\mathbb{N}$.
For every sequent $\aseq$, the number of sequents occurring in a proot tree with root $\aseq$  is polynomial w.r.t.\ $\size{\aseq} + \size{\asid}$, up to \normalization.
\end{lemma}
\begin{proof}
By definition, every \descendant of $\aseq$ is either an \direct \descendant of $\aseq$
or an \direct \descendant of some \ireducible sequent $\aseq'$.
By Condition~\ref{strat:priority} in Definition~\ref{def:strat}, \eq\ applies with the highest priority on $\aseq$, eventually yielding a (unique) sequent \eqfree $\aseq''$ of the form $\asform \swedge \apform \vdashr \asformB \swedge \apformB$, with $\size{\aseq''} \leq \size{\aseq}$ and $\vars(\asform) \subseteq \vars(\aseq)$.
By Lemma~\ref{lem:associate}, every \direct sequent $\aseq'$ is a \companion of 
of $\asform\swedge \apform$.
By Lemma~\ref{lem:poly}, the number of such sequents $\aseq'$, after \normalization (w.r.t.\ $\aseq''$), is polynomial w.r.t.\ $\size{\asform} + \size{\asid}$ (hence w.r.t.\ $\size{\aseq} + \size{\asid}$).
Indeed, it is clear that the \normalized form of $\aseq'$ is also a \companion of 
$\asform\swedge \apform$ and by definition it only contains variables in $\vars(\aseq'') \cup \{ x_i^\asort \mid i \leq \size{\aseq'} \}$. Furthermore, $\size{\aseq'}$ is polynomial w.r.t.\ $\size{\aseq} + \size{\asid}$, thus the same property holds for $\card{\vars(\aseq')}$.
Then the proof follows immediately from Lemma~\ref{lem:direct}.
\end{proof}

Putting everything together, we derive the main result of the paper:

\begin{theorem}
The validity problem for sequents $\asheap \vdashr \asheapB$, where $\asid$ is a \hdeterministic set of \prulesname with 
$\maxar{\asid} \leq \anumB$ 
and 
$\maxk{\asid} \leq \anumB$, 
for some fixed number $\anumB$,
 is in \PTIME.
 \end{theorem}
 \begin{proof}
 By Theorems \ref{theo:sound} and \ref{theo:comp}, $\asheap \vdashr \asheapB$ is valid iff 
 it admits a \fexpanded proof tree.
 We first apply inductively the inference rules in all possible ways, starting with the initial sequent $\asheap \vdashr \asheapB$ and proceeding with all the {\descendant}s of $\asheap \vdashr \asheapB$, until we obtain either an axiom or a sequent that has already been considered, up to \normalization. 
  By Lemma~\ref{lem:polydes}, the number of such {\descendant}s (up to \normalization) is polynomial w.r.t.\ $\size{\aseq} + \size{\asid}$, thus it suffices to show that each inference step can be performed in polynomial time. This can be  established by an inspection of the rules. We first observe that the set of all possible rule applications can be computed in polynomial time w.r.t.\ the size of the conclusion (in particular, both the number of premises and their size are always polynomial w.r.t.\ the size of the conclusion): this is straightforward to check for all rules except for \sep, for which the result stems from Proposition~\ref{prop:sepbranching}.
Also, it is straightforward to verify that the application conditions of the rules can be tested  in polynomial time. 
 
Afterwards, it suffices to compute the set of provable sequents 
among those that are reachable from the initial one and check whether $\asheap \vdashr \asheapB$  occurs in this set.
Since infinite proof trees are allowed, we actually compute the complement of this set, i.e., the set of sequents that are {\em not} provable. This can  be done  in polynomial time, using a straightforward fixpoint algorithm, by computing the least  set of sequents $\aseq$ such that, for all rule applications with conclusion $\aseq$, at least one of the premises is not provable. 
More precisely, a sequent is not provable if it has no successor (no rule can be applied on it), or more generally, if for all possible rule applications, at least one of the premises is not provable.

 Note that, in order to apply rule \sep\ on a sequent that is not \narrow, it is necessary to check that the left premise is valid, and this must be done by recursively applying the proof procedure. This is feasible because the left premise is always \narrow (since its right-hand side is a predicate atom). 
Consequently, the computation of valid sequents must be performed in two steps, first for \narrow sequents, then for those that are not \narrow, as the application of the rules on non \narrow sequents depends on the validity of \narrow sequents. 
By Proposition~\ref{prop:narrow} all the {\descendant}s of \narrow sequents are also \narrow, hence the converse does not hold, and the procedure is well-founded.
 \end{proof}
 
 \section{Conclusion}
 
 \label{sect:conc}

 A proof procedure has been devised for checking entailments between Separation Logic formulas
with spatial predicates denoting recursive data structures, for a new class of inductive rules, called  \hdeterministic. The soundness and completeness of the calculus have been established. 
 The conditions imposed on the data structures are stronger than those considered in \cite{IosifRogalewiczSimacek13,DBLP:conf/atva/IosifRV14} but on the other hand 
the procedure terminates in polynomial time, provided the arity of the predicates is bounded. 
Both the considered fragment and the proof procedure are very different from previously known polynomial-time decision procedures for testing entailments \cite{cook-haase-ouaknine-parkinson-worell11,DBLP:conf/concur/ChenSW17}. In particular, the considered formulas may contain several non-built-in and non-compositional predicates (in the sense of \cite{DBLP:journals/fmsd/EneaLSV17}), and the structures may contain unallocated nodes, denoting data (with no predicate other than equality).
In addition, several lower bounds have been established for the entailment problem, showing that any relaxing of the proposed conditions makes the problem untractable.

Future work  includes the implementation of the procedure and the evaluation of its practical performance.
The combination with theory reasoning (to reason on data stored inside the structures)  will also be considered. 
In particular, it would be interesting to define a fragment that encompasses both 
the \hdeterministic\ {\prule}s considered in the present paper and the  SHLIDe fragment defined in \cite{DBLP:conf/fossacs/LeL23}, as the conditions defining both classes of rules are very different and non comparable, 
as evidenced by the examples and explanations given in the introduction of the paper.


\end{document}